\documentclass[11pt,a4paper,reqno]{amsart}

\pdfoutput=1

\usepackage[margin=1in]{geometry}
\usepackage{amsmath,enumerate,amsthm,graphicx,color,verbatim}
\usepackage{amsfonts,amscd}
\usepackage{latexsym}
\usepackage{amssymb,amsbsy}
\usepackage{amsrefs}
\usepackage{graphicx,xcolor}
\usepackage{mathrsfs,enumitem,dirtytalk,hyperref,hhline}

\newtheorem{theorem}{Theorem}[section]
\newtheorem{proposition}[theorem]{Proposition}
\newtheorem{corollary}[theorem]{Corollary}

\theoremstyle{definition}
\newtheorem{definition}[theorem]{Definition}

\newtheorem{example}[theorem]{Example}

 \numberwithin{equation}{section}

\newcommand{\di}{\displaystyle}

\newcommand{\ip}[2]{\left\langle{#1},{#2}\right\rangle}
\newcommand{\norm}[1]{\left\lVert{#1}\right\rVert}

\renewcommand{\Re}[1]{\mathrm{Re}\left({#1}\right)}

\newcommand{\diag}[1]{\mathrm{diag}\left(#1\right)}

\definecolor{orange}{rgb}{1,0.5,0}

\begin{document}
	
\title[Rotated CMV Operators and OPUCs]{On Rotated CMV Operators and Orthogonal Polynomials On The Unit Circle}
\author{Ryan C.H. Ang}
\address{Department of Mathematics, Xiamen University Malaysia\\ Jalan Sunsuria, Bandar Sunsuria\\ 43900 Sepang, Selangor Darul Ehsan\\ Malaysia.}
\email{ryanangch@gmail.com}

\begin{abstract}
	Split-step quantum walk operators can be expressed as a generalized version of CMV operators with complex transmission coefficients, which we call rotated CMV operators. Following the idea of Cantero, Moral and Velazquez's original construction of the original CMV operators from the theory of orthogonal polynomials on the unit circle (OPUC), we show that rotated CMV operators can be constructed similarly via a rotated version of OPUCs with respect to the same measure, and admit an analogous $\mathcal{LM}$-factorisation as the original CMV operators. We also develop the rotated second kind polynomials corresponding to the rotated OPUCs. We then use the $\mathcal{LM}$-factorisation of rotated alternate CMV operators to compute the Gesztesy-Zinchenko transfer matrices for rotated CMV operators.
\end{abstract}

\maketitle

\tableofcontents


\section{Introduction}

\subsection{Background of CMV Operators}

Let $\mathbb{D}$, $\overline{\mathbb{D}}$ and $\partial\mathbb{D}$, respectively, be the open unit disk, closed unit disk and the unit circle in the complex plane:
\begin{align*}
	\mathbb{D}&=\{z\in\mathbb{C}:|z|<1\},\\
	\overline{\mathbb{D}}&=\{z\in\mathbb{C}:|z|\leq 1\},\\
	\partial\mathbb{D}&=\{z\in\mathbb{C}:|z|=1\}.
\end{align*}
A \emph{(standard, one-sided) CMV operator}, named after Maria J. Cantero, Leandro Moral and Luis Vel\'{a}zquez \cite{Cantero2003}, is a unitary operator $\mathcal{C}$ on $\ell^2(\mathbb{N}_0)$ whose matrix representation is of the form
\begin{equation}\label{eq_cmv_matrix_onesided_intro}
	\mathcal{C}\left(\{\alpha_n,\rho_n\}_{n=0}^\infty\right)=\mathcal{C}\left(\{\alpha_n\}_{n=0}^\infty\right)=\begin{bmatrix}
		\overline{\alpha_0} & \overline{\alpha_1}\rho_0 & \rho_0\rho_1 & 0 & 0 & \cdots\\
		\rho_{0} & -\overline{\alpha_1}\alpha_0 & -\alpha_0\rho_1 & 0 & 0 & \cdots \\
		0 & \overline{\alpha_2}\rho_1 & -\overline{\alpha_2}\alpha_1 & \overline{\alpha_3}\rho_2 & \rho_3\rho_2 & \cdots\\
		0 & \rho_2\rho_1 & -\alpha_1\rho_2 & -\overline{\alpha_3}\alpha_2 & -\alpha_2\rho_3 & \cdots\\
		0 & 0 & 0 & \overline{\alpha_4}\rho_3 & -\overline{\alpha_4}\alpha_3 & \cdots \\
		\vdots & \vdots & \vdots & \vdots & \vdots & \ddots
	\end{bmatrix},
\end{equation}
where $\{\alpha_n\}_{n\in\mathbb{Z}}\subseteq\overline{\mathbb{D}}$ and $\{\rho_n\}_{n\in\mathbb{Z}}\subseteq[0,1]\subseteq\mathbb{R}$ are given by
\begin{equation} \label{eq_rho_real}
	\rho_n=\left(1-|\alpha_n|^2\right)^{\frac{1}{2}}.
\end{equation} 
Notable features of CMV operators on $\ell^2(\mathbb{N}_0)$ include:
\begin{enumerate}
	\item The matrix representation of a CMV operator is five-diagonal with finite rows and columns. In particular, denoting all possible nonzero entries by $*$, CMV matrices on $\ell^2(\mathbb{N}_0)$ have a block structure
	\begin{equation} \label{eq_cmv_block_onesided}
		\begin{bmatrix}
			* & * & * \\
			* & * & * \\
			& * & * & * & * & \\
			& * & * & * & * & \\
			& & & * & * & * & * & \\
			& & & * & * & * & * & \\
			& & & & & & & \ddots
		\end{bmatrix}
	\end{equation}
	consisting of a single $2\times 3$ block on the top left and the rest $2\times 4$ blocks along the diagonal.
	
	\item Any unitary operator on $\ell^2(\mathbb{N}_0)$ with the above five-diagonal structure is unitarily equivalent to a CMV matrix via conjugation by the unitary diagonal matrix
	\begin{equation}
		D=\diag{1,e^{is_1},e^{is_2},\cdots}
	\end{equation}
	for some $\{s_j\}_{j=1}^\infty\subseteq(-\pi,\pi]$ (see \cite{Cantero2005,Cantero2012}).
	
	\item CMV operators give a matrix representation of multiplication by $z$ in terms of orthogonal polynomials on the unit circle (OPUC) in $L^2(\partial\mathbb{D},\mu)$, where $\mu$ is a nontrivial probability measure on $\partial\mathbb{D}$. In particular, entries of the CMV matrix are precisely given by the $L^2(\partial\mathbb{D},\mu)$-inner product
	\begin{equation} \label{eq_cmv_matrix_def}
		\mathcal{C}_{k\ell}=\ip{\chi_k}{z\chi_\ell}.
	\end{equation}
 	The sequence $\{\chi_n\}_{n=0}^\infty$ is called the \textit{CMV basis}, which is obtained by orthonormalizing the ordered basis $\{1,z,z^{-1},z^2,z^{-2},\cdots\}$ of Laurent polynomials in $\partial\mathbb{D}$. They are related to the OPUCs by
	\begin{align}
		\label{eq_cmv_basis_opuc_odd}
		\chi_{2n-1}(z)&=z^{-n+1}\varphi_{2n-1}(z),\\
		\label{eq_cmv_basis_opuc_even}
		\chi_{2n}(z)&=z^{-n}\varphi_{2n}^*(z),
	\end{align}
	where $\{\varphi_n\}_{n=0}^\infty$ are the orthonormal polynomials obtained by orthonormalizing the linearly independent set $\{1,z,z^2,\cdots\}$ in $L^2(\partial\mathbb{D},\mu)$, and the \textit{reverse orthonormal polynomials} $\{\varphi_n^*\}_{n=0}^\infty$ are defined by
	\begin{align}
		\varphi_n^*(z)=z^n\overline{\varphi_n\left(\dfrac{1}{\overline{z}}\right)}.
	\end{align}
	
	\item CMV operators admit a $\mathcal{L}\mathcal{M}$-factorisation
	\begin{equation}
		\mathcal{C}=\mathcal{L}\mathcal{M}
	\end{equation}
	with
	\begin{gather}
		\mathcal{L}=\begin{bmatrix}
			\Theta_0 & & & 0 \\ & \Theta_2 & & \\ & & \Theta_4 & \\ 0 & & & \ddots
		\end{bmatrix},\quad 
		\mathcal{M}=\begin{bmatrix}
			1 & & & 0 \\ & \Theta_1 & & \\ & & \Theta_3 & \\ 0 & & & \ddots
		\end{bmatrix},\quad
		\Theta_n=\begin{bmatrix}
			\overline{\alpha_n} & \rho_n \\ \rho_n & -\alpha_n
		\end{bmatrix},
	\end{gather}
	where $1$ denotes the $1\times 1$ block with entry $1$.
\end{enumerate}

Instead of orthonormalizing the ordered basis $\{1,z,z^{-1},z^2,z^{-2},\cdots\}$, one could also orthonormalize the ordered basis $\{1,z^{-1},z,z^{-2},z^2,\cdots\}$ to obtain the \textit{alternate CMV basis} $\{\tilde{\chi}_n\}_{n=0}^\infty$, which are related to the OPUCs by
\begin{align}
	\label{eq_cmv_alt_basis_opuc_odd}
	\tilde{\chi}_{2n-1}(z)&=z^{-n}\varphi_{2n-1}^*(z),\\
	\label{eq_cmv_alt_basis_opuc_even}
	\tilde{\chi}_{2n}(z)&=z^{-n}\varphi_{2n}(z).
\end{align}
The \textit{alternate CMV matrix} on $\ell^2(\mathbb{N}_0)$ is then defined by
\begin{equation} \label{eq_cmv_alt_matrix_def}
	\tilde{\mathcal{C}}_{k\ell}=\ip{\tilde{\chi}_k}{z\tilde{\chi}_\ell}.
\end{equation}
One can easily verify that
\begin{equation}\label{eq_cmv_alt_matrix_onesided}
	\tilde{\mathcal{C}}=\mathcal{C}^T=\mathcal{M}\mathcal{L},
\end{equation}
where $\mathcal{C}^T$ is the transpose of $\mathcal{C}$.

Barry Simon \cite{Simon2005a} then used the $\mathcal{LM}$-factorisation of $\mathcal{C}$ to define the CMV operator $\mathcal{E}$ on $\ell^2(\mathbb{Z})$ by
\begin{equation}
	\mathcal{E}=\tilde{\mathcal{L}}\tilde{\mathcal{M}},
\end{equation}
where
\begin{gather}
	\tilde{\mathcal{L}}=\begin{bmatrix}
		\ddots & & & 0 & \\ & \Theta_0 & & &  \\ & & \Theta_2 & & \\  & & & \Theta_4 & \\ & 0 & & & \ddots
	\end{bmatrix},\quad 
	\tilde{\mathcal{M}}=\begin{bmatrix}
		\ddots & & & 0 & \\ & \Theta_1 & & &  \\ & & \Theta_3 & & \\  & & & \Theta_5 & \\ & 0 & & & \ddots
	\end{bmatrix},
\end{gather}
which gives $\mathcal{E}$ a matrix representation of the form
\begin{equation}
	\mathcal{E}\left(\{\alpha_n,\rho_n\}_{n=0}^\infty\right)=\begin{bmatrix}
		\ddots & \vdots & \vdots & \vdots & \vdots & \vdots & \vdots & \vdots\\
		\hdots & -\overline{\alpha_{-1}}\alpha_{-2} & -\alpha_{-2}\rho_{-1} & 0 & 0 & 0 & 0 & \hdots\\
		\hdots & \overline{\alpha_0}\rho_{-1} & \boxed{-\overline{\alpha_0}\alpha_{-1}} & \overline{\alpha_1}\rho_0 & \rho_0\rho_1 & 0 & 0 & \hdots\\
		\hdots & \rho_0\rho_{-1} & -\alpha_{-1}\rho_{0} & -\overline{\alpha_1}\alpha_0 & -\alpha_0\rho_1 & 0 & 0 & \hdots \\
		\hdots & 0 & 0 & \overline{\alpha_2}\rho_1 & -\overline{\alpha_2}\alpha_1 & \overline{\alpha_3}\rho_2 & \rho_3\rho_2 & \hdots\\
		\hdots & 0 & 0 & \rho_2\rho_1 & -\alpha_1\rho_2 & -\overline{\alpha_3}\alpha_2 & -\alpha_2\rho_3 & \hdots\\
		\hdots & 0 & 0 & 0 & 0 & \overline{\alpha_4}\rho_3 & -\overline{\alpha_4}\alpha_3 & \hdots \\
		\vdots & \vdots & \vdots & \vdots & \vdots & \vdots & \vdots & \ddots
	\end{bmatrix},
\end{equation}
where the $(0,0)^{th}$ entry is boxed. So $\mathcal{E}$ has the structure of $2\times 4$ blocks along the diagonal extending in both directions:
\begin{equation}
	\begin{bmatrix}
		\ddots \\
		& * & * & * & * \\
		& * & * & * & * \\
		& & & * & * & * & * & \\
		& & & * & * & * & * & \\
		& & & & & * & * & * & * & \\
		& & & & & * & * & * & * & \\
		& & & & & & & & & \ddots
	\end{bmatrix}.
\end{equation}
One can define the extended alternate CMV matrix $\tilde{\mathcal{E}}$ similarly by
\begin{equation}
	\tilde{\mathcal{E}}=\tilde{\mathcal{M}}\tilde{\mathcal{L}}.
\end{equation}

CMV theory originated from the problem of finding matrix representations of multiplication by $z$ (up to unitary equivalence) in $L^2(\partial\mathbb{D},\mu)$. A natural choice of an orthonormal basis of $L^2(\partial\mathbb{D},\mu)$ would be the orthonormal polynomials $\{\varphi_n\}_{n=0}^\infty$, which gives the GGT representation after Ya. L. Geronimus \cite{Geronimus1944}, William B. Gragg \cite{Gragg1993} and Alexander V. Teplyaev \cite{Teplyaev1992}. However the GGT representation has the following limitations:
\begin{enumerate}
	\item The orthonormal polynomials $\{\varphi_n\}_{n=0}^\infty$ form a basis of $L^2(\partial\mathbb{D},\mu)$ (and hence the GGT matrix is unitary) if and only if the Verblunsky coefficients $\{\alpha_n\}_{n=0}^\infty$ satisfy $$\sum_{n=0}^\infty |\alpha_n|^2=\infty;$$
	
	\item The rows of the GGT matrix may be infinite;
	
	\item There is no GGT representation in the $\ell^2(\mathbb{Z})$ case if $\log\mu'\notin L^1(\partial\mathbb{D},d\theta/2\pi)$.
\end{enumerate}
Cantero, Moral and Vel\'{a}zquez then subsequently found the correct basis which produced a five-diagonal matrix representation \cite{Cantero2003} overcoming the above issues. More details of the GGT and CMV representations can be read in Section 4.1 and 4.2 of \cite{Simon2005a}, respectively.

One can frequently discover analogues from the theories of orthogonal polynomials on the unit circle (OPUC) and orthogonal polynomials on the real line (OPRL). A canonical example is Favard's theorem from OPRL and Verblunsky's theorem \cite{Verblunsky1935} from OPUC, which establishes bijections between measures $\mu$ supported on $\mathbb{R}$ (resp. $\partial\mathbb{D}$) and the recurrence relation coefficients $\{a_n,b_n\}_{n=0}^\infty$ for the OPRLs (resp. $\{\alpha_n\}_{n=0}^\infty$ for the OPUCs). Another pair of analogues would be Jacobi matrices from OPRL and CMV matrices from OPUC, which are central in the study of self-adjoint and unitary operators, respectively. Studies in perturbation theory are interested in how changes in the recurrence relation coefficients affect changes in the measures. In the OPRL case, the measures $\mu$ are simply the spectral measures of the Jacobi matrix generated by $\{a_n,b_n\}_{n=0}^\infty$, which is unitarily equivalent to multiplication by $z$ in $L^2(\mathbb{R},\mu)$. CMV matrices are the OPUC analogue of Jacobi matrices - that is, $\mu$ can be realized as the spectral measure of the CMV matrices generated by $\{\alpha_n\}_{n=0}^\infty$, which is also unitarily equivalent to multiplication by $z$ in $L^2(\partial\mathbb{D},\mu)$.

\subsection{Motivation, Content and Organization}

Since quantum walks are always modelled as unitary operators on $\ell^2$, it is unsurprising that CMV operators are often relevant in studies of quantum walks, whose connection was first discovered by Maria J. Cantero, F. Alberto Grünbaum, Leandro Moral and Luis Vel\'{a}zquez \cite{Cantero2012}, known as the CGMV connection. 

Recent works \cite{Cedzich2023a,Yang2022,Cedzich2023} involve the study of quantum walk operators $W:\ell^2(\mathbb{Z})\otimes\mathbb{C}^2\to \ell^2(\mathbb{Z})\otimes\mathbb{C}^2$ of the form
\begin{equation}
	W=SQ
\end{equation}
for some unitary operators $S$ and $Q$. Such quantum walk operators $W$ are called split-step quantum walks, which are used to study magnetic quantum walks \cite{Cedzich2023a,Yang2022}. Using the fact that $\ell^2(\mathbb{Z})\otimes\mathbb{C}^2\simeq\ell^2(\mathbb{Z},\mathbb{C}^2)$, one can show that $W$ can be expressed as an operator $\mathcal{E}^\angle$ on $\ell^2(\mathbb{Z})$ (see \cite[Appendix A]{Cedzich2023a} for details) with matrix representation
\begin{equation}
	\mathcal{E}^\angle\left(\{\alpha_n,\rho_n\}_{n=0}^\infty\right)=\begin{bmatrix}
		\ddots & \vdots & \vdots & \vdots & \vdots & \vdots & \vdots & \vdots\\
		\hdots & -\overline{\alpha_{-1}}\alpha_{-2} & -\alpha_{-2}\rho_{-1} & 0 & 0 & 0 & 0 & \hdots\\
		\hdots & \overline{\alpha_0\rho_{-1}} & \boxed{-\overline{\alpha_0}\alpha_{-1}} & \overline{\alpha_1}\rho_0 & \rho_0\rho_1 & 0 & 0 & \hdots\\
		\hdots & \overline{\rho_0\rho_{-1}} & -\alpha_{-1}\overline{\rho_0} & -\overline{\alpha_1}\alpha_0 & -\alpha_0\rho_1 & 0 & 0 & \hdots \\
		\hdots & 0 & 0 & \overline{\alpha_2\rho_1} & -\overline{\alpha_2}\alpha_1 & \overline{\alpha_3}\rho_2 & \rho_3\rho_2 & \hdots\\
		\hdots & 0 & 0 & \overline{\rho_2\rho_1} & -\alpha_1\overline{\rho_2} & -\overline{\alpha_3}\alpha_2 & -\alpha_2\rho_3 & \hdots\\
		\hdots & 0 & 0 & 0 & 0 & \overline{\alpha_4\rho_3} & -\overline{\alpha_4}\alpha_3 & \hdots \\
		\vdots & \vdots & \vdots & \vdots & \vdots & \vdots & \vdots & \ddots
	\end{bmatrix}.
\end{equation}
The $(0,0)^{th}$ entry is boxed as before. One immediately notices that $\mathcal{E}^\angle$ looks exactly like $\mathcal{E}$, where instead of $\{\rho_n\}_{n\in\mathbb{Z}}\subseteq[0,1]\subseteq\mathbb{R}$, one now has $\{\rho_n\}_{n\in\mathbb{Z}}\subseteq\overline{\mathbb{D}}\subseteq\mathbb{C}$ satisfying
\begin{equation}
	|\rho_n|^2+|\alpha_n|^2=1.
\end{equation}
Hence $\mathcal{E}^\angle$ is a generalized version of the (extended) CMV matrix $\mathcal{E}$, which we shall call (extended) \textit{rotated CMV matrices}, for reasons that will become obvious when we construct these matrices later. 

While property (2) on page 2 implies that every rotated CMV matrix is unitarily conjugate to a standard CMV matrix by a diagonal of the form
\begin{equation}
	\tilde{D}= \diag{\cdots,e^{is_{-3}},e^{is_{-2}},e^{is_{-1}},\boxed{1},e^{is_1},e^{is_2},e^{is_3},\cdots},
\end{equation}
Cedzich, Fillman, Li, Ong and Zhou proved that a stronger result is obtained: this unitary conjugation takes the rotated CMV matrix $\mathcal{E}^\angle\left(\{\alpha_n,\rho_n\}_{n=0}^\infty\right)$ into the standard CMV matrix $\mathcal{E}\left(\{\alpha_n,|\rho_n|\}_{n=0}^\infty\right)$, thus the Verblunsky coefficients $\{\alpha_n\}_{n=0}^\infty$ are invariant under this conjugation \cite[Corollary 2.3]{Cedzich2023}.

One may now ask: how are the $\ell^2(\mathbb{N}_0)$ versions of rotated CMV operators linked to OPUC theory? Moreover, how similar would the results be in comparison to the non-rotated case?

Clearly the $\ell^2(\mathbb{N}_0)$ version of $\mathcal{E}^\angle$ also satisfies properties (1) and (2) from page 2. In this paper, we focus on deducing analogues of properties (3) and (4) from page 2. Indeed by following a similar procedure presented in \cite{Simon2005a}, we show that rotated CMV matrices can be generated by a rotated sequence of orthonormal polynomials $\{\varphi_n^\angle\}_{n=0}^\infty=\{c_n\varphi_n\}_{n=0}^\infty$ on the unit circle with $\{c_n\}_{n=0}^\infty\subseteq\partial\mathbb{D}$, which gives an analogue of (3). Moreover, we show that rotated CMV matrices admit a rotated $\mathcal{L}\mathcal{M}$-decomposition
\begin{equation}
	\mathcal{C}^\angle=\mathcal{L}^\angle\mathcal{M}^\angle,
\end{equation}
where
\begin{gather}
	\mathcal{L}^\angle=\begin{bmatrix}
		\Theta_0^\angle & & & 0 \\ & \Theta_2^\angle & & \\ & & \Theta_4^\angle & \\ 0 & & & \ddots
	\end{bmatrix},\quad 
	\mathcal{M}^\angle=\begin{bmatrix}
		1 & & & 0 \\ & \Theta_1^\angle & & \\ & & \Theta_3^\angle & \\ 0 & & & \ddots
	\end{bmatrix},\quad
	\Theta_n^\angle=\begin{bmatrix}
		\overline{\alpha_n} & \rho_n \\ \overline{\rho_n} & -\alpha_n
	\end{bmatrix},
\end{gather}
thus giving an analogue of (4). Then following Simon's idea in \cite{Simon2005b}, we define the rotated CMV operator on $\ell^2(\mathbb{Z})$ similarly by
\begin{equation} \label{eq_rcmv_extended_def}
	\mathcal{E}^\angle=\tilde{\mathcal{L}}^\angle\tilde{\mathcal{M}}^\angle,
\end{equation}
where
\begin{gather}
	\tilde{\mathcal{L}}^\angle=\begin{bmatrix}
		\ddots & & & 0 & \\ & \Theta_0^\angle & & &  \\ & & \Theta_2^\angle & & \\  & & & \Theta_4^\angle & \\ & 0 & & & \ddots
	\end{bmatrix},\quad 
	\tilde{\mathcal{M}}^\angle=\begin{bmatrix}
		\ddots & & & 0 & \\ & \Theta_1^\angle & & &  \\ & & \Theta_3^\angle & & \\  & & & \Theta_5^\angle & \\ & 0 & & & \ddots
	\end{bmatrix}.
\end{gather}
We also develop the rotated second kind polynomials $\{\psi_n\}_{n=0}^\infty$ and show that for any fixed $z\in\partial\mathbb{D}$, the sequence
\begin{equation}
	\left\{\begin{bmatrix}
		\psi_n^\angle(z,\mu) \\ -z^n\overline{\psi_n^{\angle,*}(z,\mu)}
	\end{bmatrix}+r\begin{bmatrix}
		\varphi_n^\angle(z,\mu) \\ z^n\overline{\varphi_n^{\angle,*}(z,\mu)}
	\end{bmatrix}\right\}_{n=0}^\infty\in\ell^2(\mathbb{N}_0)^2
\end{equation}
if and only if $r=F(z)$, where $F$ is the Carath\'{e}odory function of $\mu$.

Studies of quantum walks often involve studying the spectrum of the quantum walk operator, which determines the limiting behaviour of particles in the walk. In particular, a useful element from the spectral theory of CMV operators is the Gesztesy-Zinchenko transfer matrix $T_n$ associated to the extended CMV operator $\mathcal{E}$, given by
\begin{equation} \label{eq_Tn_GZ_rCMV_intro}
	T_n(z)=\begin{cases}
		\dfrac{1}{\rho_n}\begin{bmatrix}
			-\overline{\alpha_n} & z \\ z^{-1} & -\alpha_n
		\end{bmatrix} & n\text{ even};\\
		~\\
		\dfrac{1}{\rho_n}\begin{bmatrix}
			-\alpha_n & 1 \\ 1 & -\overline{\alpha_n}
		\end{bmatrix} & n\text{ odd}.
	\end{cases}
\end{equation}
The Gesztesy-Zinchenko transfer matrix was introduced by Fritz Gesztesy and Maxim Zinchenko in their study of Weyl-Titchmarsh theory for CMV operators \cite{Gesztesy2006}. To explain the role of the Gesztesy-Zinchenko transfer matrix in the spectral theory of CMV operators, we first quote the following result from \cite{Gesztesy2006}:
\begin{theorem} \cite[Lemma 2.2]{Gesztesy2006} \label{thm_cmv_extended_GZ_lem2.2_intro}
	Let $z\in\mathbb{C}-\{0\}$ and $f=\{f_n(z)\}_{n\in\mathbb{Z}}$, $g=\{g_n(z)\}_{n\in\mathbb{Z}}$ be sequences of complex functions. Define the unitary operator $\mathcal{U}$ on $\ell^2(\mathbb{Z})^2$ by
	\begin{equation*} \label{eq_GZ_transfer_matrix_intro}
		\mathcal{U}=\begin{bmatrix}
			\mathcal{E} & 0 \\ 0 & \tilde{\mathcal{E}}
		\end{bmatrix}
		=\begin{bmatrix}
			\tilde{\mathcal{L}}\tilde{\mathcal{M}} & 0 \\ 
			0 & \tilde{\mathcal{M}}\tilde{\mathcal{L}}
		\end{bmatrix}
		=\begin{bmatrix}
			0 & \tilde{\mathcal{L}} \\ \tilde{\mathcal{M}} & 0
		\end{bmatrix}^2.
	\end{equation*}
	Then the following statements are equivalent:
	\begin{enumerate}
		\item $(\mathcal{E} f)(z)=zf(z)$ and $(\tilde{\mathcal{M}} f)(z)=zg(z)$;
		
		\item $(\mathcal{E}^T g)(z)=zg(z)$ and $(\tilde{\mathcal{L}} g)(z)=f(z)$;
		
		\item $(\tilde{\mathcal{M}} f)(z)=zg(z)$ and $(\tilde{\mathcal{L}} g)(z)=f(z)$;
		
		\item $\mathcal{U}\begin{bmatrix}
			f(z) \\ g(z)
		\end{bmatrix}=z\begin{bmatrix}
			f(z) \\ g(z)
		\end{bmatrix}$ and $(\tilde{\mathcal{M}}f)(z)=zg(z)$;
		
		\item $\mathcal{U}\begin{bmatrix}
			f(z) \\ g(z)
		\end{bmatrix}=z\begin{bmatrix}
			f(z) \\ g(z)
		\end{bmatrix}$ and $(\tilde{\mathcal{L}}g)(z)=f(z)$;
		
		\item for each $n\in\mathbb{Z}$,
		\begin{equation*} 
			\begin{bmatrix}
				f_{n+1}(z) \\ g_{n+1}(z)
			\end{bmatrix}=T_{n+1}(z)\begin{bmatrix}
				f_n(z) \\ g_n(z)
			\end{bmatrix}.
		\end{equation*}
	\end{enumerate}
\end{theorem}
From a spectral-theoretical point of view, one can easily read from Theorem \ref{thm_cmv_extended_GZ_lem2.2_intro} that the Gesztesy-Zinchenko transfer matrix directly describes the relationship between entries of the generalized eigensolutions without reference to OPUCs. Moreover, unlike operators on $\ell^2(\mathbb{N}_0)$, one does not have OPUCs to work with in the $\ell^2(\mathbb{Z})$ case.

We prove an analogue of Theorem \ref{thm_cmv_extended_GZ_lem2.2_intro} for the rotated CMV operator $\mathcal{E}^\angle$, which looks largely identical to Theorem \ref{thm_cmv_extended_GZ_lem2.2_intro}, except that we replace all operators by their rotated counterparts. In particular, we show that the rotated Gesztesy-Zinchenko transfer matrix also carries an identical form to (\ref{eq_GZ_transfer_matrix_intro}), where the $\rho_n$'s are now complex numbers in $\mathbb{D}$, which implies that unlike the original case, the inverse of the rotated Gesztesy-Zinchenko transfer matrix would carry $\overline{\rho_n}$'s instead of $\rho_n$'s. The key obstacle to overcome was that we could not directly replace $\mathcal{E}^T$ with $\left(\mathcal{E}^\angle\right)^T$, since rotating the standard CMV matrices causes a loss of transpose symmetry. We bypass the issue of transpose symmetry by observing that since $\mathcal{E}^T=\tilde{\mathcal{M}}\tilde{\mathcal{L}}$, a similar proof of Theorem \ref{thm_cmv_extended_GZ_lem2.2_intro} can be applied to the operator $\tilde{\mathcal{E}}^\angle=\tilde{\mathcal{M}}^\angle\tilde{\mathcal{L}}^\angle$, where $\tilde{\mathcal{L}}^\angle$ and $\tilde{\mathcal{M}}^\angle$ are rotated versions of $\tilde{\mathcal{L}}$ and $\tilde{\mathcal{M}}$, respectively. This fact motivates the definition of a rotated extended CMV matrix as mentioned.

The paper is structured as follows: in section 2 we first review relevant results from OPUC theory and Gesztesy-Zinchenko theory for CMV operators, before proceeding to present our new results in section 3.


\addtocontents{toc}{\protect\setcounter{tocdepth}{0}}

\section*{Acknowledgments}
The author is supported in part by Xiamen University Malaysia Research Fund (grant numbers XMUMRF/2020-C5/IMAT/0011 and XMUMRF/2023C11/IMAT/0024). This research is also submitted as a M.Sc thesis at Xiamen University Malaysia. The author would like to thank his advisor Darren C. Ong for his guidance throughout the project; the anonymous reviewers for useful suggestions and improvements to the manuscript; and Maxim Zinchenko for helpful discussions regarding the initial conditions of the Gesztesy-Zinchenko difference equation and second kind polynomials.

\addtocontents{toc}{\protect\setcounter{tocdepth}{2}}



\section{Review of OPUCs and Gesztesy-Zinchenko Theory}

We review some key and relevant concepts from OPUC theory and Gesztesy-Zinchenko theory \cite{Gesztesy2006}. For more details, we refer the reader to an encyclopedia of two comprehensive volumes \cite{Simon2005a,Simon2005b} by Barry Simon, who also published a condensed summary \cite{Simon2005c} of these two volumes. We largely follow the definitions, notations and language from \cite{Simon2005a} and \cite{Simon2005b}, and cite their original sources whenever possible. 

\subsection{Orthogonal Polynomials On The Unit Circle (OPUC)}

We review the basic theory of orthogonal polynomials on the unit circle, in particular selected results from Sections 1.1 and 1.5 of \cite{Simon2005a}.

Let $\mu$ be a \textit{nontrivial} probability measure $\mu$ on $\partial\mathbb{D}$, that is, a measure $\mu$ supported on an infinite set with $\mu(\partial\mathbb{D})=1$. We consider the complex Hilbert space $L^2(\partial\mathbb{D},\mu)$ equipped with the Hermitian inner product which is conjugate linear in the first entry:
\begin{equation} \label{eq_L2_ip}
	\ip{f}{g}=\int_{\partial\mathbb{D}} \overline{f(z)}g(z)\,d\mu(z).
\end{equation}
The $L^2$-\textit{norm} of any $f\in L^2(\partial\mathbb{D},\mu)$ is then given by
\begin{equation} \label{eq_L2_norm}
	\norm{f}_{L^2}=\norm{f}=\ip{f}{f}^{\frac{1}{2}}.
\end{equation}
Since $\mu$ is a nontrivial probability measure on $\partial\mathbb{D}$, the set $\{1,z,z^2,\cdots\}$ is linearly independent in $L^2(\partial\mathbb{D},\mu)$, which allows us to apply the Gram-Schmidt process to obtain a set of \textit{monic} orthogonal polynomials $\{\Phi_n\}_{n=0}^\infty$, which we call the \textit{standard OPUCs}. The corresponding \textit{orthonormal} polynomials $\{\varphi_n\}_{n=1}^\infty$ are then given by
\begin{equation} \label{eq_opuc_orthonormal}
	\varphi_n(z)=\dfrac{\Phi_n(z)}{\norm{\Phi_n(z)}}.
\end{equation}
Define the anti-unitary map $R_n:L^2(\partial\mathbb{D},\mu)\to L^2(\partial\mathbb{D},\mu)$ by
\begin{equation} \label{eq_opuc_reverse_map}
	R_n(f)(z):=f^*(z)=z^n\overline{f\left(\dfrac{1}{\overline{z}}\right)}.
\end{equation}
If $p_n(z)=\di\sum_{j=0}^n c_jz^j$ is a polynomial of degree $n$, then
\begin{equation}
	p_n^*(z)=R_n(p_n)(z)=\sum_{j=0}^n \overline{c_j}z^{n-j}=\sum_{j=0}^n \overline{c_{n-j}}z^j.
\end{equation}
In particular, for the monic OPUCs $\{\Phi_n(z)\}_{n=0}^\infty$, we have
\begin{equation}
	\Phi_n^*(z)=z^n\left(\overline{z^n+\text{lower order terms}}\right)=1+\text{higher order terms},
\end{equation}
and hence for every $n\in\mathbb{N}_0$, we have
\begin{equation} \label{eq_opuc_reverse_eval_0}
	\Phi_n^*(0)=1.
\end{equation}
The polynomials $\{\Phi_n^*\}_{n=0}^\infty$ and $\{\varphi_n^*\}_{n=0}^\infty$ are called the \textit{standard reverse polynomials} or \textit{Szeg\H{o} duals} of $\{\Phi_n\}_{n=0}^\infty$ and $\{\varphi_n\}_{n=0}^\infty$, respectively.

The most significant property of the reverse OPUCs is that they are the unique polynomials orthogonal to $\{z,z^2,\cdots,z^n\}$ in $L^2(\partial\mathbb{D},\mu)$ up to a scalar multiple, which is a direct consequence of the fact that $\Phi_n$ is the unique monic polynomial of degree $n$ orthogonal to $\{1,z,\cdots,z^{n-1}\}$ in $L^2(\partial\mathbb{D},\mu)$. Moreover, the reverse OPUCs and the OPUCs have equal $L^2(\partial\mathbb{D},\mu)$-norms, which can be simply computed by integrating the reverse OPUCs with respect to the measure $\mu$.

\pagebreak

\begin{theorem}\cite[Lemma 1.5.1]{Simon2005a} \label{thm_opuc_reverse_orthogonal_to_z_z^n}
	Let $n\in\mathbb{N}$. If $p$ is a polynomial in $L^2(\partial\mathbb{D},\mu)$ such that
	\begin{enumerate}
		\item $\deg p(z)\leq n$; and
		\item $\ip{p(z)}{z^j}=0$ for all $j\in\{1,2,\cdots,n\}$,
	\end{enumerate}
	where $\ip{\cdot}{\cdot}$ is the $L^2(\partial\mathbb{D},\mu)$-inner product given by (\ref{eq_L2_ip}), then $$p(z)=c\Phi_n^*(z)$$ for some constant $c\in\mathbb{C}$. Moreover,
	\begin{equation}
		\norm{\Phi_n^*}^2=\norm{\Phi_n}^2=\int_{\partial\mathbb{D}}\Phi_n^*(z)\,d\mu(z).
	\end{equation}
\end{theorem}

The recurrence relation governing the standard OPUCs comprises of the five terms $\Phi_{n+1}$, $\Phi_n$, $\Phi_n^*$, $z\Phi_n$ and $z\Phi_n^*$, first proved by Gabor Szeg\H{o} \cite{Szego1975} as follows:
\begin{theorem}[Forward Szeg\H{o} recursion for monic OPUCs] \label{thm_szego_recursion}
	Let $\mu$ be a nontrivial probability measure on $\partial\mathbb{D}$ and $\{\Phi_n(z,\mu)\}_{n=0}^\infty$ be the monic orthogonal polynomials as defined in Section 2.1. There exists a sequence $\{\alpha_n\}_{n=0}^\infty\subseteq\mathbb{D}$, called the \textit{Verblunsky coefficients} of the measure $\mu$, such that
	\begin{align}
		\label{eq_szegorecursion_monic}
		\Phi_{n+1}(z)=z\Phi_n(z)-\overline{\alpha_n}\Phi_{n}^*(z),\\
		\label{eq_szegorecursion_monic_reverse}
		\Phi_{n+1}^*(z)=\Phi_n^*(z)-\alpha_nz\Phi_{n}(z),
	\end{align}
	and hence
	\begin{equation} \label{eq_opuc_norm}
		\norm{\Phi_{n+1}^*}=\norm{\Phi_{n+1}}=\prod_{j=0}^n \left(1-|\alpha_j|^2\right)^{\frac{1}{2}}.
	\end{equation}
\end{theorem}

We shall write $\alpha_n(\mu)$ in place of $\alpha_n$ when we want to emphasise that $\{\alpha_n\}_{n=0}^\infty$ are the Verblunsky coefficients associated to the measure $\mu$.

The reader is warned that in this paper, we adopt Simon's convention for $\{\alpha_n\}_{n=0}^\infty$ used in \cite{Simon2005a}, \cite{Simon2005b} and \cite{Simon2005c}. Another common alternative convention differs from Simon's convention by a negative complex conjugate, which was adopted by \cite{Cantero2003} and \cite{Gesztesy2006}, whose $a_n$'s and $\alpha_n$'s, respectively, correspond to our $-\overline{\alpha_n}$'s. 

By substituting $z=0$ and (\ref{eq_opuc_reverse_eval_0}) into (\ref{eq_szegorecursion_monic}), one can show that the Verblunsky coefficients $\{\alpha_n\}_{n=0}^\infty$ are precisely given by
\begin{equation} \label{eq_verblunsky_coeff_monic}
	\alpha_n=-\overline{\Phi_{n+1}(0)}
\end{equation}
for all $n=0,1,2,\ldots$. 

Since $\Phi_0(0)=1$ by definition (as a result of orthogonalizing $\{1,z,z^2,\ldots\}$ via the Gram-Schmidt process), it is natural for one to extend (\ref{eq_verblunsky_coeff_monic}) to $n=-1$ by defining
\begin{equation} \label{eq_verblunsky_coeff_-1}
	\alpha_{-1}=-1.
\end{equation}

Once the forward Szeg\H{o} recursion for the monic OPUCs is established, one can easily derive the forward Szeg\H{o} recursion for the orthonormal polynomials by dividing both sides of (\ref{eq_szegorecursion_monic}) and (\ref{eq_szegorecursion_monic_reverse}) by $\norm{\Phi_{n+1}^*}$:

\begin{corollary}[Forward Szeg\H{o} recursion for orthonormal polynomials]
	For every $n\in\mathbb{N}_0$, we have
	\begin{align} 
		\label{eq_szegorecursion_orthonormal}
		\varphi_{n+1}(z)=\rho_n^{-1}\left(z\varphi_n(z)-\overline{\alpha_n}\varphi_n^*(z)\right),\\
		\label{eq_szegorecursion_orthonormal_reverse}
		\varphi_{n+1}^*(z)=\rho_n^{-1}\left(\varphi_n^*(z)-\alpha_nz\varphi_n(z)\right).
	\end{align}
\end{corollary}

Frequently one may find it useful to express the forward Szeg\H{o} recursions (\ref{eq_szegorecursion_orthonormal}) and (\ref{eq_szegorecursion_orthonormal_reverse}) in matrix form:
\begin{equation} \label{eq_szegorecursion_opuc_matrix}
	\begin{bmatrix}
		\varphi_{n+1}(z) \\ \varphi_{n+1}^*(z)
	\end{bmatrix}=A_n(z)\begin{bmatrix}
		\varphi_n(z) \\ \varphi_n^*(z)
	\end{bmatrix},
\end{equation}
where
\begin{equation} \label{eq_szegorecursion_opuc_An}
	A_n(z)=\rho_n^{-1}\begin{bmatrix}
		z & -\overline{\alpha_n} \\ -\alpha_n z & 1
	\end{bmatrix}.
\end{equation}
$A_n$ is called the \textit{Szeg\H{o} transfer matrix} for the orthonormal polynomials $\{\varphi_n\}_{n=0}^\infty$. 

If $\{\alpha_n\}_{n=0}^\infty\subseteq\mathbb{D}$ so that $A_n(z)$ is invertible, one may easily use (\ref{eq_szegorecursion_opuc_matrix}) to compute 
\begin{equation} \label{eq_szegorecursion_opuc_An_inverse}
	A_n^{-1}(z)=\dfrac{1}{z\rho_n}\begin{bmatrix}
		1 & \overline{\alpha_n} \\ \alpha_n z & z
	\end{bmatrix}
\end{equation}
and hence obtain the \textit{backward Szeg\H{o} recursion} for the orthonormal polynomials:

\begin{theorem}[Backward Szeg\H{o} recursion for orthonormal polynomials] \label{thm_szego_recursion_inverse}
	For every $n\in\mathbb{N}_0$, we have
	\begin{align} 
		\label{eq_szegorecursion_inverse_orthonormal}
		z\varphi_n(z)=\rho_n^{-1}\left(\varphi_{n+1}(z)+\overline{\alpha_n}\varphi_{n+1}^*(z)\right),\\
		\label{eq_szegorecursion_inverse_orthonormal_reverse}
		\varphi_n^*(z)=\rho_n^{-1}\left(\varphi_{n+1}^*(z)+\alpha_n\varphi_{n+1}(z)\right).
	\end{align}
\end{theorem}

\begin{theorem}[Christoffel-Darboux Formula] \label{thm_opuc_cd_formula}
	For any $n\in\mathbb{N}_0$ and any $\xi,z\in\mathbb{C}$ such that $\overline{\xi}{z}\neq 1$, we have
	\begin{equation} \label{eq_opuc_cd}
		\sum_{k=0}^n \overline{\varphi_k(\xi)}\varphi_k(z)=\dfrac{\overline{\varphi_{n+1}^*(\xi)}\varphi_{n+1}^*(z)-\overline{\varphi_{n+1}(\xi)}\varphi_{n+1}(z)}{1-\overline{\xi}z}.
	\end{equation}
\end{theorem}

Some applications of the Christoffel-Darboux formula are documented in \cite{Simon2008}, which in fact discusses Christoffel-Darboux kernels in more generality. In particular, the Christoffel-Darboux kernel for OPUCs is defined by the left hand side of (\ref{eq_opuc_cd}). Another interesting fact is that in \cite{Szego1975}, Szeg\H{o} used the Christoffel-Darboux formula (\ref{eq_opuc_cd}) to prove the recursions (\ref{eq_szegorecursion_monic}) and (\ref{eq_szegorecursion_monic_reverse}).

\begin{definition}[Carath\'eodory function] \label{def_caratheodory_function}
	A \textit{Carath\'eodory function} is an analytic function \\ $F:\mathbb{D}\to\mathbb{C}$ such that $F(0)=1$ and $\Re{F(z)}>0$.
\end{definition}
Given a probability measure $\mu$ on $\partial\mathbb{D}$, one may easily verify that
\begin{equation} \label{eq_caratheodory_function_measure}
	F(z)=\int_{\partial\mathbb{D}}\dfrac{\xi+z}{\xi-z}\,d\mu(\xi)
\end{equation}
defines a Carath\'eodory function. Moreover, the Herglotz-Riesz representation theorem \cite{Herglotz1911,Riesz1911} implies that for every Carath\'eodory function $F$, there exists a unique probability measure $\mu$ on $\partial\mathbb{D}$ such that $F$ has a representation of this form.

By combining the \hyperref[thm_szego_recursion]{Szeg\H{o} recursion}, Verblunsky's theorem and the Herglotz-Riesz representation theorem \cite{Herglotz1911,Riesz1911}, one establishes one-to-one correspondences between
\begin{center}
	$\left\{\text{probability measures $\mu$ on $\partial\mathbb{D}$}\right\}$\\
	$\updownarrow$\\ $\left\{\text{sequences of monic OPUCs $\{\Phi_n\}_{n=0}^\infty$ }\right\}$\\
	$\updownarrow$\\ $\left\{\text{sequences of complex numbers $\{\alpha_n\}_{n=0}^\infty$ in $\mathbb{D}$}\right\}$\\
	$\updownarrow$\\ $\left\{\text{Carath\'eodory functions $F$}\right\}$.
\end{center}

	\subsection{Second Kind Polynomials}
	
	We review selected results from \cite{Golinskii2001}, which can also be found in Section 3.2 of \cite{Simon2005a}.
	
	Let $\mu$ be a probability measure on $\partial\mathbb{D}$ and $\{\alpha_n(\mu)\}_{n=0}^\infty\subseteq\mathbb{D}$ be its Verblunsky coefficients. The \textit{second kind measure} $\mu_{-1}$ associated to $\mu$ is defined by
	\begin{equation} \label{eq_second_kind_measure}
		\alpha_n(\mu_{-1})=-\alpha_n(\mu).
	\end{equation}
	The (monic) \textit{second kind polynomials} $\{\Psi_n(z,\mu)\}_{n=0}^\infty$ associated to $\mu$ are defined to be the monic OPUCs associated to the second kind measure $\mu_{-1}$:
	\begin{equation} \label{eq_opuc_second_kind}
		\Psi_n(z,\mu)=\Phi_n(z,\mu_{-1}).
	\end{equation}	
	The \textit{second kind orthonormal polynomials} $\{\psi_n(z)\}_{n=0}^\infty$ are given by
	\begin{equation}
		\psi_n(z)=\dfrac{\Psi_n(z)}{\norm{\Psi_n(z)}}.
	\end{equation}
	\hyperref[thm_szego_recursion]{Szeg\H{o}'s recursion} for the orthonormal second kind OPUCs is given by
	\begin{align}
		\psi_{n+1}(z)&=\rho_n^{-1}\Bigl(z\psi_n(z)-\left(\overline{-\alpha_n}\right)\psi_n^*(z)\Bigr),\\
		\psi_{n+1}^*(z)&=\rho_n^{-1}\Bigl(\psi_n^*(z)-\left(-\alpha_n\right)z\psi_n(z)\Bigr),
	\end{align}
	in matrix form by
	\begin{equation} \label{eq_szegorecursion_orthonormal_lambda}
		\begin{bmatrix}
			\psi_{n+1}(z) \\ -\psi_{n+1}^*(z)
		\end{bmatrix}=A_n(z)\begin{bmatrix}
			\psi_n(z) \\ -\psi_n^*(z)
		\end{bmatrix},
	\end{equation}
	where $A_n(z)$ is given by (\ref{eq_szegorecursion_opuc_An}). 
	
	One immediately observes that the standard OPUCs $\begin{bmatrix}
		\varphi_n(z) \\ \varphi_n^*(z)
	\end{bmatrix}$ and the second kind polynomials $\begin{bmatrix}
		\psi_n(z) \\ -\psi_n^*(z)
	\end{bmatrix}$ satisfy the \textit{Szeg\H{o} difference equation}
	\begin{equation} \label{eq_opuc_lambda_differenceeq}
		u_{n+1}(z)=A_n(z)u_n(z)=\dfrac{1}{\rho_n}\begin{bmatrix}
			z & -\overline{\alpha_n} \\ -\alpha_n z & 1
		\end{bmatrix}u_n(z).
	\end{equation}
	with initial conditions $u_0(z)=\begin{bmatrix}
	1 \\ 1
	\end{bmatrix}$
	and  
	$u_0(z)=\begin{bmatrix}
		1 \\ -1
	\end{bmatrix}$,
	respectively.
	
	One may discuss, in more generality, the Aleksandrov family $\{\mu_\lambda\}_{\lambda\in\partial\mathbb{D}}$ of measures \cite{Aleksandrov1987} associated to $\mu$, which satisfy 
	\begin{equation}
		\alpha_n(\mu_\lambda)=\lambda\alpha_n(\mu),
	\end{equation}
 	and their associated OPUCs $\{\varphi_n^\lambda\}_{n=0}^\infty$ with boundary condition $\lambda$, which are defined similarly. Golinskii and Nevai
 	\cite{Golinskii2001} drew the links between the Aleksandrov family and OPUC theory. Since we are primarily interested in the boundary condition $\lambda=-1$, we do not pursue this direction in further detail.
 	
 	\begin{example}
 		Consider the probability measures $\mu$ and $\nu$ on $\partial\mathbb{D}$ given by
 		$$d\mu=(1-\cos\theta)\dfrac{d\theta}{2\pi},\qquad
 		d\nu=\dfrac{1}{2}\dfrac{d\theta}{2\pi}+\dfrac{1}{2}d\delta_0(\theta),$$
 		where $d\theta$ represents the Lebesgue measure on $\mathbb{R}$ and $\delta_0$ represents the Dirac delta measure centred at $\theta=0$ (equivalently $z=1$). Their Verblunsky coefficients satisfy
 		\begin{equation}
 			\alpha_n(\mu)=-\dfrac{1}{n+2}=-\alpha_n(\nu),
 		\end{equation}
 		so $\nu=\mu_{-1}$ and $\mu=\nu_{-1}$, that is, $\mu$ and $\nu$ are, respectively, the second kind measures of each other. Their corresponding orthonormal polynomials, or equivalently, the second kind polynomials of their respective counterparts, are given by
 		\begin{gather}
 			\varphi_n(z,\mu)=\psi_n(z,\nu)=\left(\dfrac{2}{(n+1)(n+2)}\right)^{\frac{1}{2}}\sum_{k=0}^n(k+1)z^k,\\
 			\varphi_n(z,\nu)=\psi_n(z,\mu)=\left(1-\dfrac{1}{2}n\alpha_{n-1}\right)^{-\frac{1}{2}}\Big(z^n-\alpha_{n-1}(z^{n-1}+z^{n-2}+\cdots+1)\Big),
 		\end{gather}
 		respectively. The reader is referred to \cite[Examples 1.6.3 and 1.6.4]{Simon2005a} for detailed computations.
 	\end{example}
	
	\begin{proposition}\cite[Proposition 3.2.2]{Simon2005a}
		The second kind polynomials satisfy
		\begin{align}
			\label{eq_opuc_with_boundary_condition_monic_z^n}
			\Psi_n^*(z)\Phi_n(z)+\Phi_n^*(z)\Psi_n(z)&=2z^n\prod_{j=0}^{n-1}\rho_j^2,\\
			\label{eq_opuc_with_boundary_condition_orthonormal_z^n}
			\psi_n^*(z)\varphi_n(z)+\varphi_n^*(z)\psi_n(z)&=2z^n.
		\end{align}
	\end{proposition}
	
	For $z\in\partial\mathbb{D}$, substituting (\ref{eq_opuc_reverse_map}) into (\ref{eq_opuc_with_boundary_condition_orthonormal_z^n}) gives
	\begin{align*}
		2z^n=z^n\cdot 2\Re{\overline{\psi_n(z)}\varphi_n(z)}
	\end{align*}
	and hence
	\begin{equation}
		\Re{\overline{\psi_n(z)}\varphi_n(z)}=1,
	\end{equation}
	which implies that
	\begin{equation}
		\left|\psi_n(z)\right|\left|\varphi_n(z)\right|\geq 1.
	\end{equation}
	
	\begin{theorem}[Mixed Christoffel-Darboux Formula] \label{thm_opuc_cd_formula_mixed}
		For any $n\in\mathbb{N}_0$ and any $\xi,z\in\mathbb{C}$ such that $\overline{\xi}{z}\neq 1$,
		\begin{equation} \label{eq_opuc_cd_mixed}
			\sum_{k=0}^n \overline{\varphi_k(\xi)}\psi_k(z)=\dfrac{2-\overline{\varphi_{n+1}^*(\xi)}\psi_{n+1}^*(z)-\overline{\varphi_{n+1}(\xi)}\psi_{n+1}(z)}{1-\overline{\xi}z}.
		\end{equation}
	\end{theorem}
	
	\begin{proposition}\cite[Proposition 3.2.8]{Simon2005a}
		For any $n\in\mathbb{N}_0$,
		\begin{align}
			\label{eq_opuc_second_kind_caratheodory}
			\psi_n(z)&=\int_{\partial\mathbb{D}}\left(\varphi_n(\xi)-\varphi_n(z)\right)\dfrac{\xi+z}{\xi-z}\,d\mu(\xi),\\
			\label{eq_opuc_reverse_second_kind_caratheodory}
			\psi_n^*(z)&=\int_{\partial\mathbb{D}}\left(\varphi_n^*(z)-\varphi_n^*(\xi)\right)\dfrac{\xi+z}{\xi-z}\,d\mu(\xi).
		\end{align}
	\end{proposition}
	In fact, Geronimus \cite{Geronimus1954} first defined the second kind polynomials by (\ref{eq_opuc_second_kind_caratheodory}), which Geronimo \cite{Geronimo1993} first used to prove the following theorem:
	
	\begin{theorem}[Geronimo \cite{Geronimo1993}, Golinskii-Nevai \cite{Golinskii2001}] \label{thm_geronimo_golinskii_nevai}
		Let $\mu$ be a probability measure on $\partial\mathbb{D}$, $z\in\mathbb{D}$ and $r\in\mathbb{C}$. Then the sequence
		\begin{equation}
			\left\{\begin{bmatrix}
				\psi_n(z,\mu) \\ -\psi_n^*(z,\mu)
			\end{bmatrix}+r\begin{bmatrix}
				\varphi_n(z,\mu) \\ \varphi_n^*(z,\mu)
			\end{bmatrix}\right\}_{n=0}^\infty\in\ell^2(\mathbb{N}_0,\mathbb{C}^2)
		\end{equation}
		if and only if $r=F(z)$, where $F$ is the Carath\'{e}odory function of $\mu$.
	\end{theorem}
	Hence $\left\{\begin{bmatrix}
		\psi_n(z,\mu) \\ -\psi_n^*(z,\mu)
	\end{bmatrix}+r\begin{bmatrix}
		\varphi_n(z,\mu) \\ \varphi_n^*(z,\mu)
	\end{bmatrix}\right\}_{n=0}^\infty$ is a $\ell^2$-solution of the difference equation (\ref{eq_opuc_lambda_differenceeq}), called the Weyl solution, for the reason that the Carath\'{e}odory function $F$ plays an analogous role to the Weyl $m$-function in OPRL.
	
	While Theorem \ref{thm_geronimo_golinskii_nevai} first appeared in \cite{Geronimo1993}, Golinskii-Nevai provided a second proof in \cite{Golinskii2001} using the Christoffel-Darboux formulae (\ref{eq_opuc_cd}) and (\ref{eq_opuc_cd_mixed}).

	\subsection{Gesztesy-Zinchenko Formulation for Extended CMV Operators}
	
	Due to different conventions used here compared to those in \cite{Gesztesy2006}, for the reader's convenience, here we review and translate relevant results from \cite{Gesztesy2006} into our notations. The reader is gently reminded that our $\alpha_n$'s are $-\overline{\alpha_n}$'s in \cite{Gesztesy2006}, and the indices in (\ref{eq_Tn_GZ_CMV}) differ by one: our $n+1$ is $n$ in \cite{Gesztesy2006}.
	
	In developing Weyl-Titchmarsh theory for CMV operators in \cite{Gesztesy2006}, Fritz Gesztesy and Maxim Zinchenko discovered that by studying the solutions $f(z)=\{f_n(z)\}_{n\in\mathbb{Z}}\in\ell^2(\mathbb{Z})$ of the equation $$\mathcal{E}f(z)=zf(z)$$
	together with another sequence $g(z)=\{g_n(z)\}_{n\in\mathbb{Z}}\in\ell^2(\mathbb{Z})$ gives elegant symmetries when coupled appropriately, producing the said Gesztesy-Zinchenko transfer matrix in the process. The result is stated precisely in the following theorem: 
	\begin{theorem}\cite[Lemma 2.2]{Gesztesy2006} \label{thm_cmv_extended_GZ_lem2.2}
		Let $z\in\mathbb{C}-\{0\}$ and $f=\{f_n(z)\}_{n\in\mathbb{Z}}$, $g=\{g_n(z)\}_{n\in\mathbb{Z}}$ be sequences of complex functions. Define the unitary operator $\mathcal{U}$ on $\ell^2(\mathbb{Z})^2$ by
		\begin{equation}
			\mathcal{U}=\begin{bmatrix}
				\mathcal{E} & 0 \\ 0 & \tilde{\mathcal{E}}
			\end{bmatrix}
			=\begin{bmatrix}
				\tilde{\mathcal{L}}\tilde{\mathcal{M}} & 0 \\ 
				0 & \tilde{\mathcal{M}}\tilde{\mathcal{L}}
			\end{bmatrix}
			=\begin{bmatrix}
				0 & \tilde{\mathcal{L}} \\ \tilde{\mathcal{M}} & 0
			\end{bmatrix}^2.
		\end{equation}
		Then the following statements are equivalent:
		\begin{enumerate}
			\item $(\mathcal{E} f)(z)=zf(z)$ and $(\tilde{\mathcal{M}} f)(z)=zg(z)$;
			
			\item $(\tilde{\mathcal{E}} g)(z)=zg(z)$ and $(\tilde{\mathcal{L}} g)(z)=f(z)$;
			
			\item $(\tilde{\mathcal{M}} f)(z)=zg(z)$ and $(\tilde{\mathcal{L}} g)(z)=f(z)$;
			
			\item $\mathcal{U}\begin{bmatrix}
				f(z) \\ g(z)
			\end{bmatrix}=z\begin{bmatrix}
				f(z) \\ g(z)
			\end{bmatrix}$ and $(\tilde{\mathcal{M}}f)(z)=zg(z)$;
			
			\item $\mathcal{U}\begin{bmatrix}
				f(z) \\ g(z)
			\end{bmatrix}=z\begin{bmatrix}
				f(z) \\ g(z)
			\end{bmatrix}$ and $(\tilde{\mathcal{L}}g)(z)=f(z)$;
			
			\item for each $n\in\mathbb{Z}$,
			\begin{equation} \label{eq_GZ_transition_cmv}
				\begin{bmatrix}
					f_{n+1}(z) \\ g_{n+1}(z)
				\end{bmatrix}=T_{n+1}(z)\begin{bmatrix}
					f_n(z) \\ g_n(z)
				\end{bmatrix},
			\end{equation}
			where
			\begin{equation} \label{eq_Tn_GZ_CMV}
				T_n(z)=\begin{cases}
					\dfrac{1}{\rho_n}\begin{bmatrix}
						-\overline{\alpha_n} & z \\ z^{-1} & -\alpha_n
					\end{bmatrix} & n\text{ even};\\
					~\\
					\dfrac{1}{\rho_n}\begin{bmatrix}
						-\alpha_n & 1 \\ 1 & -\overline{\alpha_n}
					\end{bmatrix} & n\text{ odd}.
				\end{cases}
			\end{equation}
		\end{enumerate}
	\end{theorem}
	$T_n$ is called the \textit{Gesztesy-Zinchenko transfer matrix} associated to $\mathcal{E}$.
	
	If there exists $K\in\mathbb{Z}$ such that $\alpha_K=e^{is}\in\partial\mathbb{D}$ for some $s\in(-\pi,\pi]$, then $\rho_K=0$, then $\mathcal{E}$ can be expressed as a direct sum
	\begin{equation}
		\mathcal{E}=\mathcal{E}_{K^-}(s)\oplus\mathcal{E}_{K^+}(s),
	\end{equation}
	where $\mathcal{E}_{K^-}(s)$ and $\mathcal{E}_{K^+}(s)$ are operators on $\ell^2(\mathbb{Z}\cap\left(-\infty,K]\right)$ and $\ell^2\left(\mathbb{Z}\cap[K+1,\infty)\right)$, respectively. One may check that $\mathcal{E}_{K^+}$ looks almost like $\mathcal{C}$, whereas $\mathcal{E}_{K^-}$ also looks almost like $\mathcal{C}$ but extending in the opposite direction. Hence following a similar argument to the proof of the $\mathcal{LM}$-decomposition of CMV matrices, one may show that $\mathcal{E}_{K^\pm}(s)$ also admit $\mathcal{L}\mathcal{M}$-decompositions
	\begin{equation} \label{eq_cmv_split}
		\mathcal{E}_{K^\pm}(s)=\mathcal{L}_{K^\pm}(s)\mathcal{M}_{K^\pm}(s).
	\end{equation}
	We are particularly interested in the case where $s=\pi$, which gives $\alpha_K=-1$ following \cite{Simon2005a}'s convention. Hence one may simply verify that $\mathcal{E}_{K^+}(\pi)$ looks exactly like $\mathcal{C}$. Subsequently, we shall simply denote $\mathcal{E}_{K^\pm}(\pi)$ by $\mathcal{E}_{K^\pm}$. The operators $\mathcal{L}_{K^\pm}$ and $\mathcal{M}_{K^\pm}$ are also defined similarly. 
	
	One can also prove analogues of Theorem \ref{thm_cmv_extended_GZ_lem2.2} for the half-lattice operators $\mathcal{E}_{K^+}$ and $\mathcal{E}_{K^-}$:
	
	\begin{theorem} \cite[Lemma 2.3]{Gesztesy2006} \label{thm_cmv_extended_GZ_lem2.3_right}
		Let $z\in\mathbb{C}-\{0\}$ and $f_{K^+}=\{f_n(z)\}_{n=K+1}^\infty$, $g_{K^+}=\{g_n(z)\}_{n=K+1}^\infty$ be sequences of complex functions. Define the operator $\mathcal{U}_{K^+}$ on $\ell^2(\mathbb{Z}\cap[K+1,\infty))^2$ by
		\begin{equation}
			\mathcal{U}_{K^+}=\begin{bmatrix}
				\mathcal{E}_{K^+} & 0 \\ 0 & \tilde{\mathcal{E}}_{K^+}
			\end{bmatrix}
			=\begin{bmatrix}
				\tilde{\mathcal{L}}_{K^+}\tilde{\mathcal{M}}_{K^+} & 0 \\ 
				0 & \tilde{\mathcal{M}}_{K^+}\tilde{\mathcal{L}}_{K^+}
			\end{bmatrix}
			=\begin{bmatrix}
				0 & \tilde{\mathcal{L}}_{K^+} \\ \tilde{\mathcal{M}}_{K^+} & 0
			\end{bmatrix}^2,
		\end{equation}
		where $K\in\mathbb{Z}$ is such that $\alpha_K\in\partial\mathbb{D}$. Then the following statements are equivalent:
		\begin{enumerate}
			\item $(\mathcal{E}_{K^+} f_{K^+})(z)=zf_{K^+}(z)$ and $(\tilde{\mathcal{M}}_{K^+} f_{K^+})(z)=zg_{K^+}(z)$;
			
			\item $(\tilde{\mathcal{E}}_{K^+} g_{K^+})(z)=zg_{K^+}(z)$ and $(\tilde{\mathcal{L}}_{K^+} g_{K^+})(z)=f_{K^+}(z)$;
			
			\item $(\tilde{\mathcal{M}}_{K^+} f_{K^+})(z)=zg_{K^+}(z)$ and $(\tilde{\mathcal{L}}_{K^+} g_{K^+})(z)=f_{K^+}(z)$;
			
			\item $\mathcal{U}_{K^+}\begin{bmatrix}
				f_{K^+}(z) \\ g_{K^+}(z)
			\end{bmatrix}=z\begin{bmatrix}
				f_{K^+}(z) \\ g_{K^+}(z)
			\end{bmatrix}$ and $(\tilde{\mathcal{M}}_{K^+}f_{K^+})(z)=zg_{K^+}(z)$;
			
			\item $\mathcal{U}_{K^+}\begin{bmatrix}
				f_{K^+}(z) \\ g_{K^+}(z)
			\end{bmatrix}=z\begin{bmatrix}
				f_{K^+}(z) \\ g_{K^+}(z)
			\end{bmatrix}$ and $(\tilde{\mathcal{L}}_{K^+}g_{K^+})(z)=f_{K^+}(z)$;
			
			\item for each $n\in\mathbb{Z}\cap[K+1,\infty)$,
			\begin{equation} \label{eq_GZ_transition_cmv_+}
				\begin{bmatrix}
					f_{n+1}(z) \\ g_{n+1}(z)
				\end{bmatrix}=T_{n+1}(z)\begin{bmatrix}
					f_n(z) \\ g_n(z)
				\end{bmatrix},
			\end{equation}
			where $T_n(z)$ is given by (\ref{eq_Tn_GZ_CMV}) and
			\begin{equation} \label{eq_GZ_CMV_fK}
				f_{K+1}(z)=\begin{cases}
					zg_{K+1}(z) & K \text{ even}; \\ g_{K+1}(z) & K \text{ odd}.
				\end{cases}
			\end{equation}
		\end{enumerate}
	\end{theorem}
	
	\begin{theorem} \cite[Lemma 2.3]{Gesztesy2006}  \label{thm_cmv_extended_GZ_lem2.3_left}
		Let $z\in\mathbb{C}-\{0\}$ and $f_{K^-}=\{f_n(z)\}_{n=-\infty}^K$, $g_{K^-}=\{g_n(z)\}_{n=-\infty}^K$ be sequences of complex functions. Define the operator $\mathcal{U}_{K^-}$ on $\ell^2(\mathbb{Z}\cap(-\infty,K])^2$ by
		\begin{equation}
			\mathcal{U}_{K^-}=\begin{bmatrix}
				\mathcal{E}_{K^-} & 0 \\ 0 & \tilde{\mathcal{E}}_{K^-}
			\end{bmatrix}
			=\begin{bmatrix}
				\tilde{\mathcal{L}}_{K^-}\tilde{\mathcal{M}}_{K^-} & 0 \\ 
				0 & \tilde{\mathcal{M}}_{K^-}\tilde{\mathcal{L}}_{K^-}
			\end{bmatrix}
			=\begin{bmatrix}
				0 & \tilde{\mathcal{L}}_{K^-} \\ \tilde{\mathcal{M}}_{K^-} & 0
			\end{bmatrix}^2,
		\end{equation}
		where $K\in\mathbb{Z}$ is such that $\alpha_K\in\partial\mathbb{D}$. Then the following statements are equivalent:
		\begin{enumerate}
			\item $(\mathcal{E}_{K^-} f_{K^-})(z)=zf_{K^-}(z)$ and $(\tilde{\mathcal{M}}_{K^-} f_{K^-})(z)=zg_{K^-}(z)$;
			
			\item $(\tilde{\mathcal{E}}_{K^-} g_{K^-})(z)=zg_{K^-}(z)$ and $(\tilde{\mathcal{L}}_{K^+-} g_{K^-})(z)=f_{K^-}(z)$;
			
			\item $(\tilde{\mathcal{M}}_{K^-} f_{K^-})(z)=zg_{K^-}(z)$ and $(\tilde{\mathcal{L}}_{K^-} g_{K^-})(z)=f_{K^-}(z)$;
			
			\item $\mathcal{U}_{K^-}\begin{bmatrix}
				f_{K^-}(z) \\ g_{K^-}(z)
			\end{bmatrix}=z\begin{bmatrix}
				f_{K^-}(z) \\ g_{K^-}(z)
			\end{bmatrix}$ and $(\tilde{\mathcal{M}}_{K^-}f_{K^-})(z)=zg_{K^-}(z)$;
			
			\item $\mathcal{U}_{K^-}\begin{bmatrix}
				f_{K^-}(z) \\ g_{K^-}(z)
			\end{bmatrix}=z\begin{bmatrix}
				f_{K^-}(z) \\ g_{K^-}(z)
			\end{bmatrix}$ and $(\tilde{\mathcal{L}}_{K^-}g_{K^-})(z)=f_{K^-}(z)$;
			
			\item for each $n\in\mathbb{Z}\cap(-\infty,K]$,
			\begin{equation} \label{eq_GZ_transition_cmv_-}
				\begin{bmatrix}
					f_{n-1}(z) \\ g_{n-1}(z)
				\end{bmatrix}=T_{n}^{-1}(z)\begin{bmatrix}
					f_n(z) \\ g_n(z)
				\end{bmatrix},
			\end{equation}
			where $T_n(z)$ is given by (\ref{eq_Tn_GZ_CMV}) and
			\[f_K(z)=\begin{cases}
				-g_K(z) & K \text{ even}; \\ -zg_K(z) & K \text{ odd}.
			\end{cases}\]
		\end{enumerate}
	\end{theorem}
	
	Suppose $z\in\mathbb{C}-\{0\}$, $\left\{\begin{bmatrix}
		f_n^+(z) \\ g_n^+(z)
	\end{bmatrix}\right\}_{n=K}^\infty$ and $\left\{\begin{bmatrix}
		p_n^+(z) \\ q_n^+(z)
	\end{bmatrix}\right\}_{n=K}^\infty$ are linearly independent solutions of (\ref{eq_GZ_transition_cmv_+}) satisfying the initial conditions
	\begin{align}
		\label{eq_GZ_cmv_transition_init_fg_pq_+_even}
		\begin{bmatrix} 
			f_K^+(z) \\ g_K^+(z)
		\end{bmatrix}=\begin{bmatrix}
			z \\ 1
		\end{bmatrix},\qquad
		\begin{bmatrix}
			p_K^+(z) \\ q_K^+(z)
		\end{bmatrix}=\begin{bmatrix}
			z \\ -1
		\end{bmatrix}
	\end{align}
	when $K$ is even, and
	\begin{align}
		\label{eq_GZ_cmv_transition_init_fg_pq_+_odd}
		\begin{bmatrix}
			f_K^+(z) \\ g_K^+(z)
		\end{bmatrix}=\begin{bmatrix}
			1 \\ 1
		\end{bmatrix},\qquad
		\begin{bmatrix}
			p_K^+(z) \\ q_K^+(z)
		\end{bmatrix}=\begin{bmatrix}
			-1 \\ 1
		\end{bmatrix}
	\end{align}
	when $K$ is odd. Similarly, also suppose $\left\{\begin{bmatrix}
		f_n^-(z) \\ g_n^-(z)
	\end{bmatrix}\right\}_{n=-\infty}^K$ and $\left\{\begin{bmatrix}
		p_n^-(z) \\ q_n^-(z)
	\end{bmatrix}\right\}_{n=-\infty}^K$ are linearly independent solutions of (\ref{eq_GZ_transition_cmv_-}) satisfying the initial conditions
	\begin{align}
		\label{eq_GZ_cmv_transition_init_fg_-_even}
		\begin{bmatrix} 
			f_K^-(z) \\ g_K^-(z)
		\end{bmatrix}=\begin{bmatrix}
			1 \\ -1
		\end{bmatrix},\qquad
		\begin{bmatrix}
			p_K^-(z) \\ q_K^-(z)
		\end{bmatrix}=\begin{bmatrix}
			1 \\ 1
		\end{bmatrix}
	\end{align}
	when $K$ is even, and
	\begin{align}
		\label{eq_GZ_cmv_transition_init_fg_-_odd}
		\begin{bmatrix}
			f_K^-(z) \\ g_K^-(z)
		\end{bmatrix}=\begin{bmatrix}
			-z \\ 1
		\end{bmatrix},\qquad
		\begin{bmatrix}
			p_K^-(z) \\ q_K^-(z)
		\end{bmatrix}=\begin{bmatrix}
			z \\ 1
		\end{bmatrix}
	\end{align}
	when $K$ is odd. 
	
	To explain why the $^+$ sequences here start from index $K$ instead of $K+1$, one may assume that we are discussing the half lattice CMV matrices by their own rights - not as a consequence of a split of a full lattice CMV matrix on $\ell^2(\mathbb{Z})$. If these half lattice operators are obtained from a split, then one may simply replace the indices accordingly.
	
	These initial conditions are chosen to be the simplest choice such that the $\left\{\begin{bmatrix}
		f_n^+(z) \\ g_n^+(z)
	\end{bmatrix}\right\}_{n=K}^\infty$ sequence satisfies (\ref{eq_GZ_CMV_fK}) up to a scalar multiple, and the $\left\{\begin{bmatrix}
		p_n^+(z) \\ q_n^+(z)
	\end{bmatrix}\right\}_{n=K}^\infty$ sequence becomes the analogue of the second kind polynomial sequence $\left\{\begin{bmatrix}
		\psi_n(z) \\ -\psi_n^*(z)
	\end{bmatrix}\right\}_{n=K}^\infty$. This formulation allows one to prove an analogue of Theorem \ref{thm_geronimo_golinskii_nevai} for $\mathcal{E}_{K^+}$:
	\begin{theorem} \cite[Corollary 2.14]{Gesztesy2006}
		Let $z\in\mathbb{D}$ and $r\in\mathbb{C}$. Then the sequence
		\begin{equation}
			\left\{\begin{bmatrix}
				p_n^+(z) \\ q_n^+(z)
			\end{bmatrix}+r\begin{bmatrix}
				f_n^+(z) \\ g_n^+(z)
			\end{bmatrix}\right\}_{n=K}^\infty\in \ell^2\left(\mathbb{Z}\cap[K+1,\infty),\mathbb{C}^2\right)
		\end{equation}
		if and only if $r=F(z)$, where $F$ is the Carath\'{e}odory function of the spectral measure of $\mathcal{E}_{K^+}$.
	\end{theorem}
	The choices of initial conditions of the sequences $\left\{\begin{bmatrix}
		f_n^-(z) \\ g_n^-(z)
	\end{bmatrix}\right\}_{n=-\infty}^K$ and $\left\{\begin{bmatrix}
		p_n^-(z) \\ q_n^-(z)
	\end{bmatrix}\right\}_{n=-\infty}^K$ and their consequences can be explained similarly.
	
	One may use (\ref{eq_GZ_transition_cmv}) to extend the sequences $\left\{\begin{bmatrix}
		f_n^+(z) \\ g_n^+(z)
	\end{bmatrix}\right\}_{n=K}^\infty$ and $\left\{\begin{bmatrix}
		p_n^+(z) \\ q_n^+(z)
	\end{bmatrix}\right\}_{n=K}^\infty$ to all\\ $n<K$, and similarly extend $\left\{\begin{bmatrix}
		f_n^-(z) \\ g_n^-(z)
	\end{bmatrix}\right\}_{n=-\infty}^K$ and $\left\{\begin{bmatrix}
		p_n^-(z) \\ q_n^-(z)
	\end{bmatrix}\right\}_{n=-\infty}^K$ to all $n>K$.
	
	\vspace{2\baselineskip}
	
	\begin{table}[ht]
		\centering
		\begin{tabular}{||c||c|c|c||}
			\hhline{|=|=|=|=|} & & & \\ [-2ex]
			$n$ & $K-1$ & $K$ even & $K+1$\\ [0.5ex]
			\hhline{||=||=|=|=|} & & & \\ [-1ex]
			$\begin{bmatrix}
				f_n^+(z) \\ g_n^+(z)
			\end{bmatrix}$ & $\dfrac{1}{\rho_K}\begin{bmatrix}
				\left(1+\alpha_K\right)z \\ 1+\overline{\alpha_K}
			\end{bmatrix}$ & $\begin{bmatrix}
				z \\ 1
			\end{bmatrix}$ & $\dfrac{1}{\rho_{K+1}}\begin{bmatrix}
				1-\alpha_{K+1}z \\ z-\overline{\alpha_{K+1}}
			\end{bmatrix}$ \\ [3ex]
			\hline & & & \\ [-1ex]
			$\begin{bmatrix}
				p_n^+(z) \\ q_n^+(z)
			\end{bmatrix}$ & $\dfrac{1}{\rho_K}\begin{bmatrix}
				\left(-1+\alpha_K\right)z \\ 1-\overline{\alpha_K}
			\end{bmatrix}$ & $\begin{bmatrix}
				z \\ -1
			\end{bmatrix}$ & $\dfrac{1}{\rho_{K+1}}\begin{bmatrix}
				-1-\alpha_{K+1}z \\ z+\overline{\alpha_{K+1}}
			\end{bmatrix}$ \\ [3ex]
			\hline & & & \\ [-1ex]
			$\begin{bmatrix}
				f_n^-(z) \\ g_n^-(z)
			\end{bmatrix}$ & $\dfrac{1}{\rho_K}\begin{bmatrix}
				\alpha_K-z \\ z^{-1}-\overline{\alpha_K}
			\end{bmatrix}$ & $\begin{bmatrix}
				1 \\ -1
			\end{bmatrix}$ & $\dfrac{1}{\rho_{K+1}}\begin{bmatrix}
				-1-\alpha_{K+1} \\ 1+\overline{\alpha_{K+1}}
			\end{bmatrix}$ \\ [3ex]
			\hline & & & \\ [-1ex]
			$\begin{bmatrix}
				p_n^-(z) \\ q_n^-(z)
			\end{bmatrix}$ & $\dfrac{1}{\rho_K}\begin{bmatrix}
				\alpha_K+z \\ z^{-1}+\overline{\alpha_K}
			\end{bmatrix}$ & $\begin{bmatrix}
				1 \\ 1
			\end{bmatrix}$ & $\dfrac{1}{\rho_{K+1}}\begin{bmatrix}
				1-\alpha_{K+1} \\ 1-\overline{\alpha_{K+1}}
			\end{bmatrix}$ \\ [3ex]
			\hhline{|=|=|=|=|}
		\end{tabular}
		\caption{The sequences $\begin{bmatrix}
				f_n^\pm(z) \\ g_n^\pm(z)
			\end{bmatrix}$ and $\begin{bmatrix}
				p_n^\pm(z) \\ q_n^\pm(z)
			\end{bmatrix}$ for $n=K-1,K+1$, $K$ even} \label{table_GZ_cmv_even}
	\end{table}
	
	
	\begin{table}[ht]
		\centering
		\begin{tabular}{||c||c|c|c||}
			\hhline{|=|=|=|=|} & & & \\ [-2ex]
			$n$ & $K-1$ & $K$ odd & $K+1$\\ [0.5ex]
			\hhline{||=||=|=|=|} & & & \\ [-1ex]
			$\begin{bmatrix}
				f_n^+(z) \\ g_n^+(z)
			\end{bmatrix}$ & $\dfrac{1}{\rho_K}\begin{bmatrix}
				1+\overline{\alpha_K} \\ 1+\alpha_K
			\end{bmatrix}$ & $\begin{bmatrix}
				1 \\ 1
			\end{bmatrix}$ & $\dfrac{1}{\rho_{K+1}}\begin{bmatrix}
				z-\overline{\alpha_{K+1}} \\ z^{-1}-\alpha_{K+1}
			\end{bmatrix}$ \\ [3ex]
			\hline & & & \\ [-1ex]
			$\begin{bmatrix}
				p_n^+(z) \\ q_n^+(z)
			\end{bmatrix}$ & $\dfrac{1}{\rho_K}\begin{bmatrix}
				1-\overline{\alpha_K} \\ -1+\alpha_K
			\end{bmatrix}$ & $\begin{bmatrix}
				-1 \\ 1
			\end{bmatrix}$ & $\dfrac{1}{\rho_{K+1}}\begin{bmatrix}
				z+\overline{\alpha_{K+1}} \\ -z^{-1}-\alpha_{K+1}
			\end{bmatrix}$ \\ [3ex]
			\hline & & & \\ [-1ex]
			$\begin{bmatrix}
				f_n^-(z) \\ g_n^-(z)
			\end{bmatrix}$ & $\dfrac{1}{\rho_K}\begin{bmatrix}
				1-\overline{\alpha_K}z \\ \alpha_K-z
			\end{bmatrix}$ & $\begin{bmatrix}
				-z \\ 1
			\end{bmatrix}$ & $\dfrac{1}{\rho_{K+1}}\begin{bmatrix}
				\left(1+\overline{\alpha_{K+1}}\right)z \\ -1-\alpha_{K+1}
			\end{bmatrix}$ \\ [3ex]
			\hline & & & \\ [-1ex]
			$\begin{bmatrix}
				p_n^-(z) \\ q_n^-(z)
			\end{bmatrix}$ & $\dfrac{1}{\rho_K}\begin{bmatrix}
				1+\overline{\alpha_K}z \\ z+\alpha_K
			\end{bmatrix}$ & $\begin{bmatrix}
				z \\ 1
			\end{bmatrix}$ & $\dfrac{1}{\rho_{K+1}}\begin{bmatrix}
				\left(1-\overline{\alpha_{K+1}}\right)z \\ 1-\alpha_{K+1}
			\end{bmatrix}$ \\ [3ex]
			\hhline{|=|=|=|=|}
		\end{tabular}
		\caption{The sequences $\begin{bmatrix}
				f_n^\pm(z) \\ g_n^\pm(z)
			\end{bmatrix}$ and $\begin{bmatrix}
				p_n^\pm(z) \\ q_n^\pm(z)
			\end{bmatrix}$ for $n=K-1,K+1$, $K$ odd} \label{table_GZ_cmv_odd}
	\end{table}

\vfill


\section{Rotated CMV Operators}

\subsection{Rotated Orthogonal Polynomials On The Unit Circle}

The Verblunsky coefficients $\{\alpha_n\}_{n=0}^\infty$ is a sequence of complex numbers in $\mathbb{D}$, which implies that the sequence $\{\rho_n\}_{n=0}^\infty$ defined by (\ref{eq_rho_real}) are real numbers. We are interested in a slightly generalized variation of standard CMV matrices whose matrix representation looks similar to $\mathcal{C}$ as given in (\ref{eq_cmv_matrix_onesided_intro}), but instead of $\{\rho_n\}_{n=0}^\infty\subseteq[0,1]$, we now want $\{\rho_n\}_{n=0}^\infty\subseteq\overline{\mathbb{D}}$. That is, $\{\rho_n\}_{n=0}^\infty\subseteq[0,1]$ is now a sequence of complex numbers satisfying
\begin{equation}
	|\rho_n|^2+|\alpha_n|^2=1
\end{equation}
and we want our new operator, which we will call \textit{rotated CMV operators}, to have matrix representation of the form

\begin{equation}\label{eq_rcmv_matrix_onesided_intro}
	\begin{bmatrix}
		\overline{\alpha_0} & \overline{\alpha_1}\rho_0 & \rho_0\rho_1 & 0 & 0 & \cdots\\
		\overline{\rho_0} & -\overline{\alpha_1}\alpha_0 & -\alpha_0\rho_1 & 0 & 0 & \cdots \\
		0 & \overline{\alpha_2}\overline{\rho_1} & -\overline{\alpha_2}\alpha_1 & \overline{\alpha_3}\rho_2 & \rho_3\rho_2 & \cdots\\
		0 & \overline{\rho_2\rho_1} & -\alpha_1\overline{\rho_2} & -\overline{\alpha_3}\alpha_2 & -\alpha_2\rho_3 & \cdots\\
		0 & 0 & 0 & \overline{\alpha_4}\overline{\rho_3} & -\overline{\alpha_4}\alpha_3 & \cdots \\
		\vdots & \vdots & \vdots & \vdots & \vdots & \ddots
	\end{bmatrix}.
\end{equation}
For each $n\in\mathbb{N}_0$, since $\rho_n\in\overline{\mathbb{D}}$, we can write $\rho_n=|\rho_n|e^{is_n}$ for some $s_n\in(-\pi,\pi]$. Hence we obtain a new sequence $\{s_n\}_{n=0}^\infty\subseteq(-\pi,\pi]$ of arguments. By using the fact that when $\rho_n\in[0,1]$, then $s_n=0$ for all $n\in\mathbb{N}_0$ so $|\rho_n|=\rho_n$, and observing that
\begin{align*}
	\varphi_n(z)=\dfrac{\Phi_n(z)}{\norm{\Phi_n}}=\dfrac{\Phi_n(z)}{\di\prod_{j=0}^{n-1}|\rho_j|},
\end{align*}
we notice that by simply dividing the orthonormal polynomials $\varphi_n$ by $\di\prod_{j=0}^{n-1}e^{is_j}$, we effectively transform the denominator consisting of real $\rho_j$'s to complex $\rho_j$'s while preserving their moduli. This process can be simply thought of assigning an argument $s_j$ to each $|\rho_j|$, which is akin to rotating the real $\rho_j$'s by angle $s_j$. 

The process of rotating the real $\rho_n$'s in the standard CMV matrix explains why we would want to call operators with matrix representation given by (\ref{eq_rcmv_matrix_onesided_intro}) \textit{rotated CMV operators}. 

For sake of convenience, we define the sequence $\{\zeta_n\}_{n=0}^\infty\subseteq\partial\mathbb{D}$ by
\begin{equation} \label{eq_ropuc_zeta}
	\zeta_n=e^{is_n}.
\end{equation}
Following this idea, we now define the \textit{rotated polynomials} $\{\varphi_n^{\angle}\}_{n=0}^\infty$ and their reverse polynomials $\{\varphi_n^{\angle,*}\}_{n=0}^\infty$ as follows:
\begin{definition}[Rotated polynomials]
	Let $\mu$ be a nontrivial probability measure on $\partial\mathbb{D}$, \\ $\{\varphi_n(z,\mu)\}_{n=0}^\infty$ be the standard orthonormal polynomials defined by (\ref{eq_opuc_orthonormal}) and $\{\varphi_n^*(z,\mu)\}_{n=0}^\infty$ be their reverse polynomials.
	Given a sequence $\{\zeta_n\}_{n=0}^\infty\subseteq\partial\mathbb{D}$, the \textit{rotated polynomials} $\{\varphi_n^\angle\}_{n=0}^\infty$ associated to the sequence $\{\zeta_n\}_{n=0}^\infty$ are defined by
	\begin{equation} \label{eq_ropuc_def}
		\varphi_{n}^\angle(z)=\left(\prod_{j=0}^{n-1}\zeta_j\right)\varphi_n(z)=\zeta_0\zeta_1\cdots\zeta_{n-1}\varphi_{n}(z),
	\end{equation}
	and the \textit{rotated reverse polynomials} $\{\varphi_n^{\angle,*}\}_{n=0}^\infty$ are similarly defined by
	\begin{equation} \label{eq_ropuc_reverse_def}
		\varphi_{n}^{\angle,*}(z)=\left(\prod_{j=0}^{n-1}\zeta_j\right)\varphi_n^*(z)=\zeta_0\zeta_1\cdots\zeta_{n-1}\varphi_{n}^*(z).
	\end{equation}
\end{definition}
Suppose $\varphi_n^\angle(z)$ has representation $\varphi_n^\angle(z)=\di\sum_{j=0}^n c_{n,j}z^j$. Then one can show that the rotated orthonormal polynomials $\varphi_n^\angle$ are related to the monic polynomials $\Phi_n$ by
\allowdisplaybreaks
\begin{align*}
	\varphi_n^\angle(z)
	&=\left(\prod_{j=0}^{n-1}\overline{\rho_j}^{-1}\right)\Phi_n(z).
\end{align*}
Since $\Phi_n$ is monic, equating coefficients of $z^n$ gives us $$c_{n,n}=\di\prod_{j=0}^{n-1}\overline{\rho_j}^{-1}=\prod_{j=0}^{n-1}\left(\left|\rho_j\right|\zeta_j^{-1}\right)^{-1}$$
so we can obtain the monic OPUCs $\{\Phi_n(z)\}_{n=0}^\infty$ by dividing the rotated OPUCs by $\di\prod_{j=0}^{n-1}\overline{\rho_j}^{-1}$. 

Again if we consider the special case where $\zeta_n=1$ for all $n\in\mathbb{N}_0$, then $\rho_n\in[0,1]$ which gives us the canonical orthonormal polynomials $\{\varphi_n\}_{n=0}^\infty$.

Dividing the orthonormal polynomials $\varphi_n$ by $\di\prod_{j=0}^{n-1}\zeta_j^{-1}$ instead of $\di\prod_{j=0}^{n-1}\zeta_j$ is simply a matter of convention. Suppose $\{\zeta_n\}_{n=0}^\infty$ is given by (\ref{eq_ropuc_zeta}). If we view division by $\zeta_j$'s as defining $$\rho_j=|\rho_j|\zeta_j=|\rho_j|e^{is_j},$$ then division by $\zeta_j^{-1}$'s may be simply viewed as defining $$\overline{\rho_j}=|\rho_j|\zeta_j^{-1}=|\rho_j|e^{-is_j}.$$
That is, we assign the angle $-s_j$ to $|\rho_j|$ instead of $s_j$, which may be simply viewed as rotation in the opposite direction. It turns out that under our definition of $L^2$-inner product which is conjugate linear in the first argument, this is the convention we need to to adopt so that the rotated polynomials $\{\varphi_n^\angle,\varphi_n^{\angle,*}\}_{n=0}^\infty$ generate an operator with matrix representation given by (\ref{eq_rcmv_matrix_onesided_intro}) with $\overline{\rho_n}$'s in the lower triangular entries and $\rho_n$'s in the upper triangular entries. If we define
\begin{align}
	\label{eq_ropuc_diffdef}
	\varphi_{n}^\angle(z)&=\left(\prod_{j=0}^{n-1}\zeta_j^{-1}\right)\varphi_n(z),\\
	\label{eq_ropuc_reverse_diffdef}
	\varphi_{n}^{\angle,*}(z)&=\left(\prod_{j=0}^{n-1}\zeta_j^{-1}\right)\varphi_n^*(z)
\end{align}
instead, by following our subsequent construction, one can show that $\{\varphi_n^\angle,\varphi_n^{\angle,*}\}_{n=0}^\infty$ generate an operator with matrix representation
\begin{equation}
	\begin{bmatrix}
		\overline{\alpha_0} & \overline{\alpha_1}\overline{\rho_0} & \overline{\rho_0}\overline{\rho_1} & 0 & 0 & \cdots\\
		\rho_0 & -\overline{\alpha_1}\alpha_0 & -\alpha_0\overline{\rho_1} & 0 & 0 & \cdots \\
		0 & \overline{\alpha_2}\rho_1 & -\overline{\alpha_2}\alpha_1 & \overline{\alpha_3}\overline{\rho_2} & \overline{\rho_3}\overline{\rho_2} & \cdots\\
		0 & \rho_2\rho_1 & -\alpha_1\rho_2 & -\overline{\alpha_3}\alpha_2 & -\alpha_2\overline{\rho_3} & \cdots\\
		0 & 0 & 0 & \overline{\alpha_4}\rho_3 & -\overline{\alpha_4}\alpha_3 & \cdots \\
		\vdots & \vdots & \vdots & \vdots & \vdots & \ddots
	\end{bmatrix}
\end{equation}
with $\overline{\rho_n}$'s in the upper triangular entries and $\rho_n$'s in the lower triangular entries instead.

However, unlike the Szeg\H{o} duals, the rotated reverse polynomials are not obtained from the rotated polynomials by applying the anti-unitary map $R_n$ as defined in (\ref{eq_opuc_reverse_map}). They are instead related by (\ref{eq_ropuc_reverse_ropuc}) as follows:
\begin{proposition}
	The rotated reverse polynomials $\{\varphi_n^{\angle,*}\}_{n=0}^\infty$ are related to the rotated orthonormal polynomials $\{\varphi_n^\angle\}_{n=0}^\infty$ by 
	\begin{equation} \label{eq_ropuc_reverse_ropuc}
		\varphi_n^{\angle,*}(z)=\left(\prod_{j=0}^{n-1}\zeta_j^2\right)z^n\overline{\varphi_n^\angle\left(\dfrac{1}{\overline{z}}\right)}.
	\end{equation}
\end{proposition}

\begin{proof}
	Apply (\ref{eq_opuc_reverse_map}) and (\ref{eq_ropuc_def}) to (\ref{eq_ropuc_reverse_def}).
\end{proof}

If we define the map $R_n^\angle:L^2(\partial\mathbb{D},\mu)\to L^2(\partial\mathbb{D},\mu)$ associated to the sequence $\{\zeta_n\}_{n=0}^\infty$ by
\begin{equation}
	R_n^\angle(f)(z)=\left(\prod_{j=0}^{n-1}\zeta_j^2\right)z^n\overline{f\left(\dfrac{1}{\overline{z}}\right)},
\end{equation}
then one can show that $R_n^\angle$ is also anti-unitary by simply using the fact that $\zeta_n\in\partial\mathbb{D}$. Moreover, $R_n^\angle$ is also an involution, i.e. ${R_n^\angle}^2=I$, or equivalently $R_n^\angle={R_n^\angle}^{-1}$.

It is immediately clear from our definition that the rotated polynomials also form an orthogonal set, sharing the same norm as the standard orthonormal polynomials:

\begin{proposition} \label{prop_ropuc_orthonormal}
	$\{\varphi_n^\angle\}_{n=0}^\infty $ is an orthonormal set in $L^2(\partial\mathbb{D},\mu)$.
\end{proposition}

\begin{proof}
	Follows directly from the fact that $\{\varphi_n\}_{n=0}^\infty$ are orthonormal in $L^2(\partial\mathbb{D},\mu)$ and
	\[\left|\prod_{j=0}^{n-1}\zeta_j\right|=1. \qedhere\]
\end{proof}

Orthonormality hence permits us to call the rotated polynomials $\{\varphi_n\}_{n=0}^\infty$ the \textit{rotated OPUCs}. Moreover, we can show that the rotated OPUCs satisfy identical recurrence relations to \hyperref[thm_szego_recursion]{Szeg\H{o}'s recursion}:

\begin{proposition}[Forward Szeg\H{o} recursion for rotated OPUCs] \label{prop_ropuc_szego_recursion}
	Let $\mu$ be a nontrivial probability measure on $\partial\mathbb{D}$, $\{\varphi_n(z,\mu)\}_{n=1}^\infty$ be the standard orthonormal OPUCs given by (\ref{eq_opuc_orthonormal}), $\{\varphi_n^\angle(z,\mu)\}$ be the rotated OPUCs associated to the sequence $\{\zeta_n\}_{n=0}^\infty\subseteq\partial\mathbb{D}$, $\{\varphi_n^*(z,\mu)\}_{n=1}^\infty$ and $\{\varphi_n^{\angle,*}(z,\mu)\}_{n=1}^\infty$ be their reverse polynomials. Given the Verblunsky coefficients $\{\alpha_n\}_{n=0}^\infty\subseteq\mathbb{D}$ associated to the standard OPUCs $\{\varphi_n(z,\mu)\}$ (and $\{\varphi_n^*(z,\mu)\}$), define the sequences $\{\left|\rho_n\right|\}_{n=0}^\infty\subseteq\mathbb{R}$ and $\{\rho_n\}_{n=0}^\infty\subseteq\mathbb{D}$ by 
	\begin{gather}
		\left|\rho_n\right|=\left(1-\left|\alpha_n\right|^2\right)^{\frac{1}{2}},\\
		\label{eq_rho_complex}
		\rho_n=\left|\rho_n\right|\zeta_n.
	\end{gather}
	Then
	\begin{align} 
		\label{eq_szego_recursion_ropuc}
		\varphi_{n+1}^\angle(z)&=\overline{\rho_n}^{-1}\left(z\varphi_n^\angle(z)-\overline{\alpha_n}\varphi_n^{\angle,*}(z)\right),\\
		\label{eq_szego_recursion_ropuc_reverse}
		\varphi_{n+1}^{\angle,*}(z)&=\overline{\rho_n}^{-1}\left(\varphi_n^{\angle,*}(z)-\alpha_nz\varphi_n^{\angle}(z)\right).
	\end{align}
\end{proposition}

\begin{proof}
	\allowdisplaybreaks
	To obtain (\ref{eq_szego_recursion_ropuc}), we multiply both sides of (\ref{eq_szegorecursion_orthonormal}) by $\di\prod_{j=0}^{n-1}\zeta_j$, or equivalently, apply (\ref{eq_szegorecursion_orthonormal}) to (\ref{eq_ropuc_def}).
	Equation (\ref{eq_szego_recursion_ropuc_reverse}) is obtained similarly by applying (\ref{eq_szegorecursion_orthonormal_reverse}) to (\ref{eq_ropuc_reverse_def}).
\end{proof}

In fact, (\ref{eq_szego_recursion_ropuc}) and (\ref{eq_szego_recursion_ropuc_reverse}) are obtained from (\ref{eq_szegorecursion_orthonormal}) and (\ref{eq_szegorecursion_orthonormal_reverse}) by simply replacing the real $\rho_n=|\rho_n|$'s by $|\rho_n|\zeta_n^{-1}$, which become complex $\rho_n$'s as defined in (\ref{eq_rho_complex}). Moreover, if $\rho_n\in[0,1]\subseteq\mathbb{R}$, then we obtain back exactly the original forward Szeg\H{o} recursion given by Theorem \ref{thm_szego_recursion}. Since the sequence $\{\zeta_n\}_{n=0}^\infty\subseteq\partial\mathbb{D}$ is allowed to be arbitrary, by simply dividing $\varphi_n(z)$ by the complex number $\di\prod_{j=0}^{n-1}\zeta_j^{-1}\in\partial\mathbb{D}$, we have effectively transformed the real $\rho_j$'s from (\ref{eq_rho_real}) into the complex numbers $\overline{\rho_j}$'s, as we wished.

Again we can rewrite the forward Szeg\H{o} recursion in terms of the \textit{rotated Szeg\H{o} transfer matrix} $A_n^\angle(z)$ by
\begin{equation} \label{eq_szego_recursion_ropuc_transition_matrix}
	A_n^\angle(z)=\overline{\rho_n}^{-1}\begin{bmatrix}
		z & -\overline{\alpha_n} \\ -\alpha_n z & 1
	\end{bmatrix}.
\end{equation} 
Then (\ref{eq_szego_recursion_ropuc}) and (\ref{eq_szego_recursion_ropuc_reverse}) become
\begin{equation} \label{eq_szego_recursion_ropuc_matrix_form}
	\begin{bmatrix}
		\varphi_{n+1}^\angle(z) \\ \varphi_{n+1}^{\angle,*}(z)
	\end{bmatrix}=\overline{\rho_n}^{-1}\begin{bmatrix}
		z & -\overline{\alpha_n} \\ -\alpha_n z & 1
	\end{bmatrix}\begin{bmatrix}
		\varphi_{n}^\angle(z) \\ \varphi_{n}^{\angle,*}(z)
	\end{bmatrix}=A_n^\angle(z)\begin{bmatrix}
		\varphi_{n}^\angle(z) \\ \varphi_{n}^{\angle,*}(z)
	\end{bmatrix}.
\end{equation}

Again since the Verblunsky coefficients $\{\alpha_n(\mu)\}_{n=0}^\infty$ associated to the probability measure $\mu$ satisfy $|\alpha_n(\mu)|\neq 1$ for all $n\in\mathbb{N}_0$, then $A_n^\angle$ is invertible and we can obtain a rotated reverse forward Szeg\H{o} recursion:

\begin{corollary}[Backward Szeg\H{o} recursion for rotated OPUCs] \label{prop_ropuc_inverse_szego_recursion}
	Let $\mu$ be a nontrivial probability measure on $\partial\mathbb{D}$ such that $\{\alpha_n(\mu)\}_{n=0}^\infty\subseteq\mathbb{D}$. If $\{\varphi_n^\angle(z,\mu)\}_{n=0}^\infty$ are the rotated OPUCs associated to the sequence $\{\zeta_n\}_{n=0}^\infty\subseteq\partial\mathbb{D}$, $\{\varphi_n^{\angle,*}(z,\mu)\}_{n=0}^\infty$ be their reverse polynomials, and $\{\rho_n\}_{n=0}^\infty$ be defined by (\ref{eq_rho_complex}). Then
	\begin{align} 
		\label{eq_inverse_szego_recursion_ropuc}
		z\varphi_{n}^\angle(z)&=\rho_n^{-1}\left(\varphi_{n+1}^\angle(z)+\overline{\alpha_n}\varphi_{n+1}^{\angle,*}(z)\right),\\
		\label{eq_inverse_szego_recursion_ropuc_reverse}
		\varphi_{n}^{\angle,*}(z)&=\rho_n^{-1}\left(\varphi_{n+1}^{\angle,*}(z)+\alpha_n\varphi_{n+1}^{\angle}(z)\right).
	\end{align}
\end{corollary}

\begin{proof}
	Let $A_n^\angle(z)$ be defined by (\ref{eq_szego_recursion_ropuc_transition_matrix}). Then
	\[\det A_n^\angle(z)= \overline{\rho_n}^{-1}(z-z|\alpha_n|^2)=\overline{\rho_n}^{-1}z|\rho_n|^2=z\rho_n.\]
	Since $z\in\partial\mathbb{D}$, we have $z\neq 0$. Moreover $\alpha_n\in\mathbb{D}$ implies that $\rho_n\neq 0$, so $A_n^\angle$ is invertible. Hence we can compute
	\begin{equation}
		\left(A_n^\angle\right)^{-1}(z)=\dfrac{1}{z\rho_n}
		\begin{bmatrix}
			1 & \overline{\alpha_n} \\ \alpha_n z & z
		\end{bmatrix}.
	\end{equation}
	Left multiplication by $\left(A_n^\angle\right)^{-1}(z)$ to both sides of (\ref{eq_szego_recursion_ropuc_matrix_form}) then gives
	\begin{equation}
		\begin{bmatrix}
			\varphi_{n}^\angle(z) \\ \varphi_{n}^{\angle,*}(z)
		\end{bmatrix}=\left(A_n^\angle\right)^{-1}(z)\begin{bmatrix}
			\varphi_{n+1}^\angle(z) \\ \varphi_{n+1}^{\angle,*}(z)
		\end{bmatrix}=\dfrac{1}{z\rho_n}\begin{bmatrix}
			1 & \overline{\alpha_n} \\ \alpha_nz & z
		\end{bmatrix}\begin{bmatrix}
			\varphi_{n}^\angle(z) \\ \varphi_{n}^{\angle,*}(z)
		\end{bmatrix},
	\end{equation}
	which can be recovered into (\ref{eq_inverse_szego_recursion_ropuc}) and (\ref{eq_inverse_szego_recursion_ropuc_reverse}).
\end{proof}

It is sometimes useful to rewrite the forward Szeg\H{o} recursion (\ref{eq_szego_recursion_ropuc}) and (\ref{eq_szego_recursion_ropuc_reverse}) as 
\begin{align}
	\label{eq_szego_recursion_ropuc_rewrite}
	z\varphi_n^\angle(z)&=\overline{\rho_n}\varphi_{n+1}^\angle(z)+\overline{\alpha_n}\varphi_n^{\angle,*}(z),\\
	\label{eq_szego_recursion_ropuc_reverse_rewrite}
	\varphi_n^{\angle,*}(z)&=\overline{\rho_n}\varphi_{n+1}^{\angle,*}(z)+\alpha_nz\varphi_n^\angle(z),
\end{align}
and similarly the backward Szeg\H{o} recursion (\ref{eq_inverse_szego_recursion_ropuc}) and (\ref{eq_inverse_szego_recursion_ropuc_reverse}) as
\begin{align}
	\label{eq_inverse_szego_recursion_ropuc_rewrite}
	\varphi_{n+1}^\angle(z)&=\rho_nz\varphi_n^\angle(z)-\overline{\alpha_n}\varphi_{n+1}^{\angle,*}(z),\\
	\label{eq_inverse_szego_recursion_ropuc_reverse_rewrite}
	\varphi_{n+1}^{\angle,*}(z)&=\rho_n\varphi_n^{\angle,*}(z)-\alpha_n\varphi_{n+1}^\angle(z).
\end{align}

By a trivial computation, one can show that the Christoffel-Darboux formula is invariant under rotation of polynomials:
\begin{align} 
	\label{eq_ropuc_cd_opuc_cd}
	\sum_{k=0}^n \overline{\varphi_k^\angle(\xi)}\varphi_k^\angle(z)
	&=\sum_{k=0}^n \overline{\varphi_k(\xi)}\varphi_k(z)\\
	\label{eq_opuc_cd_ropuc}
	&=\dfrac{\overline{\varphi_{n+1}^*(\xi)}\varphi_{n+1}^*(z)-\overline{\varphi_{n+1}(\xi)}\varphi_{n+1}(z)}{1-\overline{\xi}z}\\
	\label{eq_ropuc_cd}
	&=\dfrac{\overline{\varphi_{n+1}^{\angle,*}(\xi)}\varphi_{n+1}^{\angle,*}(z)-\overline{\varphi_{n+1}^\angle(\xi)}\varphi_{n+1}^\angle(z)}{1-\overline{\xi}z}.
\end{align}
In other words, there is no \say{rotated Christoffel-Darboux formula} - the Christoffel-Darboux formula for standard and rotated OPUCs are identical.

\subsection{Inner Products Involving Rotated OPUCs}

We compute a selection of inner products involving the rotated OPUCs $\{\varphi_n^\angle\}_{n=0}^\infty$ and their reverses $\{\varphi_n^{\angle,*}\}_{n=0}^\infty$ which are involved in the subsequent computations the entries of rotated CMV matrices. In particular, we follow a similar approach to Simon's presentation in Section 1.5 of \cite{Simon2005a} - by proving rotated analogues of Propositions 1.5.8, 1.5.9 and 1.5.10 in \cite{Simon2005a}.

\begin{proposition} \label{prop_ropuc_ip_general}
	Let $\mu$ be a nontrivial probability measure on $\partial\mathbb{D}$, $\{\varphi_n^\angle(z,\mu)\}_{n=0}^\infty$ be the rotated OPUCs associated to the sequence $\{\zeta_n\}_{n=0}^\infty$ as defined in (\ref{eq_ropuc_def}), $\{\alpha_n\}_{n=0}^\infty$ be their Verblunsky coefficients, and the sequence $\{\rho_n\}_{n=0}^\infty$ be defined by (\ref{eq_rho_complex}). Then 
	\begin{align}
		\label{eq_ropuc_ip_jk}
		\ip{\varphi^\angle_j}{\varphi^\angle_k}&=\delta_{jk},\\
		\label{eq_ropuc_ip_j*k}
		\ip{\varphi_j^{\angle,*}}{\varphi_k^\angle}
		&=\begin{cases}
			-\overline{\alpha_{k-1}}\di\prod_{\ell=k}^{j-1}\overline{\rho_\ell} & 0\leq k\leq j-1;\\
			-\overline{\alpha_{k-1}} & k=j; \\ 0 & k\geq j+1,
		\end{cases}\\
		\label{eq_ropuc_ip_jk*}
		\ip{\varphi_j^\angle}{\varphi_k^{\angle,*}}
		&=\begin{cases}
			-\alpha_{j-1}\di\prod_{\ell=j}^{k-1}\rho_\ell & 0\leq j\leq k-1;\\	
			-\alpha_{j-1} & j=k; \\ 0 & j\geq k+1,
		\end{cases}\\
		\label{eq_ropuc_ip_j*k*}
		\ip{\varphi_j^{\angle,*}}{\varphi_k^{\angle,*}}
		&=\begin{cases}
			\di\prod_{\ell=k}^{j-1}\overline{\rho_\ell} & k\leq j-1;\\
			1 & k=j;\\
			\di\prod_{\ell=j}^{k-1}\rho_\ell & k\geq j+1.\\
		\end{cases}
	\end{align}
\end{proposition}

\begin{proof}
	Equation (\ref{eq_ropuc_ip_jk}) is simply a restatement of Proposition \ref{prop_ropuc_orthonormal}.
	For (\ref{eq_ropuc_ip_j*k}), the case $k\geq j+1$ is a direct consequence of the fact that $\varphi_k^\angle$ is orthogonal to any polynomial $P$ with $\deg P\leq k$. If $k\leq j$, then for any polynomial $P$ with $\deg P\leq j$, we have
	\begin{align*}
		\ip{\Phi_j^*(z)}{P(z)-P(0)}=0.
	\end{align*}
	By Theorem \ref{thm_opuc_reverse_orthogonal_to_z_z^n}, we have
	\begin{equation}
		\ip{\Phi_j^*}{P}=\ip{\Phi_j^*}{P(0)}=P(0)\norm{\Phi_j}^2=P(0)\prod_{\ell=0}^{j-1}\left|\rho_\ell\right|^2=P(0)\prod_{\ell=0}^{j-1}\overline{\rho_\ell}\rho_\ell
	\end{equation}
	and hence
	\begin{equation} \label{eq_ropuc_reverse_ip_proof}
		\ip{\varphi_j^{\angle,*}}{P}
		=\di\overline{\prod_{\ell=0}^{j-1}\overline{\rho_\ell}^{-1}}\ip{\Phi_j^*}{P}
		=\prod_{\ell=0}^{j-1}\rho_\ell^{-1}\left(P(0)\prod_{\ell=0}^{j-1}\overline{\rho_\ell}\rho_\ell\right)
		=P(0)\prod_{\ell=0}^{j-1}\overline{\rho_\ell}.
	\end{equation}
	Putting $P(z)=\varphi_k^\angle(z)$ then gives
	\begin{align*}
		\ip{\varphi_j^{\angle,*}}{\varphi_k^\angle}
		=\varphi_k^\angle(0)\prod_{\ell=0}^{j-1}\overline{\rho_\ell}
		=\left(\prod_{\ell=0}^{k-1}\overline{\rho_\ell}^{-1}\Phi_k(0)\right)\prod_{\ell=0}^{j-1}\overline{\rho_\ell}
		=\begin{cases}
			-\overline{\alpha_{k-1}}\di\prod_{\ell=k}^{j-1}\overline{\rho_\ell} & k\leq j-1;\\
			-\overline{\alpha_{k-1}} & k=j.
		\end{cases}
	\end{align*}
	
	Equation (\ref{eq_ropuc_ip_jk*}) is simply the complex conjugate of (\ref{eq_ropuc_ip_j*k}).
	
	To prove (\ref{eq_ropuc_ip_j*k*}), we first suppose $k\leq j$. Putting $P(z)=\varphi_k^{\angle,*}(z)$ into (\ref{eq_ropuc_reverse_ip_proof}) gives
	\begin{align*}
		\ip{\varphi_j^{\angle,*}}{\varphi_k^{\angle,*}}
		=\varphi_k^{\angle,*}(0)\prod_{\ell=0}^{j-1}\overline{\rho_\ell}
		=\left(\prod_{\ell=0}^{k-1}\overline{\rho_\ell}^{-1}\Phi_k^*(0)\right)\prod_{\ell=0}^{j-1}\overline{\rho_\ell}
		=\begin{cases}
			\di\prod_{\ell=k}^{j-1}\overline{\rho_\ell} & k\leq j-1;\\
			1 & k=j.
		\end{cases}
	\end{align*}
	If $k>j$, then
	\[\ip{\varphi_j^{\angle,*}}{\varphi_k^{\angle,*}}
	=\overline{\ip{\varphi_k^{\angle,*}}{\varphi_j^{\angle,*}}}
	=\prod_{\ell=j}^{k-1}\rho_\ell.\qedhere\]
\end{proof}

We observe that the above inner products between the rotated OPUCs are essentially identical to those for the standard orthonormal polynomials respectively (cf. \cite[Proposition 1.5.8]{Simon2005a}). The differences observed in (\ref{eq_ropuc_ip_j*k}) and (\ref{eq_ropuc_ip_j*k*}) are simply down to the single fact that when $\rho_n\in\mathbb{R}$, we have $\overline{\rho_n}=\rho_n$, but this is generally not true if $\rho_n\in\mathbb{C}$.

\begin{proposition}
	Let $\mu$ be a nontrivial probability measure on $\partial\mathbb{D}$, $\{\varphi_n^\angle(z,\mu)\}_{n=0}^\infty$ be the rotated OPUCs associated to the sequence $\{\zeta_n\}_{n=0}^\infty$ as defined in (\ref{eq_ropuc_def}), $\{\alpha_n\}_{n=0}^\infty$ be their Verblunsky coefficients, and the sequence $\{\rho_n\}_{n=0}^\infty$ be defined by (\ref{eq_rho_complex}). Then 
	\begin{align}
		\label{eq_ropuc_ip_j_z.k}
		\ip{\varphi_j^\angle}{z\varphi_k^\angle}
		&=\begin{cases}
			-\overline{\alpha_k}\alpha_{j-1}\di\prod_{\ell=j}^{k-1}\rho_\ell & 0\leq j\leq k-1;\\
			-\overline{\alpha_k}\alpha_{j-1} & j=k; \\ \overline{\rho_k} & j=k+1; \\ 0 & j\geq k+2,
		\end{cases}\\
		\label{eq_ropuc_ip_j*_z.k}
		\ip{\varphi_j^{\angle,*}}{z\varphi_k^\angle}
		&=\begin{cases}
			\overline{\alpha_k}\di\prod_{\ell=j}^{k-1}\rho_\ell & 0\leq j\leq k-1;\\
			\overline{\alpha_k} & j=k; \\ 0 & j\geq k+1.
		\end{cases}
	\end{align}
\end{proposition}

\begin{proof}
	Taking inner product of $\varphi_j^\angle$ with both sides of (\ref{eq_szego_recursion_ropuc_rewrite}) with $n=k$ gives
	\begin{equation} \label{eq_ropuc_ip_j_z.k_proof}
		\ip{\varphi_j^\angle}{z\varphi_k^\angle}=\overline{\rho_k}\ip{\varphi_j^\angle}{\varphi_{k+1}^\angle}+\overline{\alpha_k}\ip{\varphi_j^\angle}{\varphi_k^{\angle,*}}.
	\end{equation}
	By (\ref{eq_ropuc_ip_jk}), the first term on the right hand side of (\ref{eq_ropuc_ip_j_z.k_proof}) always vanishes except when $j=k+1$, where $\ip{\varphi_j^\angle}{\varphi_{k+1}}=1$. Using this result and substituting (\ref{eq_ropuc_ip_jk*}) into (\ref{eq_ropuc_ip_j_z.k_proof}) then gives (\ref{eq_ropuc_ip_j_z.k}).
	On the other hand, taking inner product of $\varphi_j^{\angle,*}$ with both sides of (\ref{eq_szego_recursion_ropuc_rewrite}) gives
	\begin{equation} \label{eq_ropuc_ip_j*_z.k_proof}
		\ip{\varphi_j^{\angle,*}}{z\varphi_k^\angle}=\overline{\rho_k}\ip{\varphi_j^{\angle,*}}{\varphi_{k+1}^\angle}+\overline{\alpha_k}\ip{\varphi_j^{\angle,*}}{\varphi_k^{\angle,*}}.
	\end{equation}
	If $j\geq k+1$, then Theorem \ref{thm_opuc_reverse_orthogonal_to_z_z^n} implies that $\varphi_j^{\angle,*}$ is orthogonal to any polynomial $P(z)$ with $1\leq\deg P(z)\leq j$, thus we have $\ip{\varphi_j^{\angle,*}}{z\varphi_k^\angle}=0$. If $j\leq k$, then using the fact that $\varphi_{k+1}$ (and hence $\varphi_{k+1}^\angle$) is orthogonal to $\{1,z,z^2,\cdots,z^{n-1}\}$ implies that $\ip{\varphi_j^{\angle,*}}{\varphi_{k+1}^\angle}=0$. Substituting (\ref{eq_ropuc_ip_j*k*}) into (\ref{eq_ropuc_ip_j*_z.k_proof}) together with these results then gives (\ref{eq_ropuc_ip_j*_z.k}).
\end{proof}

Now we are ready to compute the eight special cases of inner products which constitute the entries of our rotated CMV matrices:

\begin{proposition}
	Let $\mu$ be a nontrivial probability measure on $\partial\mathbb{D}$, $\{\varphi_n^\angle(z,\mu)\}_{n=0}^\infty$ be the rotated OPUCs associated to the sequence $\{\zeta_n\}_{n=0}^\infty$ as defined in (\ref{eq_ropuc_def}), $\{\alpha_n\}_{n=0}^\infty$ be their Verblunsky coefficients, and the sequence $\{\rho_n\}_{n=0}^\infty$ be defined by (\ref{eq_rho_complex}). Then 
	\begin{align}
		\label{eq_ropuc_ip_2n+1_z.2n+1}
		\ip{\varphi_{2n+1}^\angle}{z\varphi_{2n+1}^\angle}&=-\overline{\alpha_{2n+1}}\alpha_{2n},\\
		\label{eq_ropuc_ip_2n+3_z2.2n+1}
		\ip{\varphi_{2n+3}^\angle}{z^2\varphi_{2n+1}^\angle}&=\overline{\rho_{2n+1}\rho_{2n+2}},\\
		\label{eq_ropuc_ip_2n*_z.2n+1}
		\ip{\varphi_{2n}^{\angle,*}}{z\varphi_{2n+1}^\angle}&=\overline{\alpha_{2n+1}}\overline{\rho_{2n}},\\
		\label{eq_ropuc_ip_2n+2*_z2.2n+1}
		\ip{\varphi_{2n+2}^{\angle,*}}{z^2\varphi_{2n+1}^\angle}&=\overline{\alpha_{2n+2}}\overline{\rho_{2n+1}},\\
		\label{eq_ropuc_ip_2n*_z.2n*}
		\ip{\varphi_{2n}^{\angle,*}}{z\varphi_{2n}^{\angle,*}}&=-\overline{\alpha_{2n}}\alpha_{2n-1},\\
		\label{eq_ropuc_ip_2n-2*_2n*}
		\ip{\varphi_{2n-2}^{\angle,*}}{\varphi_{2n}^{\angle,*}}&=\rho_{2n-2}\rho_{2n-1},\\
		\label{eq_ropuc_ip_2n+1_z.2n*}
		\ip{\varphi_{2n+1}^\angle}{z\varphi_{2n}^{\angle,*}}&=-\alpha_{2n-1}\overline{\rho_{2n}},\\
		\label{eq_ropuc_ip_2n-1_2n*}
		\ip{\varphi_{2n-1}^\angle}{\varphi_{2n}^{\angle,*}}&=-\alpha_{2n-2}\rho_{2n-1}.
	\end{align}
\end{proposition}

\begin{proof}
	\allowdisplaybreaks
	Equation (\ref{eq_ropuc_ip_2n+1_z.2n+1}) is a special case of (\ref{eq_ropuc_ip_j_z.k}). To prove (\ref{eq_ropuc_ip_2n+3_z2.2n+1}), we can write
	\begin{align*}
		\ip{\varphi_{2n+3}^\angle}{z^2\varphi_{2n+1}^\angle}
		&=\int_{\partial\mathbb{D}} \overline{\varphi_{2n+3}^\angle(z)}\cdot z^2\varphi_{2n+1}^\angle(z)\,d\mu(z)\\
		&=\prod_{j=2n+1}^{2n+2} \zeta_j^{-2}\int_{\partial\mathbb{D}}\overline{\prod_{j=0}^{2n}\zeta_j^{2} z^{2n+1}\overline{\varphi_{2n+1}^\angle(z)}}\cdot \prod_{j=0}^{2n+2}\zeta_j^{2}z^{2n+3}\overline{\varphi_{2n+3}^\angle(z)}\,d\mu(z)\\
		&=\zeta_{2n+1}^2\zeta_{2n+2}^2\ip{\varphi_{2n+1}^{\angle,*}(z)}{\varphi_{2n+3}^{\angle,*}(z)}\\
		&=\zeta_{2n+1}^{-2}\zeta_{2n+2}^{-2}\overline{\rho_{2n+1}}\overline{\rho_{2n+2}}\\
		&=\overline{\rho_{2n+1}\rho_{2n+2}}
	\end{align*}
	so (\ref{eq_ropuc_ip_2n+3_z2.2n+1}) is a consequence of (\ref{eq_ropuc_ip_j*k*}). Equation (\ref{eq_ropuc_ip_2n*_z.2n+1}) is a special case of (\ref{eq_ropuc_ip_j*_z.k}). To prove (\ref{eq_ropuc_ip_2n+2*_z2.2n+1}), multiplying both sides of (\ref{eq_szego_recursion_ropuc_rewrite}) by $z$ and replacing $n$ with $2n+1$ gives
	\[z^2\varphi_{2n+1}^\angle(z)=\overline{\rho_{2n+1}}z\varphi_{2n+2}^\angle(z)+\overline{\alpha_{2n+1}}z\varphi_{2n+1}^{\angle,*}(z)\] and hence
	\[\ip{\varphi_{2n+2}^{\angle,*}}{z^2\varphi_{2n+1}^\angle}=\overline{\rho_{2n+1}}\ip{\varphi_{2n+2}^{\angle,*}}{z\varphi_{2n+2}^\angle}+\overline{\alpha_{2n+1}}\ip{\varphi_{2n+2}^{\angle,*}}{z\varphi_{2n+1}^{\angle,*}}.\]
	Since $\varphi_{2n+2}^{\angle,*}$ is orthogonal to $\{z,z^2,\cdots,z^{2n+2}\}$, by Theorem \ref{thm_opuc_reverse_orthogonal_to_z_z^n}, we have $\ip{\varphi_{2n+2}^{\angle,*}}{z\varphi_{2n+1}^{\angle,*}}=0$. Alternatively, we can also prove $\ip{\varphi_{2n+2}^{\angle,*}}{z\varphi_{2n+1}^{\angle,*}}=0$ by applying (\ref{eq_ropuc_reverse_ip_proof}) with $j=2n+2$ and $P=z\varphi_{2n+1}^{\angle,*}$. Thus (\ref{eq_ropuc_ip_2n+2*_z2.2n+1}) follows from (\ref{eq_ropuc_ip_j*_z.k}).
	By anti-unitarity of $R_{2n}^\angle$, we can write 
	\[\ip{\varphi_{2n}^{\angle,*}}{z\varphi_{2n}^{\angle,*}}=\ip{\varphi_{2n}^\angle}{z\varphi_{2n}^\angle}\]
	so (\ref{eq_ropuc_ip_2n*_z.2n*}) follows from (\ref{eq_ropuc_ip_j_z.k}). Equation (\ref{eq_ropuc_ip_2n-2*_2n*}) is a special case of (\ref{eq_ropuc_ip_j*k*}). One may show that (\ref{eq_ropuc_ip_2n+1_z.2n*}) follows from (\ref{eq_ropuc_ip_jk*}) by a similar technique used in the proof of (\ref{eq_ropuc_ip_2n+3_z2.2n+1}).
	Equation (\ref{eq_ropuc_ip_2n-1_2n*}) is a special case of (\ref{eq_ropuc_ip_jk*}).
\end{proof}

\subsection{Rotated CMV Operators}

Following a similar idea to the definition of the rotated OPUCs, given a sequence $\{\zeta_n\}_{n=0}^\infty\subseteq\partial\mathbb{D}$, we also adopt a similar definition to define the \textit{rotated CMV basis} as follows:

\begin{definition}[Rotated CMV Basis and Operator] \label{def_rcmv}
	Let $\mu$ be a nontrivial probability measure on $\partial\mathbb{D}$ and $\{\chi_n(z,\mu)\}_{n=0}^\infty$ be the standard CMV basis associated to $\mu$. The \textit{rotated CMV basis} $\{\chi_n^\angle(z,\mu)\}_{n=0}^\infty$ associated to the sequence $\{\zeta_n\}_{n=0}^\infty\subseteq\partial\mathbb{D}$ is defined by
	\begin{equation} \label{eq_rcmv_basis_def}
		\chi_n^\angle(z)=\left(\prod_{j=0}^{n-1}\zeta_j\right)\chi_n(z).
	\end{equation}
	The \textit{rotated CMV matrix} $\mathcal{C^\angle}\left(\mu,\{\zeta_n\}_{n=0}^\infty\right)=[\mathcal{C}_{k\ell}^\angle]$ associated to $\{\zeta_n\}_{n=0}^\infty\subseteq\partial\mathbb{D}$ is defined by
	\begin{equation} \label{eq_rcmv_matrix_def}
		\mathcal{C}_{k\ell}^\angle=\ip{\chi_k^\angle}{z\chi_\ell^\angle}.
	\end{equation}
\end{definition}

We can show that the relationship between the rotated CMV basis and the rotated OPUCs is analogous to (\ref{eq_cmv_basis_opuc_odd}) and (\ref{eq_cmv_basis_opuc_even}) in the following way:

\begin{proposition} \label{prop_rcmv_basis_ropuc}
	Let $\mu$ be a nontrivial probability measure on $\partial\mathbb{D}$, $\{\varphi_n^\angle(z,\mu)\}_{n=0}^\infty$ be the rotated OPUCs as defined in (\ref{eq_ropuc_def}) and $\{\chi_n^\angle(z,\mu)\}_{n=0}^\infty$ be the rotated CMV basis defined in (\ref{eq_rcmv_basis_def}). Then
	\begin{align}
		\label{eq_rcmv_basis_ropuc_odd}
		\chi_{2n-1}^\angle(z)&=z^{-n+1}\varphi_{2n-1}^\angle(z),\\
		\label{eq_rcmv_basis_ropuc_even}
		\chi_{2n}^\angle(z)&=z^{-n}\varphi_{2n}^{\angle,*}(z).
	\end{align}
\end{proposition}

\begin{proof}
	We prove (\ref{eq_rcmv_basis_ropuc_odd}) by simply invoking (\ref{eq_rcmv_basis_def}) and applying (\ref{eq_cmv_basis_opuc_odd}) and (\ref{eq_ropuc_def}).
	The proof of (\ref{eq_rcmv_basis_ropuc_even}) follows a similar argument by applying (\ref{eq_cmv_basis_opuc_even}) and (\ref{eq_ropuc_reverse_def}).
\end{proof}

The nonzero entries of the rotated CMV matrix $\mathcal{C}^\angle\left(\mu,\{\zeta_n\}_{n=0}^\infty\right)$ are computed as follows:

\begin{proposition} \label{prop_rcmv_basis_ip}
	Let $\mu$ be a nontrivial probability measure on $\partial\mathbb{D}$, $\{\alpha_n(\mu)\}_{n=0}^\infty$ be the Verblunsky coefficients associated to $\mu$, and the sequence $\{\rho_n\}_{n=0}^\infty$ be defined by (\ref{eq_rho_complex}). Then the nonzero entries of $\mathcal{C}^\angle\left(\mu,\{\zeta_n\}_{n=0}^\infty\right)$ are given by
	\begin{align}
		\label{eq_rcmv_basis_ip_2n-2_z*2n}
		\ip{\chi_{2n-2}^\angle}{z\chi_{2n}^\angle}&=\rho_{2n-2}\rho_{2n-1},\\
		\label{eq_rcmv_basis_ip_2n-1_z*2n}
		\ip{\chi_{2n-1}^\angle}{z\chi_{2n}^\angle}&=-\alpha_{2n-2}\rho_{2n-1},\\
		\label{eq_rcmv_basis_ip_2n_z*2n}
		\ip{\chi_{2n}^\angle}{z\chi_{2n}^\angle}&=-\overline{\alpha_{2n}}\alpha_{2n-1},\\
		\label{eq_rcmv_basis_ip_2n+1_z*2n}
		\ip{\chi_{2n+1}^\angle}{z\chi_{2n}^\angle}&=-\alpha_{2n-1}\overline{\rho_{2n}},\\
		\label{eq_rcmv_basis_ip_2n-2_z*2n-1}
		\ip{\chi_{2n-2}^\angle}{z\chi_{2n-1}^\angle}&=\overline{\alpha_{2n-1}}\rho_{2n-2},\\
		\label{eq_rcmv_basis_ip_2n-1_z*2n-1}
		\ip{\chi_{2n-1}^\angle}{z\chi_{2n-1}^\angle}&=-\overline{\alpha_{2n-1}}\alpha_{2n-2},\\
		\label{eq_rcmv_basis_ip_2n_z*2n-1}
		\ip{\chi_{2n}^\angle}{z\chi_{2n-1}^\angle}&=\overline{\alpha_{2n}}\rho_{2n-1},\\
		\label{eq_rcmv_basis_ip_2n+1_z*2n-1}
		\ip{\chi_{2n+1}^\angle}{z\chi_{2n-1}^\angle}&=\overline{\rho_{2n}\rho_{2n-1}}.
	\end{align}
\end{proposition}

\begin{proof}
	We observe that (\ref{eq_rcmv_basis_ip_2n-2_z*2n}) is exactly (\ref{eq_ropuc_ip_2n-2*_2n*}) by
	$$\ip{\chi_{2n-2}^\angle}{z\chi_{2n}^\angle}=\ip{z^{-(n-1)}\varphi_{2n-2}^{\angle,*}}{z\cdot z^{-n}\varphi_{2n}^{\angle,*}}=\ip{\varphi_{2n-2}^{\angle,*}}{\varphi_{2n}^{\angle,*}}=\rho_{2n-2}\rho_{2n-1}.$$
	Similarly, we can show that 
	\begin{align*}
		\ip{\chi_{2n-1}^\angle}{z\chi_{2n}^\angle}
		&=\ip{z^{-n+1}\varphi_{2n-1}^\angle}{z\cdot z^{-n}\varphi_{2n}^{\angle,*}}
		=\ip{\varphi_{2n-1}^\angle}{\varphi_{2n}^{\angle,*}}
		=-\alpha_{2n-2}\rho_{2n-1},\\
		\ip{\chi_{2n}^\angle}{z\chi_{2n}^\angle}
		&=\ip{z^{-n}\varphi_{2n}^{\angle,*}}{z\cdot z^{-n}\varphi_{2n}^{\angle,*}}
		=\ip{\varphi_{2n}^{\angle,*}}{z\varphi_{2n}^{\angle,*}}
		=-\overline{\alpha_{2n}}\alpha_{2n-1},\\
		\ip{\chi_{2n+1}^\angle}{z\chi_{2n}^\angle}
		&=\ip{z^{-(n+1)+1}\varphi_{2n+1}^\angle}{z\cdot z^{-n}\varphi_{2n}^{\angle,*}}
		=\ip{\varphi_{2n+1}^\angle}{z\varphi_{2n}^{\angle,*}}
		=-\alpha_{2n-1}\overline{\rho_{2n}},
	\end{align*}
	so (\ref{eq_rcmv_basis_ip_2n-1_z*2n}), (\ref{eq_rcmv_basis_ip_2n_z*2n}) and (\ref{eq_rcmv_basis_ip_2n+1_z*2n}) are reduced to exactly (\ref{eq_ropuc_ip_2n-1_2n*}), (\ref{eq_ropuc_ip_2n*_z.2n*}) and (\ref{eq_ropuc_ip_2n+1_z.2n*}), respectively. On the other hand, 
	\begin{align*}
		\ip{\chi_{2n-2}^\angle}{z\chi_{2n-1}^\angle}
		&=\ip{z^{-(n-1)}\varphi_{2n-2}^{\angle,*}}{z\cdot z^{-n+1}\varphi_{2n-1}^\angle}
		=\ip{\varphi_{2n-2}^{\angle,*}}{z\varphi_{2n-1}^\angle}
		=\overline{\alpha_{2n-1}\rho_{2n-2}},\\
		\ip{\chi_{2n-1}^\angle}{z\chi_{2n-1}^\angle}
		&=\ip{z^{-n+1}\varphi_{2n-1}^\angle}{z\cdot z^{-n+1}\varphi_{2n-1}^\angle}
		=\ip{\varphi_{2n-1}^\angle}{z\varphi_{2n-1}^\angle}
		=-\overline{\alpha_{2n-1}}\alpha_{2n-2},\\
		\ip{\chi_{2n}^\angle}{z\chi_{2n-1}^\angle}
		&=\ip{z^{-n}\varphi_{2n}^{\angle,*}}{z\cdot z^{-n+1}\varphi_{2n-1}^\angle}
		=\ip{\varphi_{2n}^{\angle,*}}{z^2\varphi_{2n-1}^\angle}
		=\overline{\alpha_{2n}\rho_{2n-1}},\\
		\ip{\chi_{2n+1}^\angle}{z\chi_{2n-1}^\angle}
		&=\ip{z^{-(n+1)+1}\varphi_{2n+1}^\angle}{z\cdot z^{-n+1}\varphi_{2n-1}^\angle}
		=\ip{\varphi_{2n+1}^\angle}{z^2\varphi_{2n-1}^\angle}
		=\overline{\rho_{2n}\rho_{2n-1}},
	\end{align*}
	so (\ref{eq_rcmv_basis_ip_2n-2_z*2n-1}), (\ref{eq_rcmv_basis_ip_2n-1_z*2n-1}), (\ref{eq_rcmv_basis_ip_2n_z*2n-1}) and (\ref{eq_rcmv_basis_ip_2n+1_z*2n-1}) are simply (\ref{eq_ropuc_ip_2n*_z.2n+1}), (\ref{eq_ropuc_ip_2n+1_z.2n+1}), (\ref{eq_ropuc_ip_2n+2*_z2.2n+1}) and (\ref{eq_ropuc_ip_2n+3_z2.2n+1}), respectively, with a reindexing.
\end{proof}

All other entries of the rotated CMV matrix are zero by an identical reasoning as in the standard case.

Hence the rotated CMV matrix $\mathcal{C}^\angle\left(\mu,\{\zeta_n\}_{n=0}^\infty\right)$ is given by
\begin{equation}\label{eq_rcmv_matrix_onesided}
	\mathcal{C}^\angle\left(\mu,\{\zeta_n\}_{n=0}^\infty\right)=\begin{bmatrix}
		\overline{\alpha_0} & \overline{\alpha_1}\rho_0 & \rho_0\rho_1 & 0 & 0 & \cdots\\
		\overline{\rho_0} & -\overline{\alpha_1}\alpha_0 & -\alpha_0\rho_1 & 0 & 0 & \cdots\\
		0 & \overline{\alpha_2}\overline{\rho_1} & -\overline{\alpha_2}\alpha_1 & \overline{\alpha_3}\rho_2 & \rho_3\rho_2 & \cdots\\
		0 & \overline{\rho_2\rho_1} & -\alpha_1\overline{\rho_2} & -\overline{\alpha_3}\alpha_2 & -\alpha_2\rho_3 & \cdots\\
		0 & 0 & 0 & \overline{\alpha_4}\overline{\rho_3} & -\overline{\alpha_4}\alpha_3 & \cdots \\
		\vdots & \vdots & \vdots & \vdots & \vdots & \ddots
	\end{bmatrix},
\end{equation}
which looks exactly like (\ref{eq_rcmv_matrix_onesided_intro}) as we desired.

While the standard CMV matrix $\mathcal{C}$ depends solely on the Verblunsky coefficients $\{\alpha_n\}_{n=0}^\infty$ since $\rho_n$ depends solely on $\alpha_n$, now that our complex $\rho_n$'s depend on both $\{\alpha_n\}_{n=0}^\infty$ and $\{\zeta_n\}_{n=0}^\infty$, the rotated CMV matrix depends on the pairs $\{(\alpha_n,\rho_n)\}_{n=0}^\infty\subseteq\mathbb{D}\times\mathbb{D}$. Hence we shall write $\mathcal{C}^\angle\left(\{\alpha_n,\rho_n\}_{n=0}^\infty\right)$ to emphasise this dependency.

Comparing the rotated CMV matrix $\mathcal{C}^\angle\left(\mu,\{\zeta_n\}_{n=0}^\infty\right)$ in (\ref{eq_rcmv_matrix_onesided}) and the standard CMV matrix $\mathcal{C}(\mu)$ in (\ref{eq_cmv_matrix_onesided_intro}) with respect to the same probability measure $\mu$, we observe that the diagonal and lower triangular entries are identical, and the $\rho_n$'s in the upper triangular entries differ by a complex conjugate.
Furthermore, if $\{\rho_n\}_{n=0}^\infty\subseteq\mathbb{R}$, or if we replace $\{\rho_n\}_{n=0}^\infty\subseteq\mathbb{C}$ by $\{|\rho_n|\}_{n=0}^\infty\subseteq\mathbb{R}$ in the entries of the rotated CMV matrix $\mathcal{C^\angle}$, using the fact that $\overline{|\rho_n|}=|\rho_n|$, we obtain back the standard CMV matrix $\mathcal{C}$ as in (\ref{eq_cmv_matrix_onesided_intro}).

Cedzich, Fillman, Li, Ong and Zhou \cite{Cedzich2023} proved that one can always conjugate a rotated CMV matrix $\mathcal{C}^\angle\left(\{\alpha_n,\rho_n\}_{n=0}^\infty\right)$ into a standard CMV matrix $\mathcal{C}\left(\{\alpha_n,|\rho_n|\}_{n=0}^\infty\right)$ by a unitary diagonal matrix. In particular, one may verify that by defining the unitary $R$ to be 
\begin{equation} \label{eq_unitary_equiv_R}
	R=\diag{1,\zeta_0,\zeta_0\zeta_1,\zeta_0\zeta_1\zeta_2,\cdots},
\end{equation}
then $\mathcal{C}^\angle\left(\{\alpha_n,\rho_n\}_{n=0}^\infty\right)$ and $\mathcal{C}\left(\{\alpha_n,|\rho_n|\}_{n=0}^\infty\right)$ are unitarily equivalent to each other by
\begin{equation} \label{eq_rcmv_onesided_unitary_equivalence}
	\mathcal{C}\left(\{\alpha_n,|\rho_n|\}_{n=0}^\infty\right)=R\mathcal{C}^\angle\left(\{\alpha_n,\rho_n\}_{n=0}^\infty\right)R^{-1}.
\end{equation}
While the unitary equivalence is expected by Cantero, Moral and Vel\'azquez's result (property (2) in the introduction) \cite{Cantero2005,Cantero2012}, the emphasis here is the invariance of the $\alpha_n$'s under the unitary conjugation. That is, conjugation by $R$ carries a rotated CMV matrix with entries depending on $\{\alpha_n,\rho_n\}_{n=0}^\infty$ into a standard CMV matrix with entries in terms of $\{\alpha_n,|\rho_n|\}_{n=0}^\infty$.

Similarly, we can define the \textit{rotated alternate CMV basis and matrix} as follows:

\begin{definition}[Rotated Alternate CMV Basis and Operator] \label{def_rcmv_alt}
	Let $\mu$ be a nontrivial probability measure on $\partial\mathbb{D}$ and $\{\tilde{\chi}_n(z,\mu)\}_{n=0}^\infty$ be the standard alternate CMV basis associated to $\mu$. The \textit{rotated alternate CMV basis} $\{\tilde{\chi}_n^\angle(z,\mu)\}_{n=0}^\infty$ associated to the sequence $\{\zeta_n\}_{n=0}^\infty\subseteq\partial\mathbb{D}$ is defined by
	\begin{equation} \label{eq_rcmv_alt_basis_def}
		\tilde{\chi}_n^\angle(z)=\left(\prod_{j=0}^{n-1}\zeta_j\right)\tilde{\chi}_n(z).
	\end{equation}
	The \textit{rotated CMV matrix} $\tilde{\mathcal{C}}^\angle\left(\mu,\{\zeta_n\}_{n=0}^\infty\right)=[\tilde{\mathcal{C}}_{k\ell}^\angle]$ associated to $\{\zeta_n\}_{n=0}^\infty\subseteq\partial\mathbb{D}$ is defined by
	\begin{equation} \label{eq_rcmv_alt_matrix_def}
		\tilde{\mathcal{C}}_{k\ell}^\angle=\ip{\tilde{\chi}_k^\angle}{z\tilde{\chi}_\ell^\angle}.
	\end{equation}
\end{definition}

\begin{proposition} \label{prop_rcmv_alt_basis_opuc}
	Let $\mu$ be a nontrivial probability measure on $\partial\mathbb{D}$, $\{\varphi^\angle_n(z,\mu)\}_{n=0}^\infty$ be the rotated OPUCs as defined in (\ref{eq_ropuc_def}) and $\{\tilde{\chi}_n(z,\mu)\}_{n=0}^\infty$ be the rotated alternate CMV basis defined in (\ref{eq_rcmv_alt_basis_def}). Then
	\begin{align}
		\label{eq_rcmv_alt_basis_ropuc_odd}
		\tilde{\chi}_{2n-1}^\angle(z)&=z^{-n}\varphi_{2n-1}^{\angle,*}(z),\\
		\label{eq_rcmv_alt_basis_ropuc_even}
		\tilde{\chi}_{2n}^\angle(z)&=z^{-n}\varphi_{2n}^\angle(z),
	\end{align}
	and hence for all $n\in\mathbb{N}_0$ and $z\in\partial\mathbb{D}$,
	\begin{equation} \label{eq_rcmv_basis_alt_basis}
		\tilde{\chi}_n^\angle(z)=\overline{\chi_n^\angle\left(\dfrac{1}{\overline{z}}\right)}.
	\end{equation}
\end{proposition}

\begin{proof}
	Applying (\ref{eq_cmv_alt_basis_opuc_odd}) to (\ref{eq_rcmv_alt_basis_def}) gives (\ref{eq_rcmv_alt_basis_ropuc_odd}). 
	Similarly, (\ref{eq_rcmv_alt_basis_ropuc_even}) is obtained from (\ref{eq_rcmv_alt_basis_def}) and (\ref{eq_cmv_alt_basis_opuc_even}). 
\end{proof}

\begin{proposition} \label{prop_rcmv_alt_basis_ip}
	Let $\mu$ be a nontrivial probability measure on $\partial\mathbb{D}$, $\{\alpha_n(\mu)\}_{n=0}^\infty$ be the Verblunsky coefficients associated to $\mu$, and the sequence $\{\rho_n\}_{n=0}^\infty$ be defined by (\ref{eq_rho_complex}). Then the nonzero entries of $\tilde{\mathcal{C}}^\angle\left(\mu,\{\zeta_n\}_{n=0}^\infty\right)$ are given by
	\begin{align}
		\label{eq_rcmv_alt_basis_ip_2n_z*2n-2}
		\ip{\tilde{\chi}_{2n}^\angle}{z\tilde{\chi}_{2n-2}^\angle}&=\overline{\rho_{2n-2}\rho_{2n-1}},\\
		\label{eq_rcmv_alt_basis_ip_2n_z*2n-1}
		\ip{\tilde{\chi}_{2n}^\angle}{z\tilde{\chi}_{2n-1}^\angle}&=-\alpha_{2n-2}\overline{\rho_{2n-1}},\\
		\label{eq_rcmv_alt_basis_ip_2n_z*2n}
		\ip{\tilde{\chi}_{2n}^\angle}{z\tilde{\chi}_{2n}^\angle}&=-\overline{\alpha_{2n}}\alpha_{2n-1},\\
		\label{eq_rcmv_alt_basis_ip_2n_z*2n+1}
		\ip{\tilde{\chi}_{2n}^\angle}{z\tilde{\chi}_{2n+1}^\angle}&=-\alpha_{2n-1}\rho_{2n},\\
		\label{eq_rcmv_alt_basis_ip_2n-1_z*2n-2}
		\ip{\tilde{\chi}_{2n-1}^\angle}{z{\chi}_{2n-2}^\angle}&=\overline{\alpha_{2n-1}}\overline{\rho_{2n-2}},\\
		\label{eq_rcmv_alt_basis_ip_2n-1_z*2n-1}
		\ip{\tilde{\chi}_{2n-1}^\angle}{z\tilde{\chi}_{2n-1}^\angle}&=-\overline{\alpha_{2n-1}}\alpha_{2n-2},\\
		\label{eq_rcmv_alt_basis_ip_2n-1_z*2n}
		\ip{\tilde{\chi}_{2n-1}^\angle}{z\tilde{\chi}_{2n}^\angle}&=\overline{\alpha_{2n}}\rho_{2n-1},\\
		\label{eq_rcmv_alt_basis_ip_2n-1_z*2n+1}
		\ip{\tilde{\chi}_{2n-1}^\angle}{z\tilde{\chi}_{2n+1}^\angle}&=\rho_{2n}\rho_{2n-1}.
	\end{align}
\end{proposition}

\begin{proof}
	\allowdisplaybreaks
	Using a similar technique in the proof of (\ref{eq_ropuc_ip_2n+3_z2.2n+1}), together with (\ref{eq_L2_ip}), (\ref{eq_ropuc_reverse_ropuc}) and (\ref{eq_ropuc_ip_j*k*}) (or (\ref{eq_ropuc_ip_2n-2*_2n*})), one can show that
	\begin{align*}
		\ip{\tilde{\chi}_{2n}^\angle}{z\tilde{\chi}_{2n-2}^\angle}
		=\zeta_{2n-2}^{-2}\zeta_{2n-1}^{-2}\ip{\varphi_{2n-2}^{\angle,*}}{\varphi_{2n}^{\angle,*}}
		=\overline{\rho_{2n-2}\rho_{2n-1}}.
	\end{align*}
	which proves (\ref{eq_rcmv_alt_basis_ip_2n_z*2n-2}). To prove (\ref{eq_rcmv_alt_basis_ip_2n_z*2n-1}), we use the same technique and invoke (\ref{eq_L2_ip}) and (\ref{eq_ropuc_reverse_ropuc}), this time with (\ref{eq_ropuc_ip_jk*}) (or (\ref{eq_ropuc_ip_2n-1_2n*})):
	\allowdisplaybreaks
	\begin{align*}
		\ip{\tilde{\chi}_{2n}^\angle}{z\tilde{\chi}_{2n-1}^\angle}
		=\zeta_{2n-1}^{-2}\ip{\varphi_{2n-1}^\angle}{\varphi_{2n}^{\angle,*}}
		=-\alpha_{2n-2}\overline{\rho_{2n-1}}.
	\end{align*}
	We can reduce (\ref{eq_rcmv_alt_basis_ip_2n_z*2n}) to a special case of (\ref{eq_ropuc_ip_j_z.k}):
	\begin{align*}
		\ip{\tilde{\chi}_{2n}^\angle}{z\tilde{\chi}_{2n}}=\ip{z^{-n}\varphi_{2n}^\angle}{z\cdot z^{-n}\varphi_{2n}^\angle}=\ip{\varphi_{2n}^\angle}{z\varphi_{2n}^\angle}=-\overline{\alpha_{2n}}\alpha_{2n-1}.
	\end{align*}
	We show that (\ref{eq_rcmv_alt_basis_ip_2n_z*2n+1}) is a special case of (\ref{eq_ropuc_ip_jk*}):
	\begin{align*}
		\ip{\tilde{\chi}_{2n}^\angle}{z\tilde{\chi}_{2n+1}^\angle}
		=\ip{z^{-n}\varphi_{2n}^\angle}{z\cdot z^{-(n+1)}\varphi_{2n+1}^{\angle,*}}
		=\ip{\varphi_{2n}^\angle}{\varphi_{2n+1}^{\angle,*}}
		=\alpha_{2n-1}\rho_{2n}.
	\end{align*}
	To prove (\ref{eq_rcmv_alt_basis_ip_2n-1_z*2n-2}), we again use the technique in the proof of (\ref{eq_ropuc_ip_2n+3_z2.2n+1}) with (\ref{eq_L2_ip}), (\ref{eq_ropuc_reverse_ropuc}), and (\ref{eq_ropuc_ip_j*_z.k}):
	\begin{align*}
		\ip{\tilde{\chi}_{2n-1}^\angle}{z\tilde{\chi}_{2n-2}}
		=\zeta_{2n-2}^{-2}\ip{\varphi_{2n-2}^{\angle,*}}{z\varphi_{2n-1}^\angle}
		=\overline{\alpha_{2n-1}\rho_{2n-2}}.
	\end{align*}
	We can also reduce (\ref{eq_rcmv_alt_basis_ip_2n-1_z*2n-1}) to a special case of (\ref{eq_ropuc_ip_j_z.k}) by the same method.
	We show that (\ref{eq_rcmv_alt_basis_ip_2n-1_z*2n}) is a special case of (\ref{eq_ropuc_ip_j*_z.k}):
	\begin{align*}
		\ip{\tilde{\chi}_{2n-1}^\angle}{z\tilde{\chi}_{2n}}
		=\ip{z^{-n}\varphi_{2n-1}^{\angle,*}}{z\cdot z^{-n}\varphi_{2n}^{\angle}}
		=\ip{\varphi_{2n-1}^{\angle,*}}{z\varphi_{2n}^{\angle}}
		=\overline{\alpha_{2n}}\rho_{2n-1}.
	\end{align*}
	We show that (\ref{eq_rcmv_alt_basis_ip_2n-1_z*2n+1}) is a special case of (\ref{eq_ropuc_ip_j*k*}):
	\[\ip{\tilde{\chi}_{2n-1}^\angle}{z\tilde{\chi}_{2n+1}}
	=\ip{z^{-n}\varphi_{2n-1}^{\angle,*}}{z\cdot z^{-(n+1)}\varphi_{2n+1}^{\angle,*}}
	=\ip{\varphi_{2n-1}^{\angle,*}}{\varphi_{2n+1}^{\angle,*}}
	=\rho_{2n-1}\rho_{2n}. \qedhere
	\]
\end{proof} 

Hence we can compute the rotated alternate CMV matrix $\tilde{\mathcal{C}}^\angle\left(\mu,\{\zeta_n\}_{n=0}^\infty\right)$ to be
\begin{equation} \label{eq_rcmv_alt_matrix_onesided}
	\tilde{\mathcal{C}}^\angle\left(\mu,\{\zeta_n\}_{n=0}^\infty\right)
	=\begin{bmatrix}
		\overline{\alpha_0} & \rho_0 & 0 & 0 & 0 & \cdots\\
		\overline{\alpha_1\rho_0} & -\overline{\alpha_1}\alpha_0 & \overline{\alpha_2}\rho_1 & \rho_2\rho_1 & 0 & \cdots \\
		\overline{\rho_1\rho_0} & -\alpha_0\overline{\rho_1} & -\overline{\alpha_2}\alpha_1 & -\alpha_1\rho_2 & 0 & \cdots\\
		0 & 0 & \overline{\alpha_3\rho_2} & -\overline{\alpha_3}\alpha_2 & \overline{\alpha_4}\rho_3 & \cdots\\
		0 & 0 & \overline{\rho_3\rho_2} & -\alpha_2\overline{\rho_3} & -\overline{\alpha_4}\alpha_3 & \cdots \\
		\vdots & \vdots & \vdots & \vdots & \vdots & \ddots
	\end{bmatrix}.
\end{equation}

As suggested by Cedzich, Fillman, Li, Ong and Zhou \cite{Cedzich2023}, one may also verify that the rotated alternate CMV matrices are unitarily conjugate to their standard CMV counterparts by
\begin{equation} \label{eq_rcmv_onesided_alt_unitary_equivalence}
	\tilde{\mathcal{C}}\left(\{\alpha_n,|\rho_n|\}_{n=0}^\infty\right)=R\tilde{\mathcal{C}}^\angle\left(\{\alpha_n,\rho_n\}_{n=0}^\infty\right)R^{-1},
\end{equation}
where $R$ is given by (\ref{eq_unitary_equiv_R}). That is, one can conjugate the standard CMV matrices $\mathcal{C}$ (resp. $\tilde{\mathcal{C}}$) into their rotated counterparts $\mathcal{C}^\angle$ (resp. $\tilde{\mathcal{C}}^\angle$) by the same unitary matrix $R$ given by (\ref{eq_unitary_equiv_R}). Again the invariance of the $\alpha_n$'s are emphasized here.

However since $\overline{\rho_n}\neq\rho_n$ in $\mathbb{D}$, unlike the standard case where $\tilde{\mathcal{C}}=\mathcal{C}^T$, in the rotated case, we generally have $\tilde{\mathcal{C}}^\angle\neq\left(\mathcal{C}^\angle\right)^T$.
Instead, $\tilde{\mathcal{C}^\angle}$ is conjugate to $\left(\mathcal{C}^\angle\right)^T$ by $Q=R^2$ by
\begin{equation} \label{eq_rcmv_onesided_alt_unitary_equivalence_transpose}
	\left(\mathcal{C^\angle}\right)^T=Q\tilde{\mathcal{C}}^\angle Q^{-1},
\end{equation}
where $R$ is again given by (\ref{eq_unitary_equiv_R}), which can be verified by the straightforward computation
\begin{align*}
	Q\tilde{\mathcal{C}}^\angle Q^{-1}
	=R^2(R^{-1}\tilde{\mathcal{C}}R)R^{-2}
	=R\mathcal{C}^TR^{-1}
	=\left(R\mathcal{C}R^{-1}\right)^T
	=\left(\mathcal{C^\angle}\right)^T.
\end{align*}

\subsection{Factorisation of Rotated CMV Operators}

To show that the rotated CMV matrices admit $\mathcal{LM}$-factorisations analogous to the standard case, we begin by defining rotated versions of $\mathcal{L}$ and $\mathcal{M}$: 

\begin{definition}[Rotated $\mathcal{L}$ and $\mathcal{M}$]
	Let $\mu$ be a nontrivial probability measure on $\partial\mathbb{D}$ and $\{\chi_n^\angle(z,\mu)\}_{n=0}^\infty$, $\{\tilde{\chi}_n^\angle(z,\mu)\}_{n=0}^\infty$ be the rotated CMV basis and rotated alternate CMV basis defined in (\ref{eq_rcmv_basis_def}) and (\ref{eq_rcmv_alt_basis_def}), respectively. Define the operators $\mathcal{L}^\angle$ and $\mathcal{M}^\angle$ on $\ell^2(\mathbb{N}_0)$ by
	\begin{align}
		\label{eq_L_rcmv}
		\mathcal{L}_{k\ell}^\angle&=\ip{\chi_k^\angle}{z\tilde{\chi}_\ell^\angle},\\
		\label{eq_M_rcmv}
		\mathcal{M}_{k\ell}^\angle&=\ip{\tilde{\chi}_k^\angle}{\chi_\ell^\angle}.
	\end{align}
\end{definition}

\begin{theorem}[$\mathcal{L}\mathcal{M}$-factorisation for rotated CMV matrices] \label{thm_LM_rcmv}
	If $\mathcal{C}^\angle\left(\mu,\{\zeta_n\}_{n=0}^\infty\right)$ and \\ $\tilde{\mathcal{C}}^\angle\left(\mu,\{\zeta_n\}_{n=0}^\infty\right)$ are the rotated CMV matrix and rotated alternate CMV matrix as defined in (\ref{eq_rcmv_matrix_def}) and (\ref{eq_rcmv_alt_matrix_def}), respectively, then 
	\begin{align}
		\label{eq_rcmv_LM}
		\mathcal{C}^\angle&=\mathcal{L}^\angle\mathcal{M}^\angle, \\
		\label{eq_rcmv_alt_LM}
		\tilde{\mathcal{C}}^\angle&=\mathcal{M}^\angle\mathcal{L}^\angle.
	\end{align}
\end{theorem}

\begin{proof}
	Since $\{\tilde{\chi}_j\}_{j=0}^\infty$ is an orthonormal basis of $L^2(\partial\mathbb{D},\mu)$, we write
	\begin{equation}
		\chi_{\ell}^\angle=\sum_{j=0}^\infty c_{\ell,j}\tilde{\chi}_j^\angle,
	\end{equation}
	where $\{c_{\ell,j}\}_{j=0}^\infty\subseteq\mathbb{C}$. Taking inner product of $\tilde{\chi}_k^\angle$ with both sides then gives
	\begin{equation}
		\ip{\tilde{\chi}_k^\angle}{\chi_\ell^\angle}
		=\sum_{j=0}^\infty c_{\ell,j}\ip{\tilde{\chi}_k^\angle}{\tilde{\chi_j}^\angle}
		=\sum_{j=0}^\infty c_{\ell,j}\delta_{kj}=c_{\ell,k}
	\end{equation}
	and hence
	\begin{align*}
		\mathcal{C}^\angle_{k\ell}
		=\ip{\chi_k^\angle}{z\sum_{j=0}^\infty\ip{\tilde{\chi}_j^\angle}{\chi_\ell^\angle}\tilde{\chi}_j^\angle}
		=\sum_{j=0}^\infty \ip{\chi_k^\angle}{z\tilde{\chi}_j^\angle}\ip{\tilde{\chi}_j^\angle}{\chi_\ell^\angle}
		=\sum_{j=0}^\infty \mathcal{L}^\angle_{kj}\mathcal{M}^\angle_{j\ell}
		=(\mathcal{L}^\angle\mathcal{M}^\angle)_{k\ell}
	\end{align*}
	which proves (\ref{eq_rcmv_LM}). The proof of (\ref{eq_rcmv_alt_LM}) follows a similar argument.
\end{proof}

On a side note, we can express $\mathcal{C}^\angle$ and $\tilde{\mathcal{C}}^\angle$ in terms of the standard $\mathcal{L}$ and $\mathcal{M}$:
\begin{align}
	\mathcal{C}^\angle&=R\mathcal{L}\mathcal{M}R^{-1},\\
	\tilde{\mathcal{C}}^\angle&=R\mathcal{M}\mathcal{L}R^{-1}.
\end{align}

\begin{theorem}[$\Theta$-factorisation] \label{thm_LM_theta_rcmv}
	Let $\mu$ be a nontrivial probability measure on $\partial\mathbb{D}$ and $\mathcal{C}(\mu)$, $\tilde{\mathcal{C}}(\mu)$ be the CMV matrix and alternate CMV matrix as defined in (\ref{eq_cmv_matrix_def}) and (\ref{eq_cmv_alt_matrix_def}), respectively. For each $n\in\mathbb{N}_0$, define
	\begin{equation} \label{eq_theta_n_rcmv}
		\Theta_n^\angle=\begin{bmatrix}
			\overline{\alpha_n} & \rho_n \\ \overline{\rho_n} & -\alpha_n
		\end{bmatrix}.
	\end{equation}
	If $\mathcal{L}^\angle(\mu)$ and $\mathcal{M}^\angle(\mu)$ are defined as in (\ref{eq_L_rcmv}) and (\ref{eq_M_rcmv}), respectively, then
	\begin{gather}
		\label{eq_L_rcmv_theta}
		\mathcal{L}^\angle=\bigoplus_{n=0}^\infty \Theta_{2n}^\angle
		=\diag{\Theta_0^\angle,\Theta_2^\angle,\Theta_4^\angle,\cdots}
		=\begin{bmatrix}
			\Theta_0^\angle & & & 0 \\ & \Theta_2^\angle & & \\ & & \Theta_4^\angle & \\ 0 & & & \ddots
		\end{bmatrix},\\
		\label{eq_M_rcmv_theta}
		\mathcal{M}^\angle=1\oplus\bigoplus_{n=1}^\infty\Theta_{2n-1}^\angle
		=\diag{1,\Theta_1^\angle,\Theta_3^\angle,\cdots}
		=\begin{bmatrix}
			1 & & & 0 \\ & \Theta_1^\angle & & \\ & & \Theta_3^\angle & \\ 0 & & & \ddots
		\end{bmatrix},
	\end{gather}
	where $1$ denotes the $1\times 1$ block with entry $1$.
\end{theorem}

\begin{proof}
	Substituting (\ref{eq_szego_recursion_ropuc_rewrite}) into (\ref{eq_rcmv_alt_basis_ropuc_even}) and then invoking (\ref{eq_rcmv_basis_ropuc_odd}) and (\ref{eq_rcmv_basis_ropuc_even}) gives
	\begin{align*}
		z\tilde{\chi}_{2n}^\angle(z)
		&=\overline{\alpha_{2n}}\chi_{2n}^\angle(z)+\overline{\rho_{2n}}\chi_{2n+1}^\angle(z)
	\end{align*}
	and hence
	\begin{align*}
		\mathcal{L}^\angle_{2n,2n}&=\ip{\chi_{2n}^\angle}{z\tilde{\chi}_{2n}^\angle}=\overline{\alpha_{2n}},\\
		\mathcal{L}^\angle_{2n+1,2n}&=\ip{\chi_{2n+1}^\angle}{z\tilde{\chi}_{2n}^\angle}=\overline{\rho_{2n}},
	\end{align*}
	thus giving, respectively, the diagonal entries for the even rows and the entries directly below them, which corresponds to the left column of $\Theta_{2n}^\angle$. On the other hand, applying (\ref{eq_inverse_szego_recursion_ropuc_reverse_rewrite}) to (\ref{eq_rcmv_alt_basis_ropuc_odd}) together with (\ref{eq_rcmv_basis_ropuc_odd}) and (\ref{eq_rcmv_basis_ropuc_even}) gives
	\begin{align*}
		z\tilde{\chi}_{2n+1}^\angle(z)
		&=\rho_{2n}\chi_{2n}^\angle(z)-\alpha_{2n}\chi_{2n+1}^\angle(z)
	\end{align*}
	and hence
	\begin{align*}
		\mathcal{L}^\angle_{2n+1,2n+1}&=\ip{\chi_{2n+1}^\angle}{z\tilde{\chi}_{2n+1}^\angle}=-\alpha_{2n},\\
		\mathcal{L}^\angle_{2n,2n+1}&=\ip{\chi_{2n}^\angle}{z\tilde{\chi}_{2n+1}^\angle}=\rho_{2n},
	\end{align*}
	thus giving, respectively, the diagonal entries for the odd rows of $\mathcal{L}^\angle$ and the entries directly above them, which corresponds to the right column of $\Theta_{2n}^\angle$. By the block and tridiagonal structure of $\mathcal{L}^\angle$, the proof of (\ref{eq_L_rcmv_theta}) is complete.
	
	To prove (\ref{eq_M_rcmv_theta}), we apply (\ref{eq_inverse_szego_recursion_ropuc_reverse_rewrite}) to (\ref{eq_rcmv_basis_ropuc_even}) together with (\ref{eq_rcmv_alt_basis_ropuc_odd}) and (\ref{eq_rcmv_alt_basis_ropuc_even}):
	\begin{align*}
		\chi_{2n}^\angle(z)
		&=-\alpha_{2n-1}\tilde{\chi}_{2n}^\angle(z)+\rho_{2n-1}\tilde{\chi}_{2n-1}^\angle(z)
	\end{align*}
	and hence
	\begin{align*}
		\mathcal{M}^\angle_{2n,2n}&=\ip{\tilde{\chi}_{2n}}{\chi_{2n}}=-\alpha_{2n-1},\\
		\mathcal{M}^\angle_{2n-1,2n}&=\ip{\tilde{\chi}_{2n-1}}{\chi_{2n}}=\rho_{2n-1},
	\end{align*}
	giving, respectively, the diagonal entries for the even rows of $\mathcal{M}^\angle$ and the entries directly above them for all $n\geq 1$, which corresponds to the left column of $\Theta_{2n-1}^\angle$. When $n=0$, we obtain the first $1\times 1$ block with entry $-\alpha_{-1}=1$. 
	On the other hand, applying (\ref{eq_szego_recursion_ropuc_rewrite}) to (\ref{eq_rcmv_basis_ropuc_odd}) together with (\ref{eq_rcmv_alt_basis_ropuc_odd}) and (\ref{eq_rcmv_alt_basis_ropuc_even}) yields
	\begin{align*}
		\chi_{2n-1}^\angle(z)
		&=\overline{\alpha_{2n-1}}\tilde{\chi}_{2n-1}(z)+\overline{\rho_{2n-1}}\tilde{\chi}_{2n}(z)
	\end{align*}
	and hence
	\begin{align*}
		\mathcal{M}^\angle_{2n-1,2n-1}&=\ip{\tilde{\chi}_{2n-1}}{\chi_{2n-1}}=\overline{\alpha_{2n-1}},\\
		\mathcal{M}^\angle_{2n,2n-1}&=\ip{\tilde{\chi}_{2n}}{\chi_{2n-1}}=\overline{\rho_{2n-1}},
	\end{align*}
	which give, respectively, the diagonal entries for the odd rows of $\mathcal{M}^\angle$ and the entries directly below them, which corresponds to the left column of $\Theta_{2n-1}^\angle$. By the block and tridiagonal structure of $\mathcal{M}^\angle$, this completes the proof of (\ref{eq_M_rcmv_theta}).
\end{proof}

\begin{example}
	Again consider the probability measure $\mu$ on $\partial\mathbb{D}$ given by
	$$d\mu=(1-\cos\theta)\dfrac{d\theta}{2\pi}.$$
	The standard CMV matrix of $\mu$ is
	\begin{equation}
		\mathcal{C}^\angle=\begin{bmatrix}
			-\frac{1}{2} & -\frac{\sqrt{3}}{6} & \frac{\sqrt{6}}{3} & 0 & 0 & \cdots\\
			\frac{\sqrt{3}}{2} & -\frac{1}{6} & \frac{\sqrt{2}}{3} & 0 & 0 & \cdots \\
			0 & -\frac{\sqrt{2}}{6} & -\frac{1}{12} & -\frac{\sqrt{15}}{20} & \frac{3\sqrt{10}}{10} & \cdots\\
			0 & \frac{\sqrt{30}}{6} & \frac{\sqrt{15}}{12} & -\frac{1}{20} & \frac{\sqrt{6}}{10} & \cdots\\
			0 & 0 & 0 & -\frac{\sqrt{6}}{15} & -\frac{1}{30} & \cdots \\
			\vdots & \vdots & \vdots & \vdots & \vdots & \ddots
		\end{bmatrix},
	\end{equation}
	which is generated by the standard OPUCs $\{\varphi_n\}_{n=0}^\infty$ given by
	\begin{equation}
		\varphi_n(z)=\left(\dfrac{2}{(n+1)(n+2)}\right)^{\frac{1}{2}}\sum_{k=0}^n(k+1)z^k.
	\end{equation}
	On the other hand, the rotated CMV matrix and rotated alternate CMV matrix of $\mu$ associated to the constant sequence $\left\{\zeta_n\right\}_{n=0}^\infty$ defined by
	$$\zeta_n=e^{i\frac{\pi}{4}}=\dfrac{1+i}{\sqrt{2}}$$
	are given by
	\begin{align}
		\mathcal{C}^\angle&=
		\begin{bmatrix}
			-\frac{1}{2} & -\frac{\sqrt{3}}{6}\cdot\frac{1+i}{\sqrt{2}} & \frac{\sqrt{6}}{3}i & 0 & 0 & \cdots\\
			\frac{\sqrt{3}}{2}\cdot\frac{1-i}{\sqrt{2}} & -\frac{1}{6} & \frac{\sqrt{2}}{3}\cdot\frac{1+i}{\sqrt{2}} & 0 & 0 & \cdots \\
			0 & -\frac{\sqrt{2}}{6}\cdot\frac{1-i}{\sqrt{2}} & -\frac{1}{12} & -\frac{\sqrt{15}}{20}\cdot\frac{1+i}{\sqrt{2}} & \frac{3\sqrt{10}}{10}i & \cdots\\
			0 & -\frac{\sqrt{30}}{6}i & \frac{\sqrt{15}}{12}\cdot\frac{1-i}{\sqrt{2}} & -\frac{1}{20} & \frac{\sqrt{6}}{10}\cdot\frac{1+i}{\sqrt{2}} & \cdots\\
			0 & 0 & 0 & -\frac{\sqrt{6}}{15}\cdot\frac{1-i}{\sqrt{2}} & -\frac{1}{30} & \cdots \\
			\vdots & \vdots & \vdots & \vdots & \vdots & \ddots
		\end{bmatrix},\\
		\tilde{\mathcal{C}}^\angle&=
		\begin{bmatrix}
			-\frac{1}{2} & \frac{\sqrt{3}}{2}\cdot\frac{1+i}{\sqrt{2}} & 0 & 0 & 0 & \cdots\\
			-\frac{\sqrt{3}}{6}\cdot\frac{1-i}{\sqrt{2}} & -\frac{1}{6} & -\frac{1}{3\sqrt{2}}\cdot\frac{1+i}{\sqrt{2}} & \frac{\sqrt{30}}{6}i & 0 & \cdots \\
			-\frac{\sqrt{2}}{\sqrt{3}}i & \frac{\sqrt{2}}{3}\cdot\frac{1-i}{\sqrt{2}} & -\frac{1}{12} & \frac{\sqrt{15}}{12}\cdot\frac{1+i}{\sqrt{2}} & 0 & \cdots\\
			0 & 0 & -\frac{\sqrt{15}}{20}\cdot\frac{1-i}{\sqrt{2}} & -\frac{1}{20} & -\frac{\sqrt{6}}{15}\cdot\frac{1+i}{\sqrt{2}} & \cdots\\
			0 & 0 & -\frac{3\sqrt{10}}{10}i & \frac{\sqrt{6}}{10}\cdot\frac{1-i}{\sqrt{2}} & -\frac{1}{30} & \cdots \\
			\vdots & \vdots & \vdots & \vdots & \vdots & \ddots
		\end{bmatrix},
	\end{align}
	respectively, which are generated by the rotated OPUCs $\{\varphi_n^\angle\}_{n=0}^\infty$ given by
	\begin{equation}
		\varphi_n^\angle(z)=e^{i\frac{n\pi}{4}}\left(\dfrac{2}{(n+1)(n+2)}\right)^{\frac{1}{2}}\sum_{k=0}^n(k+1)z^k.
	\end{equation}
	Moreover, (\ref{eq_rcmv_onesided_alt_unitary_equivalence_transpose}) implies that $\left(\mathcal{C}^\angle\right)^T$ and $\tilde{\mathcal{C}}^\angle$ are unitarily conjugate to each other by
	\begin{equation}
		\left(\mathcal{C}^\angle\right)^T=Q\tilde{\mathcal{C}}^\angle Q^{-1}, \qquad Q=\diag{1,e^{i\frac{2\pi}{4}},e^{i\frac{4\pi}{4}},e^{i\frac{6\pi}{4}},\cdots}.
	\end{equation}
	Note that in this particular case, we have $\left(\mathcal{C}^\angle\right)^*=\tilde{\mathcal{C}}$, where $^*$ denotes the conjugate transpose, due to the reason that the Verblunsky coefficients $\{\alpha_n(\mu)\}_{n=0}^\infty\subseteq\mathbb{R}$ which gives $\overline{\alpha_n}=\alpha_n$ for all $n\in\mathbb{N}_0$. Otherwise, it is easy to see that if there exists some $n\in\mathbb{N}_0$ such that $\alpha_n\notin\mathbb{R}$, then we have $\left(\mathcal{C}^\angle\right)^*\neq\tilde{\mathcal{C}}$.
\end{example}

For each $n\in\mathbb{N}_0$, if we define
\begin{equation}
	S_n=\begin{bmatrix}
		1 & 0 \\ 0 & \zeta_n^{-1}
	\end{bmatrix},
\end{equation}
then one may easily show that $\Theta_n^\angle$ and $\Theta_n$ are related by
\begin{equation}
	\Theta_n^\angle=S_n\Theta_n S_n^{-1}.
\end{equation}

\subsection{Rotated Extended CMV Operators}

The rotated $\mathcal{L}\mathcal{M}$-factorisation of the rotated CMV matrices can be used to define the \textit{rotated extended CMV matrices} in a similar manner to the standard case:

\begin{definition} \label{def_rcmv_extended}
	Let $\{\alpha_n\}_{n\in\mathbb{Z}}$ and $\{\rho_n\}_{n\in\mathbb{Z}}$ be sequences in $\overline{\mathbb{D}}$ satisfying $$|\alpha_n|^2+|\rho_n|^2=1$$ for all $n\in\mathbb{Z}$. Define
	\begin{equation} \label{eq_rTheta_extended}
		\Theta_n^\angle=\begin{bmatrix}
			\overline{\alpha_n} & \rho_n \\ \overline{\rho_n} & -\alpha_n
		\end{bmatrix}
	\end{equation}
	for all $n\in\mathbb{Z}$ and the operators $\tilde{\mathcal{L}}^\angle$ and $\tilde{\mathcal{M}}^\angle$ on $\ell^2(\mathbb{Z})$ by
	\begin{align}
		\label{eq_L_rcmv_extended}
		\tilde{\mathcal{L}}^\angle&=\bigoplus_{k=-\infty}^\infty \Theta_{2k}^\angle,\\
		\label{eq_M_rcmv_extended}
		\tilde{\mathcal{M}}^\angle&=\bigoplus_{k=-\infty}^\infty \Theta_{2k-1}^\angle.
	\end{align}
	The \textit{rotated extended CMV operator} $\mathcal{E}^\angle\left(\{\alpha_n,\rho_n\}_{n=0}^\infty\right)$ and \textit{rotated extended alternate CMV operator} $\tilde{\mathcal{E}}^\angle\left(\{\alpha_n,\rho_n\}_{n=0}^\infty\right)$ on $\ell^2(\mathbb{Z})$ are defined, respectively, by
	\begin{align}
		\label{eq_rcmv_extended}
		\mathcal{E}^\angle&=\tilde{\mathcal{L}}^\angle\tilde{\mathcal{M}}^\angle,\\
		\label{eq_rcmv_extended_alt}
		\tilde{\mathcal{E}}^\angle&=\tilde{\mathcal{M}}^\angle\tilde{\mathcal{L}}^\angle.
	\end{align} 
\end{definition}

One may show that $\mathcal{E}^\angle\left(\{\alpha_n,\rho_n\}_{n=0}^\infty\right)$ is given explicitly by
\begin{equation}\label{eq_rcmv_matrix_extended}
	\begin{bmatrix}
		\ddots & \vdots & \vdots & \vdots & \vdots & \vdots & \vdots & \vdots\\
		\hdots & -\overline{\alpha_{-1}}\alpha_{-2} & -\alpha_{-2}\rho_{-1} & 0 & 0 & 0 & 0 & \hdots\\
		\hdots & \overline{\alpha_0\rho_{-1}} & \boxed{-\overline{\alpha_0}\alpha_{-1}} & \overline{\alpha_1}\rho_0 & \rho_0\rho_1 & 0 & 0 & \hdots\\
		\hdots & \overline{\rho_0\rho_{-1}} & -\alpha_{-1}\overline{\rho_0} & -\overline{\alpha_1}\alpha_0 & -\alpha_0\rho_1 & 0 & 0 & \hdots \\
		\hdots & 0 & 0 & \overline{\alpha_2\rho_1} & -\overline{\alpha_2}\alpha_1 & \overline{\alpha_3}\rho_2 & \rho_3\rho_2 & \hdots\\
		\hdots & 0 & 0 & \overline{\rho_2\rho_1} & -\alpha_1\overline{\rho_2} & -\overline{\alpha_3}\alpha_2 & -\alpha_2\rho_3 & \hdots\\
		\hdots & 0 & 0 & 0 & 0 & \overline{\alpha_4\rho_3} & -\overline{\alpha_4}\alpha_3 & \hdots \\
		\vdots & \vdots & \vdots & \vdots & \vdots & \vdots & \vdots & \ddots
	\end{bmatrix}.
\end{equation}
Similarly, $\tilde{\mathcal{E}}^\angle\left(\{\alpha_n,\rho_n\}_{n\in\mathbb{Z}}\right)$ is given by
\begin{equation}\label{eq_rcmv_alt_matrix_extended}
	\begin{bmatrix}
		\ddots & \vdots & \vdots & \vdots & \vdots & \vdots & \vdots & \vdots\\
		\hdots & -\overline{\alpha_{-1}}\alpha_{-2} & \overline{\alpha_0}\rho_{-1} & \rho_0\rho_{-1} & 0 & 0 & 0 & \hdots\\
		\hdots & -\alpha_{-2}\overline{\rho_{-1}} & \boxed{-\overline{\alpha_0}\alpha_{-1}} & -\alpha_{-1}\rho_0 & 0 & 0 & 0 & \hdots\\
		\hdots & 0 & \overline{\alpha_1}\overline{\rho_{0}} & -\overline{\alpha_1}\alpha_0 & \overline{\alpha_2}\rho_1 & \rho_2\rho_1 & 0 & \hdots \\
		\hdots & 0 & \overline{\rho_0\rho_1} & -\alpha_0\overline{\rho_1} & -\overline{\alpha_2}\alpha_1 & -\alpha_1\rho_2 & 0 & \hdots\\
		\hdots & 0 & 0 & 0 & \overline{\alpha_3}\overline{\rho_2} & -\overline{\alpha_3}\alpha_2 & \overline{\alpha_4}\rho_3 & \hdots\\
		\hdots & 0 & 0 & 0 & \overline{\rho_3\rho_2} & -\alpha_2\overline{\rho_3} & -\overline{\alpha_4}\alpha_3 & \hdots \\
		\vdots & \vdots & \vdots & \vdots & \vdots & \vdots & \vdots & \ddots
	\end{bmatrix}.
\end{equation}
Note that the $(0,0)^th$ entries of $\mathcal{E}^\angle$ and $\tilde{\mathcal{E}}^\angle$ are given by $-\overline{\alpha_0}\alpha_{-1}$, which are boxed in (\ref{eq_rcmv_matrix_extended}) and (\ref{eq_rcmv_alt_matrix_extended}). Similar to the standard case, we have $\mathcal{E}^\angle_{k\ell}=\mathcal{C}^\angle_{k\ell}$ and $\tilde{\mathcal{E}}^\angle_{k\ell}=\tilde{\mathcal{C}}^\angle_{k\ell}$ for all $k,\ell\geq 1$.

Analogous to the $\ell^2(\mathbb{N}_0)$ case, one may verify that $\mathcal{E}^\angle\left(\{\alpha_n,\rho_n\}_{n\in\mathbb{Z}}\right)$ is unitarily equivalent to $\mathcal{E}\left(\{\alpha_n,|\rho_n|\}_{n\in\mathbb{Z}}\right)$ by
\begin{equation} \label{eq_rcmv_extended_unitary_equivalence}
	\mathcal{E}\left(\{\alpha_n,|\rho_n|\}_{n\in\mathbb{Z}}\right)
	=\tilde{R}\mathcal{E}^\angle\left(\{\alpha_n,\rho_n\}_{n\in\mathbb{Z}}\right)\tilde{R}^{-1},
\end{equation}
where
\begin{equation} \label{eq_rcmv_extended_unitary_equivalence_R_tilde}
	\tilde{R} =\diag{\cdots,\zeta_{-3}^{-1}\zeta_{-2}^{-1}\zeta_{-1}^{-1},\zeta_{-2}^{-1}\zeta_{-1}^{-1},\zeta_{-1}^{-1},\boxed{1},\zeta_0,\zeta_0\zeta_1,\zeta_0\zeta_1\zeta_2,\cdots},
\end{equation}
where the (0,0)-th entry of $\tilde{R}$ is boxed and the sequence $\{\zeta_n\}_{n\in\mathbb{Z}}$ is defined by
\begin{equation} 
	\zeta_n=\begin{cases}
		\dfrac{\rho_n}{|\rho_n|} & \rho_n\neq 0; \\ 0 & \rho_n=0.
	\end{cases}
\end{equation}
This is the exact statement proved by Cedzich, Fillman, Li, Ong and Zhou \cite[Corollary 2.3]{Cedzich2023}.

Similarly, one can verify that $\tilde{\mathcal{E}}^\angle\left(\{\alpha_n,\rho_n\}_{n\in\mathbb{Z}}\right)$ and $\tilde{\mathcal{E}}\left(\{\alpha_n,|\rho_n|\}_{n\in\mathbb{Z}}\right)$ are unitarily equivalent by the same $\tilde{R}$ given by (\ref{eq_rcmv_extended_unitary_equivalence_R_tilde}):
\begin{equation} \label{eq_rcmv_extended_alt_unitary_equivalence}
	\tilde{\mathcal{E}}\left(\{\alpha_n,\rho_n\}_{n\in\mathbb{Z}}\right)
	=\tilde{R}\tilde{\mathcal{E}}^\angle\left(\{\alpha_n,\rho_n\}_{n\in\mathbb{Z}}\right)\tilde{R}^{-1},
\end{equation}
which gives an analogue of (\ref{eq_rcmv_onesided_alt_unitary_equivalence}). Moreover, $\tilde{\mathcal{E}}^\angle$ is also unitarily equivalent to $\left(\mathcal{E}^\angle\right)^T$ by
\begin{equation}
	\left(\mathcal{E}^\angle\right)^T 
	=\tilde{R}^2\tilde{\mathcal{E}}^\angle\tilde{R}^{-2},
\end{equation}
thus giving an analogue of (\ref{eq_rcmv_onesided_alt_unitary_equivalence_transpose}).

\subsection{Rotated Second Kind Polynomials}

Similar to the rotated OPUCs, given a probability measure $\mu$ on $\partial\mathbb{D}$ and a sequence $\{\zeta_n\}_{n=0}^\infty\subseteq\partial\mathbb{D}$, we can also define the \textit{rotated second kind polynomials} $\{\psi_n^\angle(z,\mu)\}_{n=0}^\infty$ by
\begin{equation} \label{eq_ropuc_second_kind}
	\psi_n^\angle(z)=\left(\prod_{j=0}^{n-1}\zeta_j\right)\psi_n(z),\\
\end{equation}
and the \textit{rotated reverse orthonormal polynomials} $\{\psi_n^{\angle,*}(z,\mu)\}_{n=0}^\infty$ by
\begin{equation} \label{eq_ropuc_second_kind_reverse}
	\psi_n^{\angle,*}(z)=\left(\prod_{j=0}^{n-1}\zeta_j\right)\psi_n^*(z).
\end{equation}

While we define the second kind polynomials $\{\psi_n^\angle\}_{n=0}^\infty$ by rotating the $\psi_n$'s, one may easily verify that our definition is equivalent to directly assigning the boundary condition $\lambda=-1$ to the rotated OPUCs. 

\begin{proposition}
	The rotated second kind polynomials $\{\psi_n^\angle\}_{n=0}^\infty$ with boundary condition $\lambda$ are related to the canonical orthonormal polynomials $\{\varphi_n^\angle\}_{n=0}^\infty$ by
	\begin{equation}
		\label{eq_ropuc_with_boundary_condition_bc}
		\psi_n^\angle(z,\mu)=\varphi_n^\angle(z,\mu_{\lambda=-1}),
	\end{equation}
	whereas their reverse polynomials $\{\psi_n^{\angle,*}\}_{n=0}^\infty$ and $\{\varphi_n^{\angle,*}\}_{n=0}^\infty$ are related by
	\begin{equation}
		\label{eq_ropuc_reverse_with_boundary_condition_bc}
		\psi_n^{\angle,*}(z,\mu)=\varphi_n^{\angle,*}(z,\mu_{\lambda=-1}).
	\end{equation}
\end{proposition}

\begin{proof}
	To prove (\ref{eq_ropuc_with_boundary_condition_bc}), we simply invoke (\ref{eq_ropuc_second_kind}), (\ref{eq_opuc_second_kind}) and (\ref{eq_ropuc_def}).
	Equation (\ref{eq_ropuc_reverse_with_boundary_condition_bc}) follows by directly applying (\ref{eq_ropuc_reverse_ropuc}) to (\ref{eq_ropuc_with_boundary_condition_bc}).
\end{proof}

\begin{proposition}[Forward Szeg\H{o} recursion for rotated second kind polynomials] \label{prop_ropuc_with_boundary_condition_szego_recursion}
	Let $\mu$ be a nontrivial probability measure on $\partial\mathbb{D}$ with Verblunsky coefficients $\{\alpha_n\}_{n=0}^\infty\subseteq\mathbb{D}$ and $\{\psi_n(z,\mu)\}_{n=0}^\infty$ be the rotated second kind polynomials given by (\ref{eq_ropuc_second_kind}).
	Then
	\begin{align} 
		\label{eq_szego_recursion_ropuc_second_kind}
		\psi_{n+1}^\angle(z)&=\overline{\rho_n}^{-1}\left(z\psi_n^\angle(z)+\overline{\alpha_n}\psi_n^{\angle,*}(z)\right),\\
		\label{eq_szego_recursion_ropuc_second_kind_reverse}
		-\psi_{n+1}^{\angle,*}(z)&=\overline{\rho_n}^{-1}\left(-\psi_n^{\angle,*}(z)-\alpha_nz\psi_n^\angle(z)\right).
	\end{align}
\end{proposition}

\begin{proof}
	This is essentially the \hyperref[prop_ropuc_szego_recursion]{rotated forward Szeg\H{o} recursion} for the rotated OPUCs $\{\varphi_n^\angle(z,\mu_{\lambda=-1})\}_{n=0}^\infty$. In particular, (\ref{eq_szego_recursion_ropuc_second_kind}) and (\ref{eq_szego_recursion_ropuc_second_kind_reverse}) are obtained by applying (\ref{eq_ropuc_second_kind}) and (\ref{eq_ropuc_second_kind_reverse}) to the first and second rows of (\ref{eq_szegorecursion_orthonormal_lambda}), respectively.
\end{proof}

Rewriting in matrix form, (\ref{eq_szego_recursion_ropuc_second_kind}) and (\ref{eq_szego_recursion_ropuc_second_kind_reverse}) become
\begin{equation} \label{eq_szego_recursion_ropuc_second_kind_matrix_form}
	\begin{bmatrix}
		\psi_{n+1}^\angle(z) \\ -\psi_{n+1}^{\angle,*}(z)
	\end{bmatrix}=A_n^\angle(z)\begin{bmatrix}
		\psi_{n}^\angle(z) \\ -\psi_{n}^{\angle,*}(z)
	\end{bmatrix},
\end{equation}
where $A_n^\angle(z)$ is the rotated Szeg\H{o} transfer matrix given by (\ref{eq_szego_recursion_ropuc_transition_matrix}).

Again using the fact that $A_n^\angle$ is invertible whenever $\{\alpha_n(\mu)\}_{n=0}^\infty\subseteq\mathbb{D}$, one can obtain the backward Szeg\H{o} recursion for the rotated OPUCs with boundary condition:

\begin{corollary}[Backward Szeg\H{o} recursion for rotated second kind polynomials] \label{prop_ropuc_with_boundary_condition_inverse_szego_recursion}
	Let $\mu$ be a nontrivial probability measure on $\partial\mathbb{D}$ with Verblunsky coefficients $\{\alpha_n\}_{n=0}^\infty\subseteq\mathbb{D}$ and $\{\psi_n(z,\mu)\}_{n=0}^\infty$ be the rotated second kind polynomials given by (\ref{eq_ropuc_second_kind}).
	Then
	\begin{align} 
		\label{eq_inverse_szego_recursion_ropuc_second_kind}
		z\psi_{n}^\angle(z)&=\rho_n^{-1}\left(\psi_{n+1}^\angle(z)-\overline{\alpha_n}\psi_{n+1}^{\angle,*}(z)\right),\\
		\label{eq_inverse_szego_recursion_ropuc_second_kind_reverse}
		-\psi_{n}^{\angle,*}(z)&=\rho_n^{-1}\left(-\psi_{n+1}^{\angle,*}(z)+\alpha_n\psi_{n+1}^{\angle}(z)\right).
	\end{align}
\end{corollary}

\begin{proof}
	Simply multiply $\left(A_n^\angle\right)^{-1}$ to the left of both sides of (\ref{eq_szego_recursion_ropuc_second_kind_matrix_form}).
\end{proof}

Hence the rotated polynomials $\begin{bmatrix}
	\varphi_n^\angle(z) \\ \varphi_n^{\angle,*}(z)
\end{bmatrix}$ and $\begin{bmatrix}
	\psi_n^\angle(z) \\ \psi_n^{\angle,*}(z)
\end{bmatrix}$ are solutions of the \textit{rotated Szeg\H{o} difference equation}
\begin{equation} \label{eq_rotated_szego_difference}
	u_{n+1}(z)=A_n^\angle(z)u_n(z)=\dfrac{1}{\overline{\rho_n}}\begin{bmatrix}
	z & -\overline{\alpha_n} \\ -\alpha_n z & 1
\end{bmatrix}u_n(z)
\end{equation} 
with initial conditions $\begin{bmatrix}
	1 \\ 1
\end{bmatrix}$ and $\begin{bmatrix}
	1 \\ -1
\end{bmatrix}$, respectively.

Again one may more generally define rotated OPUCs $\left\{\left(\varphi_n^\lambda\right)^\angle\right\}_{n=0}^\infty$ with boundary condition $\lambda$ in a similar fashion, then by an identical argument, prove their Szeg\H{o} recursions, and show that
\begin{equation}
	u_n(z)=\begin{bmatrix}
		\left(\varphi_n^\lambda\right)^\angle(z) \\ \overline{\lambda}\left(\varphi_n^\lambda\right)^{\angle,*}(z)
	\end{bmatrix}
\end{equation}
is a solution of
(\ref{eq_rotated_szego_difference}) with initial condition
\begin{equation}
	u_0(z)=\begin{bmatrix}
		1 \\ \overline{\lambda}
	\end{bmatrix}.
\end{equation}


We now want to establish the rotated analogues of (\ref{eq_opuc_with_boundary_condition_orthonormal_z^n}):
\begin{proposition}
	For any $n\in\mathbb{N}_0$, we have
	\begin{equation} \label{eq_ropuc_with_boundary_condition_z^n}
		\psi_n^{\angle,*}(z)\varphi_n^\angle(z)+\varphi_n^{\angle,*}(z)\psi_n^\angle(z)=2z^n\prod_{j=0}^{n-1}\zeta_j^2
	\end{equation}
	and hence
	\begin{gather}
		\Re{\overline{\psi_n^\angle(z)}\varphi_n^\angle(z)}=1,\\ \left|\psi_n(z)\right|\left|\varphi_n(z)\right|\geq 1.
	\end{gather}
\end{proposition}

\begin{proof}
	Applying (\ref{eq_ropuc_def}) and (\ref{eq_ropuc_reverse_def}) to the left hand side of (\ref{eq_ropuc_with_boundary_condition_z^n}) gives
	\begin{align*}
		\psi_n^{\angle,*}(z)\varphi_n^\angle(z)+\varphi_n^{\angle,*}(z)\psi_n^\angle(z)
		&=\left(\psi_n^*(z)\varphi_n(z)+\varphi_n^*(z)\psi_n(z)\right)\prod_{j=0}^{n-1}\zeta_j^2
		=2z^n\prod_{j=0}^{n-1}\zeta_j^2.
	\end{align*}
	Substituting (\ref{eq_ropuc_reverse_ropuc}) into (\ref{eq_ropuc_with_boundary_condition_z^n}) gives
	\begin{align*}
		2z^n\prod_{j=0}^{n-1}\zeta_j^2
		&=z^n\prod_{j=0}^{n-1}\zeta_j^2\left(\overline{\psi_n^\angle(z)}\varphi_n^\angle(z)+\overline{\overline{\psi_n^\angle(z)}\varphi_n^\angle(z)}\right)
		=\left(2z^n\prod_{j=0}^{n-1}\zeta_j^2\right)\Re{\overline{\psi_n^\angle(z)}\varphi_n^\angle(z)}
	\end{align*}
	so $\Re{\overline{\psi_n^\angle(z)}\varphi_n^\angle(z)}=1$ and hence $\left|\psi_n(z)\right|\left|\varphi_n(z)\right|\geq 1$.
\end{proof}

That is, while we generally have
$$\psi_n^{\angle,*}(z)\varphi_n^\angle(z)+\varphi_n^{\angle,*}(z)\psi_n^\angle(z)\neq 2z^n$$
unless $\zeta_j=1$ for all $j\in\{0,1,\cdots,n-1\}$, we still have $\Re{\overline{\psi_n^\angle(z)}\varphi_n^\angle(z)}=1$ and $\left|\psi_n(z)\right|\left|\varphi_n(z)\right|\geq 1$, which are also true for the non-rotated case.

One may also easily show that the mixed Christoffel-Darboux formula is also invariant under rotation of polynomials:
\begin{align} 
	\label{eq_ropuc_cd_mixed_opuc_cd_mixed}
	\sum_{k=0}^n \overline{\varphi_k^\angle(\xi)}\psi_k^\angle(z)
	&=\sum_{k=0}^n \overline{\varphi_k(\xi)}\psi_k(z)\\
	\label{eq_opuc_cd_mixed_ropuc}
	&=\dfrac{2-\overline{\varphi_{n+1}^*(\xi)}\psi_{n+1}^*(z)-\overline{\varphi_{n+1}(\xi)}\psi_{n+1}(z)}{1-\overline{\xi}z}\\
	\label{eq_ropuc_cd_mixed}
	&=\dfrac{2-\overline{\varphi_{n+1}^{\angle,*}(\xi)}\psi_{n+1}^{\angle,*}(z)-\overline{\varphi_{n+1}^\angle(\xi)}\psi_{n+1}^\angle(z)}{1-\overline{\xi}z}.
\end{align}
Hence we make an important observation: rotating the OPUCs causes the Szeg\H{o} recursion formulae and their CMV matrix representations to differ by certain complex conjugates, but does not change their Christoffel-Darboux formulae.

\begin{proposition}
	For any $n\in\mathbb{N}_0$, we have
	\begin{align}
		\label{eq_ropuc_second_kind_caratheodory}
		\psi_n^\angle(z)&=\int_{\partial\mathbb{D}}\left(\varphi_n^\angle(\xi)-\varphi_n^\angle(z)\right)\dfrac{\xi+z}{\xi-z}\,d\mu(\xi),\\
		\label{eq_ropuc_reverse_second_kind_caratheodory}
		\psi_n^{\angle,*}(z)&=\int_{\partial\mathbb{D}}\left(\varphi_n^{\angle,*}(z)-\varphi_n^{\angle,*}(\xi)\right)\dfrac{\xi+z}{\xi-z}\,d\mu(\xi).
	\end{align}
\end{proposition}

\begin{proof}
	Invoking (\ref{eq_ropuc_def}) and (\ref{eq_opuc_second_kind_caratheodory}) gives
	\begin{align*}
		\psi_n^\angle(z)&=\left(\prod_{j=0}^{n-1}\zeta_j\right)\psi_n(z)\\
		&=\int_{\partial\mathbb{D}}\left(\left(\prod_{j=0}^{n-1}\zeta_j\right)\varphi_n(\xi)-\left(\prod_{j=0}^{n-1}\zeta_j\right)\varphi_n(z)\right)\dfrac{\xi+z}{\xi-z}\,d\mu(\xi)\\
		&=\int_{\partial\mathbb{D}}\left(\varphi_n^\angle(\xi)-\varphi_n^\angle(z)\right)\dfrac{\xi+z}{\xi-z}\,d\mu(\xi)
	\end{align*}
	which proves (\ref{eq_ropuc_second_kind_caratheodory}). One can then prove (\ref{eq_ropuc_reverse_second_kind_caratheodory}) by invoking (\ref{eq_ropuc_reverse_def}) and (\ref{eq_opuc_reverse_second_kind_caratheodory}) following a similar calculation, or via (\ref{eq_ropuc_reverse_ropuc}) and (\ref{eq_ropuc_second_kind_caratheodory}).
\end{proof}

Again one can express (\ref{eq_ropuc_second_kind_caratheodory}) and (\ref{eq_ropuc_reverse_second_kind_caratheodory}) in terms of the Carath\'{e}odory function $F$ of $\mu$:
\begin{align}
	\label{eq_ropuc_second_kind_caratheodory_rewrite}
	\psi_n^\angle(z)+F(z)\varphi_n^\angle(z)&=\int_{\partial\mathbb{D}}\varphi_n^\angle(\xi)\dfrac{\xi+z}{\xi-z}\,d\mu(\xi),\\
	\label{eq_ropuc_reverse_second_kind_caratheodory_rewrite}
	-\psi_n^{\angle,*}(z)+F(z)\varphi_n^{\angle,*}(z)&=\int_{\partial\mathbb{D}}\varphi_n^{\angle,*}(\xi)\dfrac{\xi+z}{\xi-z}\,d\mu(\xi).
\end{align}

\begin{theorem}[Rotated Geronimo-Golinskii-Nevai] \label{thm_geronimo_golinskii_nevai_rot}
	Let $\mu$ be a probability measure on $\partial\mathbb{D}$, $z\in\mathbb{D}$ and $r\in\mathbb{C}$. Then the sequence
	\begin{equation}
		\Biggl\{\begin{bmatrix}
			\psi_n^\angle(z,\mu) \\ -\psi_n^{\angle,*}(z,\mu)
		\end{bmatrix}+r\begin{bmatrix}
			\varphi_n^\angle(z,\mu) \\ \varphi_n^{\angle,*}(z,\mu)
		\end{bmatrix}\Biggr\}_{n=0}^\infty\in\ell^2\left(\mathbb{N}_0,\mathbb{C}^2\right)
	\end{equation}
	if and only if $r=F(z)$, where $F$ is the Carath\'{e}odory function of $\mu$.
\end{theorem}

\begin{proof}
	Suppose that $r=F(z)$. By Theorem \ref{thm_geronimo_golinskii_nevai}, we have
	\[
		\sum_{n=0}^\infty \left|\psi_n(z,\mu)+F(z)\varphi_n(z,\mu)\right|^2<\infty,\qquad
		\sum_{n=0}^\infty \left|-\psi_n^*(z,\mu)+F(z)\varphi_n^*(z,\mu)\right|^2<\infty.
	\]
	Then for any $z\in\mathbb{D}$, we have
	\begin{align*}
		\sum_{n=0}^\infty \left|\psi_n^\angle(z,\mu)+F(z)\varphi_n^\angle(z,\mu)\right|^2
		&=\sum_{n=0}^\infty \left|\prod_{j=0}^{n-1}\zeta_j\psi_n(z,\mu)+F(z)\prod_{j=0}^{n-1}\zeta_j\varphi_n(z,\mu)\right|^2\\
		&=\sum_{n=0}^\infty \left|\psi_n(z,\mu)+F(z)\varphi_n(z,\mu)\right|^2<\infty.
	\end{align*}
	Similarly, we also have
	\begin{align*}
		\sum_{n=0}^\infty \left|-\psi_n^{\angle,*}(z,\mu)+F(z)\varphi_n^{\angle,*}(z,\mu)\right|^2
		&=\sum_{n=0}^\infty \left|-\prod_{j=0}^{n-1}\zeta_j\psi_n^*(z,\mu)+F(z)\prod_{j=0}^{n-1}\zeta_j\varphi_n^*(z,\mu)\right|^2\\
		&=\sum_{n=0}^\infty \left|-\psi_n^*(z,\mu)+F(z)\varphi_n^*(z,\mu)\right|^2<\infty.
	\end{align*}
	Hence we have proved that
	\begin{align*}
		\Biggl\{\begin{bmatrix}
			\psi_n^\angle(z,\mu) \\ -\psi_n^{\angle,*}(z,\mu)
		\end{bmatrix}+F(z)\begin{bmatrix}
			\varphi_n^\angle(z,\mu) \\ \varphi_n^{\angle,*}(z,\mu)
		\end{bmatrix}\Biggr\}_{n=0}^\infty
		\in\ell^2\left(\mathbb{N}_0,\mathbb{C}^2\right).
	\end{align*}
	Conversely, suppose that 
	$$\Biggl\{\begin{bmatrix}
		\psi_n^\angle(z,\mu) \\ -\psi_n^{\angle,*}(z,\mu)
	\end{bmatrix}+r\begin{bmatrix}
		\varphi_n^\angle(z,\mu) \\ \varphi_n^{\angle,*}(z,\mu)
	\end{bmatrix}\Biggr\}_{n=0}^\infty
	\in\ell^2\left(\mathbb{N}_0,\mathbb{C}^2\right).$$
	Then the linear combination
	\begin{gather*}
		r\Biggl\{\begin{bmatrix}
			\psi_n^\angle(z,\mu) \\ -\psi_n^{\angle,*}(z,\mu)
		\end{bmatrix}+F(z)\begin{bmatrix}
			\varphi_n^\angle(z,\mu) \\ \varphi_n^{\angle,*}(z,\mu)
		\end{bmatrix}\Biggr\}_{n=0}^\infty-F(z)\Biggl\{\begin{bmatrix}
			\psi_n^\angle(z,\mu) \\ -\psi_n^{\angle,*}(z,\mu)
		\end{bmatrix}+r\begin{bmatrix}
			\varphi_n^\angle(z,\mu) \\ \varphi_n^{\angle,*}(z,\mu)
		\end{bmatrix}\Biggr\}_{n=0}^\infty\\
		=(r-F(z))\Biggl\{\begin{bmatrix}
			\psi_n^\angle(z,\mu) \\ -\psi_n^{\angle,*}(z,\mu)
		\end{bmatrix}\Biggr\}_{n=0}^\infty\in\ell^2\left(\mathbb{N}_0,\mathbb{C}^2\right).
	\end{gather*}
	Suppose towards a contradiction that $r\neq F(z)$. Then $r-F(z)\neq 0$, so we must have
	$$\Biggl\{\begin{bmatrix}
		\psi_n^\angle(z,\mu) \\ -\psi_n^{\angle,*}(z,\mu)
	\end{bmatrix}\Biggr\}_{n=0}^\infty
	\in\ell^2\left(\mathbb{N}_0,\mathbb{C}^2\right).$$
	In particular, we must have $\left\{\psi_n^{\angle,*}(z,\mu)\right\}_{n=0}^\infty\in\ell^2(\mathbb{N}_0)$. On the other hand, taking (\ref{eq_ropuc_cd}) with $\xi=z$, replacing $(\varphi_j^\angle,\varphi_j^{\angle,*})$'s with $(\psi_j^\angle,\psi_j^{\angle,*})$'s and $n$ with $n-1$ gives
	\begin{align*}
		(1-|z|^2)\sum_{k=0}^{n-1}\left|\psi_k^\angle(z)\right|^2
		=\left|\psi_{n}^{\angle,*}(z)\right|^2-\left|\psi_{n}^\angle(z)\right|^2
	\end{align*}
	and hence
	\begin{align*}
		\left|\psi_{n}^{\angle,*}(z)\right|^2
		=(1-|z|^2)\sum_{k=0}^{n}\left|\psi_k^\angle(z)\right|^2
		\geq (1-|z|^2)\left|\psi_0^\angle(z)\right|^2 = 1-|z|^2.
	\end{align*}
	Since $\di\sum_{n=0}^\infty \left(1-|z|^2\right)$ diverges, by the comparison test, $\di\sum_{n=0}^\infty\left|\psi_{n}^{\angle,*}(z)\right|^2$ must also diverge, which contradicts $\left\{\psi_n^{\angle,*}(z,\mu)\right\}_{n=0}^\infty\in\ell^2(\mathbb{N}_0)$. Hence we must have $r=F(z)$.
\end{proof}

One may also use the invariance of the Christoffel-Darboux formulae under rotation of polynomials, and repeat the arguments used in \cite[pp. 231-234]{Simon2005a} to prove Theorem \ref{thm_geronimo_golinskii_nevai_rot}.

\subsection{Gesztesy-Zinchenko Formulation for Rotated CMV Operators}

We prove analogues of Theorems \ref{thm_cmv_extended_GZ_lem2.2}, \ref{thm_cmv_extended_GZ_lem2.3_right} and \ref{thm_cmv_extended_GZ_lem2.3_left}. In the process, we derive a rotated version of the Gesztesy-Zinchenko transfer matrix, and hence compute rotated versions of Tables \ref{table_GZ_cmv_even} and \ref{table_GZ_cmv_odd}.

\begin{theorem} \label{thm_rcmv_extended_GZ_lem2.2}
	Let $z\in\mathbb{C}-\{0\}$ and $f=\{f_n(z)\}_{n\in\mathbb{Z}}$, $g=\{g_n(z)\}_{n\in\mathbb{Z}}$ be sequences of complex functions. Define the unitary operator $\mathcal{U}^\angle$ on $\ell^2(\mathbb{Z})^2$ by
	\begin{equation}
		\mathcal{U}^\angle=\begin{bmatrix}
			\mathcal{E}^\angle & 0 \\ 0 & \tilde{\mathcal{E}}^\angle
		\end{bmatrix}
		=\begin{bmatrix}
			\tilde{\mathcal{L}}^\angle\tilde{\mathcal{M}}^\angle & 0 \\ 
			0 & \tilde{\mathcal{M}}^\angle\tilde{\mathcal{L}}^\angle
		\end{bmatrix}
		=\begin{bmatrix}
			0 & \tilde{\mathcal{L}}^\angle \\ \tilde{\mathcal{M}}^\angle & 0
		\end{bmatrix}^2.
	\end{equation}
	Then the following statements are equivalent:
	\begin{enumerate}
		\item $(\mathcal{E}^\angle f)(z)=zf(z)$ and $(\tilde{\mathcal{M}}^\angle f)(z)=zg(z)$;
		
		\item $(\tilde{\mathcal{E}}^\angle g)(z)=zg(z)$ and $(\tilde{\mathcal{L}}^\angle g)(z)=f(z)$;
		
		\item $(\tilde{\mathcal{M}}^\angle f)(z)=zg(z)$ and $(\tilde{\mathcal{L}}^\angle g)(z)=f(z)$;
		
		\item $\mathcal{U}^\angle\begin{bmatrix}
			f(z) \\ g(z)
		\end{bmatrix}=z\begin{bmatrix}
			f(z) \\ g(z)
		\end{bmatrix}$ and $(\tilde{\mathcal{M}}^\angle f)(z)=zg(z)$;
		
		\item $\mathcal{U}^\angle\begin{bmatrix}
			f(z) \\ g(z)
		\end{bmatrix}=z\begin{bmatrix}
			f(z) \\ g(z)
		\end{bmatrix}$ and $(\tilde{\mathcal{L}}^\angle g)(z)=f(z)$;
		
		\item for each $n\in\mathbb{Z}$,
		\begin{equation} \label{eq_GZ_transition_rcmv}
			\begin{bmatrix}
				f_{n+1}(z) \\ g_{n+1}(z)
			\end{bmatrix}=T_{n+1}^\angle(z)\begin{bmatrix}
				f_n(z) \\ g_n(z)
			\end{bmatrix},
		\end{equation}
		where
		\begin{equation} \label{eq_Tn_GZ_rCMV}
			T_n^\angle(z)=\begin{cases}
				\dfrac{1}{\rho_n}\begin{bmatrix}
					-\overline{\alpha_n} & z \\ z^{-1} & -\alpha_n
				\end{bmatrix} & n\text{ even};\\
				~\\
				\dfrac{1}{\rho_n}\begin{bmatrix}
					-\alpha_n & 1 \\ 1 & -\overline{\alpha_n}
				\end{bmatrix} & n\text{ odd}.
			\end{cases}
		\end{equation}
	\end{enumerate}
\end{theorem}

\begin{proof}
	Since $z\in\mathbb{C}-\{0\}$, we can rewrite $(\tilde{\mathcal{M}}^\angle f)(z)=zg(z)$ as
	\begin{equation}
		g(z)=\dfrac{1}{z}\left(\tilde{\mathcal{M}}^\angle f\right)(z).
	\end{equation}
	Suppose (1) holds. Using the fact that $\mathcal{E}^\angle=\tilde{\mathcal{L}}^\angle\tilde{\mathcal{M}}^\angle$, we obtain
	\allowdisplaybreaks
	\begin{align*}
		\left(\tilde{\mathcal{L}}^\angle g\right)(z)
		=\tilde{\mathcal{L}}^\angle\left(\dfrac{1}{z}\tilde{\mathcal{M}}^\angle f\right)(z)
		=\dfrac{1}{z}\left(\tilde{\mathcal{L}}^\angle\tilde{\mathcal{M}}^\angle f\right)(z)
		=\dfrac{1}{z}\tilde{\mathcal{E}}^\angle f(z)
		=\dfrac{1}{z}\cdot zf(z)
		=f(z)
	\end{align*}
	thus proving (3). Now suppose (3) holds. Again using the fact that $\mathcal{E}^\angle=\tilde{\mathcal{L}}^\angle\tilde{\mathcal{M}}^\angle$ gives
	\begin{align*}
		\left(\mathcal{E}^\angle f\right)(z)
		=\left(\tilde{\mathcal{L}}^\angle\tilde{\mathcal{M}}^\angle f\right)(z)
		=\tilde{\mathcal{L}}^\angle\left(\tilde{\mathcal{M}}^\angle f\right)(z)
		=\tilde{\mathcal{L}}^\angle\left(zg(z)\right)
		=z\left(\tilde{\mathcal{L}}^\angle g\right)(z)
		=zf(z)
	\end{align*}
	thus proving the equivalence of (1) and (3). On the other hand, using $\tilde{\mathcal{E}}^\angle=\tilde{\mathcal{M}}^\angle\tilde{\mathcal{L}}^\angle$ gives
	\begin{align*}
		\left(\tilde{\mathcal{E}}^\angle g\right)(z)
		=\left(\tilde{\mathcal{M}}^\angle\tilde{\mathcal{L}}^\angle g\right)(z)
		=\tilde{\mathcal{M}}^\angle\left(\tilde{\mathcal{L}}^\angle g\right)(z)
		=\tilde{\mathcal{M}}^\angle f(z)
		=zg(z)
	\end{align*}
	which proves (2). If we suppose (2) holds instead, again using $\tilde{\mathcal{E}}^\angle=\tilde{\mathcal{M}}^\angle\tilde{\mathcal{L}}^\angle$ gives
	\begin{align*}
		\left(\tilde{\mathcal{M}}^\angle f\right)(z)
		=\tilde{\mathcal{M}}^\angle\left(f(z)\right)
		=\tilde{\mathcal{M}}^\angle\left(\tilde{\mathcal{L}}^\angle g\right)(z)
		=\left(\tilde{\mathcal{E}}^\angle g\right)(z)
		=zg(z)
	\end{align*}
	thus proving the equivalence of (2) and (3). Now since (3) implies both (1) and (2), we have
	\begin{align} \label{eq_U_rotated_eigenvalue_prob}
		\mathcal{U}^\angle\begin{bmatrix}
			f(z) \\ g(z)
		\end{bmatrix}=
		\begin{bmatrix}
			\mathcal{E}^\angle & 0 \\ 0 & \tilde{\mathcal{E}}^\angle
		\end{bmatrix}\begin{bmatrix}
			f(z) \\ g(z)
		\end{bmatrix}
		=\begin{bmatrix}
			(\mathcal{E}^\angle f)(z) \\ (\tilde{\mathcal{E}}^\angle g)(z)
		\end{bmatrix}
		=\begin{bmatrix}
			zf(z) \\ zg(z)
		\end{bmatrix}
		=z\begin{bmatrix}
			f(z) \\ g(z)
		\end{bmatrix}
	\end{align}
	which proves that (3) implies (4) and (5). Conversely, (\ref{eq_U_rotated_eigenvalue_prob})
	holds if and only if $(\mathcal{E}^\angle f)(z)=zf(z)$ and $(\tilde{\mathcal{E}}^\angle g)(z)=zg(z)$. Hence if (4) is true, then (1) is true and hence (3) is true, thus proving the equivalence of (4) and (3). Similarly, we also have (5) implies (2) and hence implies (3), which proves the equivalence of (5) and (3). Now we are left to prove the equivalence of (3) and (6).
	
	When $n\in\mathbb{Z}$ is even, then the following statements are equivalent:
	\begin{gather}
		\begin{bmatrix}
			f_{n+1}(z) \\ g_{n+1}(z)
		\end{bmatrix}=T_{n+1}^\angle(z)\begin{bmatrix}
			f_n(z) \\ g_n(z)
		\end{bmatrix},\\
		\begin{cases}
			\rho_{n+1}f_{n+1}(z)=-\overline{\alpha_{n+1}}f_n(z)+zg_n(z);\\
			\rho_{n+1}g_{n+1}(z)=z^{-1}f_n(z)-\alpha_{n+1}g_n(z),
		\end{cases}\\
		\begin{cases}
			zg_n(z)=\overline{\alpha_{n+1}}f_n(z)+\rho_{n+1}f_{n+1}(z);\\
			z\rho_{n+1}g_{n+1}(z)=f_n(z)-\alpha_{n+1}zg_n(z),
		\end{cases}\\
		\begin{cases}
			zg_n(z)=\overline{\alpha_{n+1}}f_n(z)+\rho_{n+1}f_{n+1}(z);\\
			z\rho_{n+1}g_{n+1}(z)=|\rho_{n+1}|^2f_n(z)-\alpha_{n+1}\rho_{n+1}f_{n+1}(z),
		\end{cases}\\
		\begin{cases}
			zg_n(z)=\overline{\alpha_{n+1}}f_n(z)+\rho_{n+1}f_{n+1}(z);\\
			zg_{n+1}(z)=\overline{\rho_{n+1}}f_n(z)-\alpha_{n+1}f_{n+1}(z),
		\end{cases}\\
		z\begin{bmatrix}
			g_n(z) \\ g_{n+1}(z)
		\end{bmatrix}=\Theta_{n+1}^\angle\begin{bmatrix}
			f_n(z) \\ f_{n+1}(z)
		\end{bmatrix},\\
		\left(\tilde{\mathcal{M}}^\angle f\right)(z)=zg(z).
	\end{gather}
	When $n\in\mathbb{Z}$ is odd, then the following statements are equivalent:
	\begin{gather}
		\begin{bmatrix}
			f_{n+1}(z) \\ g_{n+1}(z)
		\end{bmatrix}=T_{n+1}^\angle(z)\begin{bmatrix}
			f_n(z) \\ g_n(z)
		\end{bmatrix},\\
		\begin{cases}
			\rho_{n+1}g_{n+1}(z)=f_n(z)-\overline{\alpha_{n+1}}g_n(z);\\
			\rho_{n+1}f_{n+1}(z)=-\alpha_{n+1}f_n(z)+g_n(z),
		\end{cases}\\
		\begin{cases}
			\rho_{n+1}g_{n+1}(z)=f_n(z)-\overline{\alpha_{n+1}}g_n(z);\\
			\rho_{n+1}f_{n+1}(z)=|\rho_{n+1}|^2g_n(z)-\alpha_{n+1}\rho_{n+1}g_{n+1}(z),
		\end{cases}\\
		\begin{cases}
			f_n(z)=\overline{\alpha_{n+1}}g_n(z)+\rho_{n+1}g_{n+1}(z);\\
			f_{n+1}(z)=\overline{\rho_{n+1}}g_n(z)-\alpha_{n+1}g_{n+1}(z),
		\end{cases}\\
		\begin{bmatrix}
			f_n(z) \\ f_{n+1}(z)
		\end{bmatrix}=\Theta_{n+1}^\angle\begin{bmatrix}
			g_n(z) \\ g_{n+1}(z)
		\end{bmatrix},\\
		\left(\tilde{\mathcal{L}}^\angle g\right)(z)=f(z).
	\end{gather}
	Combining both cases completes the proof of the equivalence of (3) and (6).
\end{proof}
In other words, we simply replace all operators with their rotated counterparts. A key observation is that since $\tilde{\mathcal{E}}\neq\left(\mathcal{E}^\angle\right)^T$ in general, we can not directly replace $\mathcal{E}^T$ by $\left(\mathcal{E}^\angle\right)^T$, so we lose transpose symmetry in (2). However by observing that $\tilde{\mathcal{E}}=\mathcal{E}^T=\tilde{\mathcal{M}}\tilde{\mathcal{L}}$, we can still draw the analogy with $\tilde{\mathcal{E}}^\angle=\tilde{\mathcal{M}}^\angle\tilde{\mathcal{L}}^\angle$, despite the loss of transpose symmetry. The rotated Gesztesy-Zinchenko transfer matrix $T_n^\angle$ also looks exactly like the non-rotated case given by (\ref{eq_GZ_transition_cmv}, except the real $\rho_n$'s are now replaced by complex $\rho_n$'s.

Again if there exists $K\in\mathbb{Z}$ such that $\alpha_K=e^{is}\in\partial\mathbb{D}$ for some $s\in(-\pi,\pi]$, then $\mathcal{E}^\angle$ is a direct sum
\begin{equation}
	\mathcal{E}^\angle=\mathcal{E}^\angle_{K^-}(s)\oplus\mathcal{E}^\angle_{K^+}(s),
\end{equation}
where $\mathcal{E}^\angle_{K^-}(s)$ and $\mathcal{E}^\angle_{K^+}(s)$ are operators on $\ell^2(\mathbb{Z}\cap\left(-\infty,K]\right)$ and $\ell^2\left(\mathbb{Z}\cap[K+1,\infty)\right)$, respectively. One may check that $\mathcal{E}^\angle_{K^+}$ looks almost like $\mathcal{C}^\angle$, whereas $\mathcal{E}^\angle_{K^-}$ also looks almost like $\mathcal{C}^\angle$ but extending in the opposite direction. Hence following a similar argument to Theorem \ref{thm_LM_rcmv}, one may show that $\mathcal{E}^\angle_{K^\pm}(s)$ also admit $\mathcal{L}\mathcal{M}$-decompositions
\begin{equation} \label{eq_rcmv_split}
	\mathcal{E}^\angle_{K^\pm}(s)=\mathcal{L}^\angle_{K^\pm}(s)\mathcal{M}^\angle_{K^\pm}(s).
\end{equation}
Again for the special case $s=\pi$, we denote $\mathcal{E}^\angle_{K^\pm}(\pi)$ by $\mathcal{E}^\angle_{K^\pm}$, and also define the operators $\mathcal{L}^\angle_{K^\pm}$ and $\mathcal{M}^\angle_{K^\pm}$ similarly. One may also simply verify that $\mathcal{E}^\angle_{K^+}(\pi)$ looks exactly like $\mathcal{C}^\angle$. 

If $K$ is even, then $\mathcal{L}^\angle_{K^+}$, $\mathcal{M}^\angle_{K^+}$, $\mathcal{L}^\angle_{K^-}$ and $\mathcal{M}^\angle_{K^-}$ are given by
\begin{align}
	\label{eq_rcmv_split_L+_even}
	\mathcal{L}^\angle_{K^+}&=\Theta^\angle_{K+1}\oplus\Theta^\angle_{K+3}\oplus\Theta^\angle_{K+5}\oplus\cdots,\\
	\label{eq_rcmv_split_M+_even}
	\mathcal{M}^\angle_{K^+}&=1\oplus\Theta^\angle_{K+2}\oplus\Theta^\angle_{K+4}\oplus\cdots, \\
	\label{eq_rcmv_split_L-_even}
	\mathcal{L}^\angle_{K^-}&=\cdots\oplus\Theta^\angle_{K-4}\oplus\Theta^\angle_{K-2}\oplus(-1), \\
	\label{eq_rcmv_split_M-_even}
	\mathcal{M}^\angle_{K^-}&=\cdots\oplus\Theta^\angle_{K-5}\oplus\Theta^\angle_{K-3}\oplus\Theta^\angle_{K-1},
\end{align}
where $\Theta^\angle_n$ is given by (\ref{eq_rTheta_extended}). If $K$ is odd, then they are given by
\begin{align}
	\label{eq_rcmv_split_L+_odd}
	\mathcal{L}^\angle_{K^+}&=1\oplus\Theta^\angle_{K+2}\oplus\Theta^\angle_{K+4}\oplus\cdots, \\
	\label{eq_rcmv_split_M+_odd}
	\mathcal{M}^\angle_{K^+}&=\Theta^\angle_{K+1}\oplus\Theta^\angle_{K+3}\oplus\Theta^\angle_{K+5}\oplus\cdots,\\
	\label{eq_rcmv_split_L-_odd}
	\mathcal{L}^\angle_{K^-}&=\cdots\oplus\Theta^\angle_{K-5}\oplus\Theta^\angle_{K-3}\oplus\Theta^\angle_{K-1},\\
	\label{eq_rcmv_split_M-_odd}
	\mathcal{M}^\angle_{K^-}&=\cdots\oplus\Theta^\angle_{K-4}\oplus\Theta^\angle_{K-2}\oplus(-1).
\end{align}

Now we are ready to prove analogues of Theorem \ref{thm_rcmv_extended_GZ_lem2.2} for the half-lattice operators $\mathcal{E}^\angle_{K^-}$ and $\mathcal{E}^\angle_{K^+}$:

\begin{theorem} \label{thm_rcmv_extended_GZ_lem2.3_right}
	Let $z\in\mathbb{C}-\{0\}$ and $f_{K^+}=\{f_n(z)\}_{n=K+1}^\infty$, $g_{K^+}=\{g_n(z)\}_{n=K+1}^\infty$ be sequences of complex functions. Define the operator $\mathcal{U}^\angle_{K^+}$ on $\ell^2(\mathbb{Z}\cap[K+1,\infty))^2$ by
	\begin{equation}
		\mathcal{U}^\angle_{K^+}=\begin{bmatrix}
			\mathcal{E}^\angle_{K^+} & 0 \\ 0 & \tilde{\mathcal{E}}^\angle_{K^+}
		\end{bmatrix}
		=\begin{bmatrix}
			\tilde{\mathcal{L}}^\angle_{K^+}\tilde{\mathcal{M}}^\angle_{K^+} & 0 \\ 
			0 & \tilde{\mathcal{M}}^\angle_{K^+}\tilde{\mathcal{L}}^\angle_{K^+}
		\end{bmatrix}
		=\begin{bmatrix}
			0 & \tilde{\mathcal{L}}^\angle_{K^+} \\ \tilde{\mathcal{M}}^\angle_{K^+} & 0
		\end{bmatrix}^2,
	\end{equation}
	where $K\in\mathbb{Z}$ is such that $\alpha_K\in\partial\mathbb{D}$. Then the following statements are equivalent:
	\begin{enumerate}
		\item $(\mathcal{E}^\angle_{K^+} f_{K^+})(z)=zf_{K^+}(z)$ and $(\tilde{\mathcal{M}}^\angle_{K^+} f_{K^+})(z)=zg_{K^+}(z)$;
		
		\item $(\tilde{\mathcal{E}}^\angle_{K^+} g_{K^+})(z)=zg_{K^+}(z)$ and $(\tilde{\mathcal{L}}^\angle_{K^+} g_{K^+})(z)=f_{K^+}(z)$;
		
		\item $(\tilde{\mathcal{M}}^\angle_{K^+} f_{K^+})(z)=zg_{K^+}(z)$ and $(\tilde{\mathcal{L}}^\angle_{K^+} g_{K^+})(z)=f_{K^+}(z)$;
		
		\item $\mathcal{U}^\angle_{K^+}\begin{bmatrix}
			f_{K^+}(z) \\ g_{K^+}(z)
		\end{bmatrix}=z\begin{bmatrix}
			f_{K^+}(z) \\ g_{K^+}(z)
		\end{bmatrix}$ and $(\tilde{\mathcal{M}}_{K^+}^\angle f_{K^+})(z)=zg_{K^+}(z)$;
		
		\item $\mathcal{U}^\angle_{K^+}\begin{bmatrix}
			f_{K^+}(z) \\ g_{K^+}(z)
		\end{bmatrix}=z\begin{bmatrix}
			f_{K^+}(z) \\ g_{K^+}(z)
		\end{bmatrix}$ and $(\tilde{\mathcal{L}}_{K^+}^\angle g_{K^+})(z)=f_{K^+}(z)$;
		
		\item for each $n\in\mathbb{Z}\cap[K+1,\infty)$,
		\begin{equation} \label{eq_GZ_transition_rcmv_+}
			\begin{bmatrix}
				f_{n+1}(z) \\ g_{n+1}(z)
			\end{bmatrix}=T_{n+1}^\angle(z)\begin{bmatrix}
				f_n(z) \\ g_n(z)
			\end{bmatrix},
		\end{equation}
		where $T_n^\angle(z)$ is given by (\ref{eq_Tn_GZ_rCMV}) and
		\[f_{K+1}(z)=\begin{cases}
			zg_{K+1}(z) & K \text{ even}; \\ g_{K+1}(z) & K \text{ odd}.
		\end{cases}\]
	\end{enumerate}
\end{theorem}

\begin{proof}
	The proof of the equivalence of (1), (2), (3), (4) and (5) follows a similar argument used in the proof of the equivalence of (1), (2), (3), (4) and (5) in Theorem \ref{thm_rcmv_extended_GZ_lem2.2}. The first part of the proof of the equivalence of (3) and (6) is also similar to the case in Theorem \ref{thm_rcmv_extended_GZ_lem2.2}. We are left to prove the representation of $f_K$. If $K$ is even, then $(\tilde{\mathcal{M}}^\angle_{K^+} f_{K^+})(z)=zg_{K^+}(z)$ and (\ref{eq_rcmv_split_M+_even}) imply that $$f_{K+1}(z)=zg_{K+1}(z).$$
	If $K$ is odd, then $(\tilde{\mathcal{L}}^\angle_{K^+} g_{K^+})(z)=f_{K^+}(z)$ and (\ref{eq_rcmv_split_L+_odd}) imply that \[g_{K+1}(z)=f_{K+1}(z). \qedhere\]
\end{proof}

Note that the proof of (6) implies (3) does not depend on the representation of $f_{K+1}$.

\begin{theorem} \label{thm_rcmv_extended_GZ_lem2.3_left}
	Let $z\in\mathbb{C}-\{0\}$ and $f_{K^-}=\{f_n(z)\}_{n=-\infty}^K$, $g_{K^-}=\{g_n(z)\}_{n=-\infty}^K$ be sequences of complex functions. Define the operator $\mathcal{U}^\angle_{K^-}$ on $\ell^2(\mathbb{Z}\cap(-\infty,K])^2$ by
	\begin{equation}
		\mathcal{U}^\angle_{K^-}=\begin{bmatrix}
			\mathcal{E}^\angle_{K^-} & 0 \\ 0 & \tilde{\mathcal{E}}^\angle_{K^-}
		\end{bmatrix}
		=\begin{bmatrix}
			\tilde{\mathcal{L}}^\angle_{K^-}\tilde{\mathcal{M}}^\angle_{K^-} & 0 \\ 
			0 & \tilde{\mathcal{M}}^\angle_{K^-}\tilde{\mathcal{L}}^\angle_{K^-}
		\end{bmatrix}
		=\begin{bmatrix}
			0 & \tilde{\mathcal{L}}^\angle_{K^-} \\ \tilde{\mathcal{M}}^\angle_{K^-} & 0
		\end{bmatrix}^2,
	\end{equation}
	where $K\in\mathbb{Z}$ is such that $\alpha_K\in\partial\mathbb{D}$. Then the following statements are equivalent:
	\begin{enumerate}
		\item $(\mathcal{E}^\angle_{K^-} f_{K^-})(z)=zf_{K^-}(z)$ and $(\tilde{\mathcal{M}}^\angle_{K^-} f_{K^-})(z)=zg_{K^-}(z)$;
		
		\item $(\tilde{\mathcal{E}}^\angle_{K^-} g_{K^-})(z)=zg_{K^-}(z)$ and $(\tilde{\mathcal{L}}^\angle_{K^+-} g_{K^-})(z)=f_{K^-}(z)$;
		
		\item $(\tilde{\mathcal{M}}^\angle_{K^-} f_{K^-})(z)=zg_{K^-}(z)$ and $(\tilde{\mathcal{L}}^\angle_{K^-} g_{K^-})(z)=f_{K^-}(z)$;
		
		\item $\mathcal{U}^\angle_{K^-}\begin{bmatrix}
			f_{K^-}(z) \\ g_{K^-}(z)
		\end{bmatrix}=z\begin{bmatrix}
			f_{K^-}(z) \\ g_{K^-}(z)
		\end{bmatrix}$ and $(\tilde{\mathcal{M}}_{K^-}^\angle f_{K^-})(z)=zg_{K^-}(z)$;
		
		\item $\mathcal{U}^\angle_{K^-}\begin{bmatrix}
			f_{K^-}(z) \\ g_{K^-}(z)
		\end{bmatrix}=z\begin{bmatrix}
			f_{K^-}(z) \\ g_{K^-}(z)
		\end{bmatrix}$ and $(\tilde{\mathcal{L}}_{K^-}^\angle g_{K^-})(z)=f_{K^-}(z)$;
		
		\item for each $n\in\mathbb{Z}\cap(-\infty,K]$,
		\begin{equation} \label{eq_GZ_transition_rcmv_-}
			\begin{bmatrix}
				f_{n-1}(z) \\ g_{n-1}(z)
			\end{bmatrix}=\left(T_{n}^\angle\right)^{-1}(z)\begin{bmatrix}
				f_n(z) \\ g_n(z)
			\end{bmatrix},
		\end{equation}
		where $T_n^\angle(z)$ is given by (\ref{eq_Tn_GZ_rCMV}) and
		\[f_K(z)=\begin{cases}
			-g_K(z) & K \text{ even}; \\ -zg_K(z) & K \text{ odd}.
		\end{cases}\]
	\end{enumerate}
\end{theorem}

\begin{proof}
	Similar to Theorem \ref{thm_rcmv_extended_GZ_lem2.3_right}, where we compute
	\[\det T_n^\angle(z)=\rho_n^{-1}\left(|\alpha_n|^2-1\right)=\rho_n^{-1}\left(-|\rho_n|^2\right)=-\overline{\rho_n}\]
	for all $n\in\mathbb{Z}$ and hence
	\begin{equation} \label{eq_Tn_inverse_GZ_rCMV}
		\left(T_n^\angle(z)\right)^{-1}=\begin{cases}
			\dfrac{1}{\overline{\rho_n}}\begin{bmatrix}
				\alpha_n & z \\ z^{-1} & \overline{\alpha_n}
			\end{bmatrix} & n\text{ even};\\
			~\\
			\dfrac{1}{\overline{\rho_n}}\begin{bmatrix}
				\overline{\alpha_n} & 1 \\ 1 & \alpha_n
			\end{bmatrix} & n\text{ odd},
		\end{cases}
	\end{equation}
	so again we only prove the representation of $f_K$. If $K$ is even, then $(\tilde{\mathcal{L}}^\angle_{K^-} g_{K^-})(z)=f_{K^-}(z)$ and (\ref{eq_rcmv_split_L-_even}) imply that \[g_K(z)=-f_K(z).\] 
	If $K$ is odd, then $(\tilde{\mathcal{M}}^\angle_{K^-} f_{K^-})(z)=zg_{K^-}(z)$ and (\ref{eq_rcmv_split_M-_odd}) imply that \[f_K(z)=-zg_K(z). \qedhere\]
\end{proof}
We notice that unlike the non-rotated case (\ref{eq_GZ_transition_cmv}), the inverse of the rotated Gesztesy-Zinchenko transfer matrix carry $\overline{\rho_n}$'s instead of $\rho_n$'s, which is again a direct consequence of $\overline{\rho_n}\neq\rho_n$ in $\mathbb{D}$.

\bigskip

Again we suppose $z\in\mathbb{C}-\{0\}$, $\left\{\begin{bmatrix}
	f_n^{+,\angle}(z) \\ g_n^{+,\angle}(z)
\end{bmatrix}\right\}_{n=K}^\infty$ and $\left\{\begin{bmatrix}
	p_n^{+,\angle}(z) \\ q_n^{+,\angle}(z)
\end{bmatrix}\right\}_{n=K}^\infty$ are linearly independent solutions of (\ref{eq_GZ_transition_rcmv_+}) satisfying the initial conditions
\begin{align}
	\label{eq_GZ_transition_init_fg_pq_+_even}
	\begin{bmatrix} 
		f_K^{+,\angle}(z) \\ g_K^{+,\angle}(z)
	\end{bmatrix}=\begin{bmatrix}
		z \\ 1
	\end{bmatrix},\qquad
	\begin{bmatrix}
		p_K^{+,\angle}(z) \\ q_K^{+,\angle}(z)
	\end{bmatrix}=\begin{bmatrix}
		z \\ -1
	\end{bmatrix}
\end{align}
when $K$ is even, and
\begin{align}
	\label{eq_GZ_transition_init_fg_pq_+_odd}
	\begin{bmatrix}
		f_K^{+,\angle}(z) \\ g_K^{+,\angle}(z)
	\end{bmatrix}=\begin{bmatrix}
		1 \\ 1
	\end{bmatrix},\qquad
	\begin{bmatrix}
		p_K^{+,\angle}(z) \\ q_K^{+,\angle}(z)
	\end{bmatrix}=\begin{bmatrix}
		-1 \\ 1
	\end{bmatrix}
\end{align}
when $K$ is odd. Similarly, also suppose $\left\{\begin{bmatrix}
	f_n^{-,\angle}(z) \\ g_n^{-,\angle}(z)
\end{bmatrix}\right\}_{n=-\infty}^K$ and $\left\{\begin{bmatrix}
	p_n^{-,\angle}(z) \\ q_n^{-,\angle}(z)
\end{bmatrix}\right\}_{n=-\infty}^K$ are linearly independent solutions of (\ref{eq_GZ_transition_rcmv_-}) satisfying the initial conditions
\begin{align}
	\label{eq_GZ_transition_init_fg_pq_-_even}
	\begin{bmatrix} 
		f_K^{-,\angle}(z) \\ g_K^{-,\angle}(z)
	\end{bmatrix}=\begin{bmatrix}
		1 \\ -1
	\end{bmatrix},\qquad
	\begin{bmatrix}
		p_K^{-,\angle}(z) \\ q_K^{-,\angle}(z)
	\end{bmatrix}=\begin{bmatrix}
		1 \\ 1
	\end{bmatrix}
\end{align}
when $K$ is even, and
\begin{align}
	\label{eq_GZ_transition_init_fg_pq_-_odd}
	\begin{bmatrix}
		f_K^{-,\angle}(z) \\ g_K^{-,\angle}(z)
	\end{bmatrix}=\begin{bmatrix}
		-z \\ 1
	\end{bmatrix},\qquad
	\begin{bmatrix}
		p_K^{-,\angle}(z) \\ q_K^{-,\angle}(z)
	\end{bmatrix}=\begin{bmatrix}
		z \\ 1
	\end{bmatrix}
\end{align}
when $K$ is odd. 

These choices of initial conditions are explained in a similar fashion as in the non-rotated case. In particular, the choices are identical to the non-rotated case because the relationship between $f_K$ and $g_K$ are identical in both cases.

\bigskip

Again one may use (\ref{eq_GZ_transition_rcmv}) to extend the sequences $\left\{\begin{bmatrix}
	f_n^{+,\angle}(z) \\ g_n^{+,\angle}(z)
\end{bmatrix}\right\}_{n=K}^\infty$ and $\left\{\begin{bmatrix}
	p_n^{+,\angle}(z) \\ q_n^{+,\angle}(z)
\end{bmatrix}\right\}_{n=K}^\infty$ to all $n<K$, and similarly extend $\left\{\begin{bmatrix}
	f_n^{-,\angle}(z) \\ g_n^{-,\angle}(z)
\end{bmatrix}\right\}_{n=-\infty}^K$ and $\left\{\begin{bmatrix}
	p_n^{-,\angle}(z) \\ q_n^{-,\angle}(z)
\end{bmatrix}\right\}_{n=-\infty}^K$ to all $n>K$. 

\bigskip

In particular, one can compute Tables \ref{table_GZ_rcmv_even} and \ref{table_GZ_rcmv_odd}, which looks exactly like Tables \ref{table_GZ_cmv_even} and \ref{table_GZ_cmv_odd}, except the obvious fact that $\rho_n$'s are now complex. The key observation in the rotated case is that traversing to a higher index involves division by $\rho_n$'s, whereas traversing to a lower index involves division by $\overline{\rho_n}$'s; whereas in the non-rotated case, traversing in either direction both involve division by real $\rho_n$'s.

While the initial conditions for the rotated and non-rotated cases are identical, applying the Gesztesy-Zinchenko transfer matrix would rotate the other terms in the sequences $\left\{\begin{bmatrix}
	f_n^{\pm,\angle}(z) \\ g_n^{\pm,\angle}(z)
\end{bmatrix}\right\}_n$ and $\left\{\begin{bmatrix}
	p_n^{\pm,\angle}(z) \\ q_n^{\pm,\angle}(z)
\end{bmatrix}\right\}_n$ in comparison to the non-rotated case, thus justifying our use of the superscript $^\angle$ in these sequences.  

\begin{table}[ht]
	\centering
	\begin{tabular}{||c||c|c|c||}
		\hhline{|=|=|=|=|} & & & \\ [-2ex]
		$n$ & $K-1$ & $K$ even & $K+1$\\ [0.5ex]
		\hhline{||=||=|=|=|} & & & \\ [-1ex]
		$\begin{bmatrix}
			f_n^{+,\angle}(z) \\ g_n^{+,\angle}(z)
		\end{bmatrix}$ & $\dfrac{1}{\overline{\rho_K}}\begin{bmatrix}
			\left(1+\alpha_K\right)z \\ 1+\overline{\alpha_K}
		\end{bmatrix}$ & $\begin{bmatrix}
			z \\ 1
		\end{bmatrix}$ & $\dfrac{1}{\rho_{K+1}}\begin{bmatrix}
			1-\alpha_{K+1}z \\ z-\overline{\alpha_{K+1}}
		\end{bmatrix}$ \\ [3ex]
		\hline & & & \\ [-1ex]
		$\begin{bmatrix}
			p_n^{+,\angle}(z) \\ q_n^{+,\angle}(z)
		\end{bmatrix}$ & $\dfrac{1}{\overline{\rho_K}}\begin{bmatrix}
			\left(-1+\alpha_K\right)z \\ 1-\overline{\alpha_K}
		\end{bmatrix}$ & $\begin{bmatrix}
			z \\ -1
		\end{bmatrix}$ & $\dfrac{1}{\rho_{K+1}}\begin{bmatrix}
			-1-\alpha_{K+1}z \\ z+\overline{\alpha_{K+1}}
		\end{bmatrix}$ \\ [3ex]
		\hline & & & \\ [-1ex]
		$\begin{bmatrix}
			f_n^{-,\angle}(z) \\ g_n^{-,\angle}(z)
		\end{bmatrix}$ & $\dfrac{1}{\overline{\rho_K}}\begin{bmatrix}
			\alpha_K-z \\ z^{-1}-\overline{\alpha_K}
		\end{bmatrix}$ & $\begin{bmatrix}
			1 \\ -1
		\end{bmatrix}$ & $\dfrac{1}{\rho_{K+1}}\begin{bmatrix}
			-1-\alpha_{K+1} \\ 1+\overline{\alpha_{K+1}}
		\end{bmatrix}$ \\ [3ex]
		\hline & & & \\ [-1ex]
		$\begin{bmatrix}
			p_n^{-,\angle}(z) \\ q_n^{-,\angle}(z)
		\end{bmatrix}$ & $\dfrac{1}{\overline{\rho_K}}\begin{bmatrix}
			\alpha_K+z \\ z^{-1}+\overline{\alpha_K}
		\end{bmatrix}$ & $\begin{bmatrix}
			1 \\ 1
		\end{bmatrix}$ & $\dfrac{1}{\rho_{K+1}}\begin{bmatrix}
			1-\alpha_{K+1} \\ 1-\overline{\alpha_{K+1}}
		\end{bmatrix}$ \\ [3ex]
		\hhline{||=||=|=|=|}
	\end{tabular}
	\caption{The sequences $\begin{bmatrix}
			f_n^{\pm,\angle}(z) \\ g_n^{\pm,\angle}(z)
		\end{bmatrix}$ and $\begin{bmatrix}
			p_n^{\pm,\angle}(z) \\ q_n^{\pm,\angle}(z)
		\end{bmatrix}$ for $n=K-1,K+1$, $K$ even} \label{table_GZ_rcmv_even}
\end{table}

\begin{table}[ht]
	\centering
	\begin{tabular}{||c||c|c|c||}
		\hhline{|=|=|=|=|} & & & \\ [-2ex]
		$n$ & $K-1$ & $K$ odd & $K+1$\\ [0.5ex]
		\hhline{||=||=|=|=|} & & & \\ [-1ex]
		$\begin{bmatrix}
			f_n^{+,\angle}(z) \\ g_n^{+,\angle}(z)
		\end{bmatrix}$ & $\dfrac{1}{\overline{\rho_K}}\begin{bmatrix}
			1+\overline{\alpha_K} \\ 1+\alpha_K
		\end{bmatrix}$ & $\begin{bmatrix}
			1 \\ 1
		\end{bmatrix}$ & $\dfrac{1}{\rho_{K+1}}\begin{bmatrix}
			z-\overline{\alpha_{K+1}} \\ z^{-1}-\alpha_{K+1}
		\end{bmatrix}$ \\ [3ex]
		\hline & & & \\ [-1ex]
		$\begin{bmatrix}
			p_n^{+,\angle}(z) \\ q_n^{+,\angle}(z)
		\end{bmatrix}$ & $\dfrac{1}{\overline{\rho_K}}\begin{bmatrix}
			1-\overline{\alpha_K} \\ -1+\alpha_K
		\end{bmatrix}$ & $\begin{bmatrix}
			-1 \\ 1
		\end{bmatrix}$ & $\dfrac{1}{\rho_{K+1}}\begin{bmatrix}
			z+\overline{\alpha_{K+1}} \\ -z^{-1}-\alpha_{K+1}
		\end{bmatrix}$ \\ [3ex]
		\hline & & & \\ [-1ex]
		$\begin{bmatrix}
			f_n^{-,\angle}(z) \\ g_n^{-,\angle}(z)
		\end{bmatrix}$ & $\dfrac{1}{\overline{\rho_K}}\begin{bmatrix}
			1-\overline{\alpha_K}z \\ \alpha_K-z
		\end{bmatrix}$ & $\begin{bmatrix}
			-z \\ 1
		\end{bmatrix}$ & $\dfrac{1}{\rho_{K+1}}\begin{bmatrix}
			\left(1+\overline{\alpha_{K+1}}\right)z \\ -1-\alpha_{K+1}
		\end{bmatrix}$ \\ [3ex]
		\hline & & & \\ [-1ex]
		$\begin{bmatrix}
			p_n^{-,\angle}(z) \\ q_n^{-,\angle}(z)
		\end{bmatrix}$ & $\dfrac{1}{\overline{\rho_K}}\begin{bmatrix}
			1+\overline{\alpha_K}z \\ z+\alpha_K
		\end{bmatrix}$ & $\begin{bmatrix}
			z \\ 1
		\end{bmatrix}$ & $\dfrac{1}{\rho_{K+1}}\begin{bmatrix}
			\left(1-\overline{\alpha_{K+1}}\right)z \\ 1-\alpha_{K+1}
		\end{bmatrix}$ \\ [3ex]
		\hhline{||=||=|=|=|}
	\end{tabular}
	\caption{The sequences $\begin{bmatrix}
			f_n^{\pm,\angle}(z) \\ g_n^{\pm,\angle}(z)
		\end{bmatrix}$ and $\begin{bmatrix}
			p_n^{\pm,\angle}(z) \\ q_n^{\pm,\angle}(z)
		\end{bmatrix}$ for $n=K-1,K+1$, $K$ odd} \label{table_GZ_rcmv_odd}
\end{table}

\vfill

\bibliography{rCMV_OPUC_Ryan.bib}

\end{document}